\documentclass[preprint]{elsarticle}

\usepackage{lineno}
\usepackage[bookmarksopen,bookmarksdepth=2]{hyperref}
\modulolinenumbers[5]

\hypersetup{pdftex,colorlinks=true,allcolors=blue}

\journal{Journal of \LaTeX\ Templates}

\newcounter{resume}

\usepackage{amsmath,amssymb,amsthm}
\usepackage{stmaryrd}
\usepackage{paralist}
\usepackage{enumitem}
\usepackage{xspace}
\usepackage{url}
\usepackage{tikz}
\usetikzlibrary{arrows,automata,matrix,scopes,positioning}

\usepackage{hypcap}

\usepackage{bbding}

\usepackage{multirow}

\usetikzlibrary{shapes}
\usetikzlibrary{arrows}

\newcommand{\ignore}[1]{}

\newcommand{\Q}{\mathbb{Q}}
\newcommand{\R}{\mathbb{R}}

\newcommand{\Lset}{\ensuremath{\mathcal L}\xspace}

\newcommand{\boundfun}{\ensuremath{\mathcal{B}}\xspace}

\newcommand{\B}{\boundfun}

\newcommand{\minimals}[1]{\ensuremath{\mathnormal{minimals}(#1)}}

\newcommand{\clos}[2][]{\ensuremath{#2^{\uparrow_{#1}}}\xspace}

\newcommand{\sat}[1]{\ensuremath{\llbracket #1 \rrbracket}\xspace}

\newcommand{\false} {\ensuremath{\mathbf{ff}}\xspace}
\newcommand{\true}  {\ensuremath{\mathbf{t\hspace{-.65mm}t}}\xspace}
\newcommand{\zeroFormula} {\ensuremath{\mathbf{0}}\xspace}

\newcommand{\zeroProcess} {\ensuremath{\mathtt{0}}\xspace}
\newcommand{\actionProcess}[1] {\ensuremath{#1}\xspace}
\newcommand{\aProcess} {\actionProcess{a}\xspace}

\newcommand{\exist}[1] {\ensuremath{\langle #1 \rangle}\xspace}
\newcommand{\univer}[1] {\ensuremath{[ #1 ]}\xspace}
\newcommand{\conf}[1] {\ensuremath{\llparenthesis #1 \rrparenthesis}}
\newcommand{\existuniver}[1] {\ensuremath{{\langle \! [} #1 {] \! \rangle}}}

\newcommand{\genericSemantics}[1]       {\ensuremath{#1}}
\newcommand{\genericSemanticsX}         {\genericSemantics{X}}

\newcommand{\simulation}	        {\ensuremath{\mathsf{S}}}
\newcommand{\completeSim}	        {\ensuremath{\mathsf{CS}}}
\newcommand{\readySim}		        {\ensuremath{\mathsf{RS}}}
\newcommand{\traceSim}		        {\ensuremath{\mathsf{TS}}}
\newcommand{\twoSim}		        {\ensuremath{\mathsf{2S}}}
\newcommand{\bisim}			{\ensuremath{\mathsf{BS}}}

\newcommand{\trace}                     {\ensuremath{\mathsf{T}}}
\newcommand{\completeTrace}             {\ensuremath{\mathsf{CT}}}
\newcommand{\failure}                   {\ensuremath{\mathsf{F}}}
\newcommand{\failureTrace}              {\ensuremath{\mathsf{FT}}}
\newcommand{\ready}                     {\ensuremath{\mathsf{R}}}
\newcommand{\readyTrace}                {\ensuremath{\mathsf{RT}}}
\newcommand{\impossibleFuture}          {\ensuremath{\mathsf{IF}}}
\newcommand{\impossibleFutureTrace}     {\ensuremath{\mathsf{IFT}}}
\newcommand{\possibleFuture}            {\ensuremath{\mathsf{PF}}}
\newcommand{\possibleFutureTrace}       {\ensuremath{\mathsf{PFT}}}
\newcommand{\impossibleSimulation}      {\ensuremath{\mathsf{I2}}}
\newcommand{\impossibleSimulationTrace} {\ensuremath{\mathsf{I2T}}}
\newcommand{\possibleSimulation}        {\ensuremath{\mathsf{P2}}}
\newcommand{\possibleSimulationTrace}   {\ensuremath{\mathsf{P2T}}}

\newcommand{\eqClass}[2]{\ensuremath{[#1]_{\scalebox{.5}{\bisim}}}}
\newcommand{\eqClassbisim}[1]{\eqClass{#1}{\equiv_{\bisim}}}
\newcommand{\eqClassPbisim}{\eqClassbisim{P}}
\newcommand{\eqClassQbisim}{\eqClassbisim{Q}}
\newcommand{\eqClasspbisim}{\eqClassbisim{p}}
\newcommand{\eqClassqbisim}{\eqClassbisim{q}}

\newcommand{\genericSemanticsFinite}[1]  {\ensuremath{\genericSemantics{#1}^{\mathit{fin}}}}
\newcommand{\genericSemanticsFiniteX}    {\genericSemanticsFinite{X}}
\newcommand{\genericSemanticsFiniteXbar} {\genericSemanticsFinite{\bar X}}

\newcommand{\traceFinite}                     {\ensuremath{\trace^{\mathit{fin}}}}
\newcommand{\completeTraceFinite}             {\ensuremath{\completeTrace^{\mathit{fin}}}}
\newcommand{\failureFinite}                   {\ensuremath{\failure^{\mathit{fin}}}}
\newcommand{\failureTraceFinite}              {\ensuremath{\failureTrace^{\mathit{fin}}}}
\newcommand{\readyFinite}                     {\ensuremath{\ready^{\mathit{fin}}}}
\newcommand{\readyTraceFinite}                {\ensuremath{\readyTrace^{\mathit{fin}}}}
\newcommand{\impossibleFutureFinite}          {\ensuremath{\impossibleFuture^{\mathit{fin}}}}
\newcommand{\impossibleFutureTraceFinite}     {\ensuremath{\impossibleFutureTrace^{\mathit{fin}}}}
\newcommand{\possibleFutureFinite}            {\ensuremath{\possibleFuture^{\mathit{fin}}}}
\newcommand{\possibleFutureTraceFinite}       {\ensuremath{\possibleFutureTrace^{\mathit{fin}}}}
\newcommand{\impossibleSimulationFinite}      {\ensuremath{\impossibleSimulation^{\mathit{fin}}}}
\newcommand{\impossibleSimulationTraceFinite} {\ensuremath{\impossibleSimulationTrace^{\mathit{fin}}}}
\newcommand{\possibleSimulationFinite}        {\ensuremath{\possibleSimulation^{\mathit{fin}}}}
\newcommand{\possibleSimulationTraceFinite}   {\ensuremath{\possibleSimulationTrace^{\mathit{fin}}}}

\newcommand{\maxsucc}{\ensuremath{\textit{max-succ}}}

\newcommand{\conformance}		{\ensuremath{\mathsf{C}}}

\newcommand{\depth} {\ensuremath{\mathit{depth}}\xspace}

\newcommand{\<}{\ensuremath{\lesssim}}
\newcommand{\strict}{\ensuremath{\lnsim}}

\newcommand{\pow}{{\cal P}}
\newcommand{\fun}{\,\rightarrow\,}
\newcommand{\rep}{\ensuremath{\mathit{rep}}}
\newcommand{\fin}{\ensuremath{\mathit{fin}}}

\newcommand{\powerset}[1]{\ensuremath{\mathcal{P}(#1)}}

\newcommand{\refusalarrow}[1]{\ensuremath{\stackrel{#1}{\text{\large\bf \----}\!\texttt{\footnotesize \#}}}}
\newcommand{\readyarrow}[1]{\ensuremath{\stackrel{#1}{\text{\large\bf \----}\!\bullet}}}
\newcommand{\impossiblearrow}[1]{\ensuremath{\stackrel{#1}{\text{\large\bf \----}}{\!\!\!\texttt{\footnotesize \#}\!\texttt{\footnotesize \#}\ }}}
\newcommand{\possiblearrow}[1]{\ensuremath{\stackrel{#1}{\text{\large\bf \----}}{\!\!\!\!\bullet\!\bullet\ }}}
\newcommand{\impossiblesimulationarrow}[1]{\ensuremath{\stackrel{#1}{\text{\large\bf \----}}{\!\!\!\texttt{\footnotesize \#}\!\texttt{\footnotesize \#}\!\texttt{\footnotesize \#}\ }}}
\newcommand{\possiblesimulationarrow}[1]{\ensuremath{\stackrel{#1}{\text{\large\bf \----}}{\!\!\!\!\bullet\!\!\bullet\!\!\bullet\ }}}

\newcommand{\traces}{\ensuremath{\mathit{traces}}\xspace}

\newcommand{\vgspectrumSet}{\ensuremath{\mathsf{vanG\text{-}spectrum}}\xspace}
\newcommand{\bspectrumSet}{\ensuremath{\mathsf{Btime\text{-}spectrum}}\xspace}
\newcommand{\lspectrumSet}{\ensuremath{\mathsf{Ltime\text{-}spectrum}}\xspace}

\newcommand{\specialitemOne}{\scriptsize\text{\FiveStarOpen}\xspace}
\newcommand{\specialitemTwo}{\ensuremath{\bowtie}\xspace}
\newcommand{\specialitemThree}{\ensuremath{\circ}\xspace}

\newcommand{\Sequivclass}[1]{\ensuremath{#1_{\downarrow_{\simulation}}}\xspace}

\newcommand{\barchiX}[1]{\ensuremath{\bar{\chi}_{_{#1}}}\xspace}
\newcommand{\barchiS}   {\barchiX{\simulation}}
\newcommand{\barchiCS}  {\barchiX{\completeSim}}
\newcommand{\barchiRS}  {\barchiX{\readySim}}
\newcommand{\barchiTS}  {\barchiX{\traceSim}}
\newcommand{\barchitwoS}{\barchiX{\twoSim}}
\newcommand{\barchiBS}  {\barchiX{\bisim}}

\newcommand{\chiX}[1]{\ensuremath{\chi_{_{#1}}}\xspace}
\newcommand{\chiS}   {\chiX{\simulation}}

\newcommand{\chiBS}  {\chiX{\bisim}}

\newcommand{\getFormula}[2]{\ensuremath{\mathit{formula}_{#1}(#2)}\xspace}
\newcommand{\getFormulatrace}[1]{\getFormula{\trace}{#1}}
\newcommand{\getFormulacompleteTrace}[1]{\getFormula{\completeTrace}{#1}}
\newcommand{\getFormulaready}[1]{\getFormula{\ready}{#1}}
\newcommand{\getFormulafailure}[1]{\getFormula{\failure}{#1}}
\newcommand{\getFormulafailureTrace}[1]{\getFormula{\failureTrace}{#1}}
\newcommand{\getFormulareadyTrace}[1]{\getFormula{\readyTrace}{#1}}
\newcommand{\getFormulapossibleFuture}[1]{\getFormula{\possibleFuture}{#1}}
\newcommand{\getFormulaimpossibleFuture}[1]{\getFormula{\impossibleFuture}{#1}}
\newcommand{\getFormulaimpossibleFutureTrace}[1]{\getFormula{\impossibleFutureTrace}{#1}}
\newcommand{\getFormulapossibleFutureTrace}[1]{\getFormula{\possibleFutureTrace}{#1}}
\newcommand{\getFormulaimpossibleSimulation}[1]{\getFormula{\impossibleSimulation}{#1}}
\newcommand{\getFormulapossibleSimulation}[1]{\getFormula{\possibleSimulation}{#1}}
\newcommand{\getFormulaimpossibleSimulationTrace}[1]{\getFormula{\impossibleSimulationTrace}{#1}}
\newcommand{\getFormulapossibleSimulationTrace}[1]{\getFormula{\possibleSimulationTrace}{#1}}

\newcommand{\simulatedby}[1]{\ensuremath{simulated\text{-}by(#1)}}

\newcommand{\smallset}[2]{\setWithRelation{#1}{\leq}{#2}}
\newcommand{\equalset}[2]{\setWithRelation{#1}{=}{#2}}
\newcommand{\setWithRelation}[3]{\ensuremath{#1_{#2 #3}}}

\newtheorem{proposition}{Proposition}
\newtheorem{definition}{Definition}
\newtheorem{lemma}{Lemma}
\newtheorem{theorem}{Theorem}
\newtheorem{corollary}{Corollary}
\newtheorem{remark}{Remark}

\begin{document}

\begin{frontmatter}

  \title{When are prime formulae characteristic?\tnoteref{mytitlenote1}\tnoteref{mytitlenote}}

  \tnotetext[mytitlenote1]{This paper is an extended version of~\cite{mfcs15}.}
  \tnotetext[mytitlenote]{Research supported by the Spanish projects \emph{TRACES} (TIN2015-67522-C3-3-R) and 
  \emph{RISCO} (TIN2015-71819-P), the project \emph{MATHADOR} (COGS 724.464) of the European Research Council, 
  the Spanish addition to \emph{MATHADOR} (TIN2016-81699-ERC), and the projects \emph{Nominal SOS} (project 
  nr.~141558-051)
 %
 %
    and \emph{Decidability and Expressiveness for Interval Temporal Logics}
    (project nr.~130802-051) of the Icelandic Research Fund.
    In addition, D. Della Monica acknowledges the financial support from a Marie
    Curie INdAM-COFUND-2012 Outgoing Fellowship (contract n.~UFMBAZ-2016/0000914, grant n.~PCCOFUND-GA-2012-600198).
  }
\author[aquila,iceland]{Luca Aceto}
\ead{luca.aceto@gssi.it, luca@ru.is}

\author[indam,madrid,napoli]{Dario Della Monica}
\ead{dario.dellamonica@unina.it}

\author[imdea]{Ignacio F\'abregas}
\ead{ignacio.fabregas@imdea.org}

\author[iceland]{Anna Ing\'olfsd\'ottir}
\ead{annai@ru.is}

\address[aquila]{Gran Sasso Science Institute (GSSI), L'Aquila, Italy}
\address[iceland]{ICE-TCS, School of Computer Science,
Reykjavik University, Reykjavik, Iceland}
\address[indam]{Istituto Nazionale di Alta Matematica ``F. Severi'' (INdAM), Italy}
\address[madrid]{Departamento de Sistemas Inform\'aticos y Computaci\'on,
Universidad Complutense de Madrid, Madrid, Spain}
\address[napoli]{Department of Electrical Engineering
and Information Technology, University of Naples ``Federico II'', Naples, Italy}
\address[imdea]{IMDEA Software Institute, Madrid, Spain}

\begin{abstract}
%
In the setting of the modal logic that characterizes modal refinement over modal transition systems, Boudol and Larsen showed that the formulae for which model checking can be reduced to preorder checking, that is, the characteristic formulae,  are exactly the consistent and prime ones. This paper presents general, sufficient conditions guaranteeing that characteristic formulae are exactly the consistent and prime ones. It is shown that the given conditions apply to various behavioural relations in the literature. In particular, characteristic formulae are exactly the prime and consistent ones for all the semantics in van Glabbeek's linear time-branching time spectrum. 
\end{abstract}

\begin{keyword}
%
  process semantics \sep logics\sep characteristic formulae\sep (bi)simulation
%
\end{keyword}

\end{frontmatter}


\section{Introduction} \label{sec:introduction}
Model checking and equivalence/preorder checking are the two main
approaches to the computer-aided verification of reactive
systems~\cite{AcetoEtAl07b,DBLP:books/daglib/0020348,ClarkeEtAl-book}. In model
checking, one typically describes the behaviour of a computing system
using a state-transition model, such as a labelled transition
system~\cite{Keller76}, and specifications of properties systems
should exhibit are expressed using some modal or temporal
logic. In this approach, system verification amounts to checking
whether a system is a model of the formulae describing a given
specification. When using equivalence/preorder checking instead,
systems and their specifications are both expressed in the same
state-machine-based formalism. In this approach, checking whether a
system correctly implements its specification amounts to verifying
whether the state machines describing them are related by some
suitable notion of behavioural equivalence/preorder.
(See~\cite{dFGPR13,VanGlabbeek01} for taxonomic studies of the plethora of behavioural relations that have
been considered in the field of concurrency theory.)

The study of the connections between the above-mentioned approaches to
system verification and, more generally, between behavioural and
logical approaches to defining the semantics of reactive systems has
been one of the classic topics of research in concurrency theory since
the work of Hennessy and Milner, who showed in~\cite{HennessyMilner85}
that bisimilarity~\cite{Milner89,Park81} coincides with the
equivalence induced by a multi-modal logic, which is now commonly
called Hennessy-Milner Logic, over labelled transition systems
satisfying a mild finiteness constraint. Similar modal
characterization results have been established for all the semantics
in van Glabbeek's linear-time/branching-time
spectrum~\cite{VanGlabbeek01}. Such results are truly fundamental in
concurrency theory and have found a variety of applications in
process theory---see, for instance, the
papers~\cite{AcetoH92,BloomFG04}. However, they do not provide a
``practically useful''
connection between model checking and equivalence/preorder checking;
indeed, using a modal characterization theorem to
prove that two labelled transition systems are equated by some notion
of behavioural equivalence would involve showing that the two systems
satisfy exactly the same (typically infinite) set of formulae in the modal logic that
characterizes the equivalence of interest.

A bridge between model checking and equivalence/preorder checking is
provided by the notion of {\em characteristic
formula}~\cite{GrafS86a,SteffenI94}. Intuitively,
a characteristic formula provides a complete logical characterization
of the behaviour of a process modulo some notion of behavioural
equivalence or preorder.
The complexity of the problem of deciding whether a formula is characteristic
for a process has been studied in~\cite{pls17,10.1007/978-3-319-72056-2_1} in
the setting of several modal logics and $\mu$-calculus.
At least for finite labelled transition
systems, characteristic formulae can be used to
reduce equivalence/preorder checking to model checking effectively~\cite{CleavelandS91}.
A natural question to ask is for what kinds of logical specifications
model checking can be reduced to establishing a behavioural relation
between an implementation and a labelled transition system that
suitably encodes the specification. To the best of our knowledge, this
question was first addressed by Boudol and Larsen, who showed
in~\cite{BoudolL1992} that, in the context of the modal logic that
characterizes modal refinement over modal transition systems, the
formulae that are ``graphically representable'' (that is, the ones
that are characteristic for some process) are exactly the
consistent and prime ones. (A formula is {\em prime} if whenever it
implies a disjunction of two formulae, it implies one of the
disjuncts.) A similar result is given in~\cite{AcetoEtAl11b} in the
setting of covariant-contravariant simulation. Moreover, each formula
in the logics considered in~\cite{AcetoEtAl11b,BoudolL1992} can be
``graphically represented'' by a (possibly empty) finite set of
processes.

To our mind, those are very pleasing results that show the very close
connection between logical and behavioural approaches to verification
in two specific settings. But, how general are they? Do similar
results hold for the plethora of other process semantics and their
modal characterizations studied in the literature? And, if so, are
there general sufficient conditions guaranteeing that characteristic
formulae are exactly the consistent and prime ones?
The purpose of this article is to provide answers to those
questions.
In particular, we aim to understand when the notions of characteristic and prime
formulae coincide (we refer to such a correspondence as \emph{characterization
  by primality}), thus providing a characterization of logically defined processes
by means of prime formulae.

We work in an abstract setting (described in
Section~\ref{sec:preliminary}), and, instead of investigating each behavioural 
semantics separately, we define the process semantics as the preorder induced by some logic, i.e. a process $p$ is smaller than a process $q$ if the set of
logical properties of $p$ is strictly included in that of $q$.
By investigating preorders defined in this way, we can identify common properties for all logically characterized preorders. 
%
%
%
It turns out that characteristic formulae are always consistent and prime (Theorem~\ref{theo:char_prime}).
Therefore our main task is to provide sufficiently general conditions guaranteeing
that consistent and prime formulae are characteristic formulae for
some process. 

In Section~\ref{sec:characterization_by_primality}, we introduce the notion of
{\em decomposable logic} and show that, for such logics, consistent and prime
formulae are characteristic for some process
(Theorem~\ref{theo:decomp}). (Intuitively, a logic is decomposable if, for each
consistent formula, the set of processes satisfying it includes the set of
processes satisfying a characteristic formula and the logic is sufficiently
expressive to witness this inclusion.) We then proceed to identify features that
make a logic decomposable, thus paving the way to showing the decomposability of
a number of logical formalisms (Section~\ref{subsec:logic_characterization}). In
particular, we provide two paths to decomposability, namely, a logic is
decomposable if
\begin{itemize}
\item the set of formulae satisfied by each process can be finitely
  characterized in a suitable technical sense (see
  Definition~\ref{def:characterized}) and some additional mild assumptions are met
  (Corollary~\ref{cor:decomposability_condition_for_spectrum}); or

\item each formula can be expressed as the union of the meaning of
  characteristic formulae (see Definition~\ref{def:representation}) and some
  additional assumptions are met
  (Proposition~\ref{prop:decomposability-for-finite-modal-covariant}).

\end{itemize}

In order to show the applicability of our general framework, we use it in
Sections~\ref{sec:finitely_many_processes}--\ref{sec:spectrum} to prove
characterization by primality (i.e., characteristic formulae are exactly the
consistent and prime ones) for a variety of logical characterizations of process
semantics. In particular, this applies to all the semantics in van Glabbeek's
linear time-branching time spectrum. In all these cases, there is a perfect
match between the behavioural and logical view of processes: not only do the
logics characterize processes up to the chosen notion of behavioural relation,
but processes represent all the consistent and prime formulae in the logics.

Finally, in Section~\ref{sec:conclusions} we provide an assessment of the work
done and outline future research directions.
In particular, in Section~\ref{sec:conformance}, we give evidence
(Proposition~\ref{prop:no_characterization_conformance}) that the first path to
decomposability we provide (based on
Corollary~\ref{cor:decomposability_condition_for_spectrum}) cannot be used to
show the decomposability of the logic characterizing conformance
simulation~\cite{FabregasEtAl10-logics}.
However, we are confident that the alternative path to decomposability (through
Proposition~\ref{prop:decomposability-for-finite-modal-covariant}) might serve
the purpose, and we plan to address the issue in the near future.

This paper is an extended version of~\cite{mfcs15}.
Compared with that preliminary study, the present work contains the following
new material.

\begin{itemize}

\item We provide a more elegant proof of characterization by primality for the
  case when there are only finitely many processes in
  Section~\ref{sec:finitely_many_processes}.
  Such a proof uses an approach that allows us to deal, in a uniform way, with
  modal refinement and covariant-contravariant simulation as well; it involves the
  notion of graphical representation of a formula (which is the inverse of
  characteristic formula of a process) and the one of finitely representability
  (Section~\ref{sec:finitely_many_processes}).

\item We prove characterization by primality for all of the logics
  characterizing the linear semantics in van Glabbeek's spectrum
  (Section~\ref{sec:linear-time-spectrum}).
  Such a result does not appear in~\cite{mfcs15} and is thus original.

\item We include, in Section~\ref{sec:conformance}
  (Proposition~\ref{prop:no_characterization_conformance}), a counterexample
  showing that conformance simulation cannot be dealt with using
  Corollary~\ref{cor:decomposability_condition_for_spectrum} ( the first path to
  decomposability we provide).

\item We give detailed proofs for all the results appearing in the conference
  version.

\end{itemize}



\section{Process semantics defined logically} \label{sec:preliminary}
We assume that $\Lset$ is a language interpreted over a non-empty set $P$, which we  refer to as a set of processes. Thus, $\Lset$ is equipped
with a semantic function $\sat{ \cdot }_\Lset: \Lset \fun\pow(P)$ (where $\pow(P)$ denotes the powerset of $P$), and we say that  $p\in P$ satisfies  $\phi\in\Lset$  whenever $p\in \sat{\phi}_\Lset$.
For all $p,q \in P$, we define the following notions:
\begin{compactitem}[\hspace{3mm}$\bullet$]
\item 
$\Lset(p) = \{ \phi \in \Lset \mid p \in\sat{\phi}_\Lset\}$: the set of formulae in $\Lset$ that $p$ satisfies;
we assume $\Lset(p) \neq \emptyset$, for each $p \in P$;
\item $\clos[\Lset]{p} = \{ p' \in P \mid \Lset(p)\subseteq\Lset(p') \}$: the \emph{upwards closure} of $p$ (with respect to $\Lset$);
\item $p$ and $q$ are \emph{logically equivalent} (with respect to $\Lset$), denoted by $p \equiv_\Lset q$, iff $\Lset(p)=\Lset(q)$;
\item  $p$ and $q$ are \emph{incomparable} (with respect to $\Lset$) iff neither $\Lset(p)\subseteq\Lset(q)$ nor $\Lset(q)\subseteq\Lset(p)$ holds. 
\end{compactitem}

We say that a formula $\phi \in \Lset$ is \emph{consistent} iff $\sat{\phi} _\Lset \neq \emptyset$.
Formulae $\phi, \psi \in \Lset$ are said to be \emph{logically equivalent}) (or simply \emph{equivalent}) iff
$\llbracket \phi \rrbracket_\Lset = \llbracket \psi \rrbracket_\Lset$.
When it is clear from the context, we omit the logic \Lset in the subscript
(and in the text).
For example, we write $\sat{\phi}$, $\equiv$, and $\clos{p}$
instead of $\sat{\phi}_\Lset$, $\equiv_\Lset$, and $\clos[\Lset]{p}$,
respectively.
We note that $\Lset(p)\subseteq\Lset(q)$ defines a preorder between processes, which we refer to as the \emph{logical preorder} characterized by $\Lset$. We say that a preorder over $P$ is \emph{logically characterized} or simply \emph{logical} if it is characterized by some logic $\Lset$.

For a subset $S\subseteq P$ we say that:
\begin{compactitem}[\hspace{3mm}$\bullet$]
\item $S$ is \emph{upwards closed } iff for all $p\in P$, 
if $p\in S$ then $\clos{p}\subseteq S$;
\item
 $p\in S$ is \emph{minimal} in $S$ iff 
for each $q \in S$, if $\Lset(q)\subseteq\Lset(p)$ then $\Lset(q)=\Lset(p)$;
\item
 $p\in S$ is a \emph{least element} in $S$ iff $\Lset(p)\subseteq\Lset(q)$
 for each $q \in S$.
\end{compactitem}
Clearly, if $p$ is a least element in a set $S$, then
$p$ is also minimal in $S$.
Notice that, if a set $S$ contains a least element, then
it is the unique minimal element in $S$, up to equivalence. 
\subsection{Characteristic and prime formulae, and graphical representation}
\label{subsec:characteristic_prime_preliminary}
We introduce here the crucial notion of
\emph{characteristic formula} for a  process~\cite{AcetoEtAl07b,GrafS86a,SteffenI94}
(along with the inverse notion of \emph{graphical representation}
of a formula)
and the one of \emph{prime formula}~\cite{AcetoEtAl11b,BoudolL1992}, in the setting of logical preorders over processes.
Our aim in this study  is to  investigate when these notions coincide,
thus providing a  characterization of logically defined processes
by means of prime formulae, which we will often refer to as \emph{characterization by primality}. To begin with, in this section we study such a connection between the above-mentioned notions in
a very general setting. As it turns out, for logically characterized preorders, the property of being
characteristic always implies primality (Theorem~\ref{theo:char_prime}).
The main focus of this paper becomes  therefore to investigate under what conditions a consistent and prime formula is characteristic for some process with respect to a logical preorder (Section~\ref{sec:characterization_by_primality}). 


\begin{definition}[Characteristic formula] \label{def:characteristic_formula}
  A formula $\phi\in\Lset$ is \emph{characteristic (within logic \Lset)} for $p
  \in P$ iff, for all $q\in P$, it holds that $q\in\sat{\phi}$ if and only if
  $\Lset(p)\subseteq\Lset(q).$
\end{definition}
It is worth observing that if $\phi$ is characteristic for some process $p$,
then $p \in \sat{\phi}$.

The following simple properties related to characteristic formulae will be useful in what follows.
\begin{proposition} \label{prop:general_properties_of_formulae}
The following properties hold for all $p,q\in P$ and $\phi\in\Lset$:
\begin{compactenum}[$(i)$]
\item \label{item:char1} $\phi$ is characteristic for $p$ if and only if $\sat{\phi}=\clos{p}$;
\item \label{item:char2} a characteristic formula for $p$, if it exists, is unique up to logical equivalence (and can therefore be referred to as $\chi(p)$);
\item\label{item:chareq}if the characteristic formulae for $p$ and $q$, namely $\chi(p)$ and $\chi(q)$, exist then  $\sat{\chi(p)}\subseteq\sat{\chi(q)}$ if and only if $\Lset(q)\subseteq\Lset(p)$.
\end{compactenum}
\end{proposition}
 \begin{proof}
Property
  (\ref{item:char1}) follows directly from the definition of the characteristic formula
  and the one of upward closure of $p$,
  while (\ref{item:char2}) and (\ref{item:chareq}) follow easily from (\ref{item:char1}). 
 \end{proof}
Next we state two useful properties.
\begin{proposition}\label{prop:upward_closure} 
The following properties hold for each $\phi\in\Lset$:
\begin{compactenum}[$(i)$]
\item \label{item:upclose}
$\sat{\phi}$ is upwards closed, and
\item \label{item:char3} for every $p \in P$, if $\chi(p)$ exists and $p\in\sat{\phi}\subseteq\sat{\chi(p)}$, then $\sat{\phi}=\sat{\chi(p)}$.
\end{compactenum}
\end{proposition}
 \begin{proof}
 We prove the two claims separately.
 \begin{compactenum}[$(i)$]
 \item
 Assume that $p \in \sat{\phi}$ and that $q\in\clos{p}$, or equivalently that $\Lset(p)\subseteq\Lset(q)$. As, by assumption,  
$\phi\in\Lset(p)$, we have that $\phi\in\Lset(q)$, and therefore that $q\in\sat{\phi}$ as we wanted to prove.
 \item
 This statement follows from the fact that $\sat{\phi}$ is upwards closed and, since
 $p \in \sat{\phi}$ by assumption, it includes $\clos{p}=\sat{\chi(p)}$.
 \qedhere
 \end{compactenum}
\end{proof}

The inverse of the notion of characteristic formula
is the one of \emph{graphical representation} of a formula.  For the sake of simplicity, in what follows we simply write ``represents'' instead of ``graphically represents''.

\begin{definition}[Graphical representation] \label{def:representation}
We say that $S\subseteq P$ \emph{represents} $\phi$ iff the elements in $S$ are pairwise incomparable and $\sat{\phi}=\bigcup_{p\in S} \clos{p}$. 
If $S$ is finite  we say that $\phi$ is \emph{finitely represented} by $S$ or simply finitely represented.
If $S=\{p\}$  we say that $p$ \emph{represents} $\phi$ (or that $\phi$ \emph{is represented by} $p$).
\end{definition}
Observe that there is possibly more than one graphical representation of a formula.
However, any two graphical representations of a formula are equivalent in the following sense,
so the representation is unique up to process equivalence.
\begin{proposition}
If $S,T\subseteq P$ represent $\phi$ then for all $p\in S$ there is some $q\in T$ such that $\Lset(p)=\Lset(q)$ and vice versa.
\end{proposition}
\begin{proof}
Assume that both $S$ and $T$ represent $\phi$ and therefore that
$\sat{\phi}=\bigcup_{p\in S} \clos{p} =\bigcup_{q\in T} \clos{q}$.
Assume that $p\in S$. Then $p\in \clos{p} \subseteq\bigcup_{q\in T} \clos{q}$. This implies that
$p\in \clos{q}$ for some $q \in T$, and therefore that $\Lset(q) \subseteq \Lset(p)$.
Symmetrically, we have that $\Lset(p') \subseteq \Lset(q)$, for some $p' \in S$.
Since elements of $S$ are pairwise incomparable,
$\Lset(p') \subseteq \Lset(q) \subseteq \Lset(p)$ implies $p=p'$,
and therefore we conclude that $\Lset(p)=\Lset(q)$.
Following the same argument it is possible to show that for each
$q \in T$ there exists some $p \in S$ such that $\Lset(p)=\Lset(q)$.
\end{proof}

We now define what it means for a formula to be prime.
\begin{definition}[Prime formula]
 We say that $\phi \in \Lset$ is \emph{prime} iff
 for each non-empty, finite subset of formulae $\Psi \subseteq \Lset$ it holds that
 $\sat{\phi} \subseteq \bigcup_{\psi \in \Psi} \sat{\psi}$ implies
 $\sat{\phi} \subseteq \sat{\psi}$ for some $\psi \in \Psi$.
\end{definition}
Observe that our definition is an equivalent, semantic version of the
one given in~\cite{BoudolL1992}. This serves our purpose to keep the
discussion as abstract as possible.
%
In this perspective, we want to abstract (at least at this point
of the discussion) from the syntactic details of the logical formalism,
while the classic definition of prime formula tacitly applies only to languages that
feature at least the Boolean connective $\vee$.

We provide here the first piece of our characterization by primality,
by showing that the property of being
characteristic implies primality without any extra assumption on the language \Lset or its interpretation. 
\begin{theorem} \label{theo:char_prime}
 Let $\phi \in \Lset$. If $\phi$ is a characteristic formula for some $p \in P$, then
 $\phi$ is prime and consistent.
\end{theorem}
 \begin{proof}
The formula $\phi$ is obviously consistent because $p \in \sat{\chi(p)} = \sat{\phi}$.
Towards proving that $\chi(p)$ is prime,
we assume that $\sat{\chi(p)} \subseteq \bigcup_{i\in I}\sat{\psi_i}$, where $I$ is finite and non-empty.
By our assumption, since $p \in \sat{\chi(p)}$, then for some $i\in I$,
$p \in \sat{\psi_i}$  holds.  As, by Proposition \ref{prop:upward_closure}(\ref{item:upclose}), 
\sat{\psi_i} is upwards closed, using Proposition~\ref{prop:general_properties_of_formulae}(\ref{item:char1}) we can conclude that $\sat{\chi(p)}=\clos{p}\subseteq\sat{\psi_i}$ as we wanted to prove.
\end{proof}

Notice that the converse is not true in general, that is, there exist formulae
that are consistent and prime but not characteristic.
To see this, let $P=\Q$, $\Lset=\R$ and $\sat{\phi}=\{p\in\Q\mid \phi\leq p\}$.
Clearly, all formulae are consistent.
Then, $\Lset(p)=\{\phi\in\R\mid \phi\leq p\}$ which implies that
$\Lset(p)\subseteq\Lset(q)$ iff $p\leq q$ iff $q\in\sat{p}$.
This means that, for each $p \in \Q$, $\phi=p$ is characteristic for $p$ and
therefore the characteristic formula is well-defined for all $p \in P$.
Furthermore $\sat{\phi}\cup\sat{\psi}=\{p \in \Q \mid \min \{ \phi,\psi \} \leq
p \}$ for all $\phi, \psi \in \Lset$, which implies that all formulae are prime.
On the other hand $\phi=\sqrt{2}\not\in\Q$ cannot be characteristic for any process as $\sat{\sqrt{2}}$ does not have a least element.


\section{Characterization by primality for logical preorders}\label{sec:characterization_by_primality}
In this section we introduce sufficient conditions under which the converse of Theorem~\ref{theo:char_prime}
is also true for logical preorders, that is, conditions guaranteeing that any consistent, prime formula
is characteristic. 

As a first step, we introduce the notion of \emph{decomposable} logic.
We show that if a logic is decomposable, then we have a logical characterization
of processes by primality. 
Some of the results involve the Boolean connectives $\wedge$ and $\vee$,
whose intended semantics is the standard one.

\begin{definition}[Decomposability]\label{def:decomposability}
We say that a formula $\phi\in\Lset$ is \emph{decomposable} iff
$\sat{\phi}=\sat{\chi(p)}\cup\sat{\psi_p}$ for some $p\in P$ and $\psi_p\in\Lset$, with $p\not\in\sat{\psi_p}$.
We say that $\Lset$ is \emph{decomposable} iff each consistent formula $\phi\in\Lset$ is
decomposable or characteristic for some $p \in P$.
\end{definition}

The attentive reader will have noticed that a characteristic formula is also
decomposable provided that ``false'' is expressible in the logic
$\mathcal{L}$. At this point of the paper we are not making any assumption on
the logic (not even that it includes false, negation and conjunction), so
defining decomposability as we do yields a more abstract and general framework.

\begin{proposition}\label{prop:decomp}
For all $\phi,\psi\in\Lset$ and $p\in P$, where $p\not\in\sat{\psi}$ and $\phi$ is prime,
if $\sat{\phi}=\sat{\chi(p)}\cup\sat{\psi}$, then $\sat{\phi}=\sat{\chi(p)}$.
\end{proposition}
 \begin{proof}
 Assume that $\phi$ is prime, that $\sat{\phi}=\sat{\chi(p)}\cup\sat{\psi}$ and that $p\not\in\sat{\psi}$. As $p\in\sat{\phi}$, it is clear that $\sat{\phi}\not\subseteq\sat{\psi}$ and, as $\phi$ is prime, we have that $\sat{\phi}\subseteq\sat{\chi(p)}$ which, thanks to Proposition~\ref{prop:upward_closure}(\ref{item:char3}), implies that $\sat{\phi}=\sat{\chi(p)}$.
 \end{proof}

The following theorem allows us to reduce the problem of relating the notions
of prime and characteristic formulae in a given logic to the problem of establishing the
decomposability property for that logic.
This provides us with a  very general setting towards characterization by primality.

\begin{theorem}\label{theo:decomp}
If $\Lset$ is decomposable then every formula in $\Lset$ that is consistent and prime is also characteristic for some $p\in P$.
\end{theorem}
 \begin{proof}
 The claim follows directly from Proposition \ref{prop:decomp}.
 \end{proof}

\subsection{Paths to decomposability} \label{subsec:logic_characterization}

The aim of this section is to identify features that make a logic decomposable,
thus paving the way towards showing the decomposability of a number of logical
formalisms in the next sections.
%
%
%
First, we observe that if a characteristic formula $\chi(p)$ exists for every $p \in P$,
then what we are left to do is to define, for each $\phi \in \Lset$, a formula $\psi_p$, for some $p \in P$, with the properties mentioned in Definition~\ref{def:decomposability}, as captured by the following proposition.
%
\begin{proposition} \label{prop:requirements_for_bisimulation}
 Let \Lset be a logic such that
 \begin{inparaenum}[$(i)$]
  \item \label{item:req_all_p_char1} $\chi(p)$ exists for each $p \in P$, and
  \item \label{item:req_neg_chi} for each consistent formula $\phi$
  there exist $p \in \sat{\phi}$ and $\psi_p \in \Lset$
  such that $p \notin\sat{\psi_p}$ and $\sat{\phi} \setminus \sat{\chi(p)} \subseteq \sat{\psi_p} \subseteq \sat{\phi}$.
 \end{inparaenum}
 Then \Lset is decomposable.
\end{proposition}

\begin{proof}
  Let $\phi \in \Lset$ be consistent and let us consider the formula $\psi_p$
  (with $p \in \sat{\phi}$) whose existence is guaranteed by proviso
  (\ref{item:req_neg_chi}) of the proposition.
  By proviso (\ref{item:req_all_p_char1}), the formula $\chi(p)$ exists in \Lset
  as well and, since $\phi$ is upwards closed and $p \in \sat{\phi}$, we have
  $\sat{\chi(p)} = \clos{p}\subseteq \sat{\phi}$.
  It immediately follows that $\sat{\chi(p)}\cup\sat{\psi_p} \subseteq
  \sat{\phi}$ as, by hypothesis, $\sat{\psi_p} \subseteq \sat{\phi}$ holds as
  well.
  Moreover, from $\sat{\phi} \setminus \sat{\chi(p)} \subseteq \sat{\psi_p}$ and
  $\sat{\chi(p)} \subseteq \sat{\phi}$, we have that $\sat{\phi} \subseteq
  \sat{\psi_p} \cup \sat{\chi(p)}$.
  Hence, $\sat{\phi} = \sat{\psi_p} \cup \sat{\chi(p)}$, and
%
%
   we are done.
\end{proof}

Clearly, when dealing with formalisms featuring at least the Boolean operators
$\neg$ and $\wedge$, as it is the case with the logic for the bisimulation
semantics (Section~\ref{sec:branch-time-spectr}), such a formula $\psi_p$ is
easily defined as $\neg \chi(p) \wedge \phi$.
This is stated in the following corollary.
\begin{corollary}\label{cor:negation}
Let \Lset be a logic featuring at least the Boolean connective $\wedge$
such that $\chi(p)$ exists for each $p \in P$, and there exists some formula $\bar{\chi}(p)\in\Lset$ where $\sat{\bar{\chi}(p)}= P \setminus \sat{\chi(p)}$. Then \Lset is decomposable.
\end{corollary}

The situation is more complicated for the other logics for the
semantics in the branching time-linear time spectrum (which we consider in
Section~\ref{sec:spectrum}), as negation is in general not expressible in
these logics, not even for characteristic formulae.
Therefore, instead we will prove a slightly stronger statement than the one in Corollary \ref{cor:negation} by identifying a weaker condition than the existence of a negation of the characteristic formulae (that we assume to exist) that also leads to decomposability of the logic.
This is described in the following proposition.
\begin{proposition} \label{prop:barchi}
 Let $\phi \in \Lset$, $p$ be a minimal element in $\sat{\phi}$
 such that $\chi(p)$ exists in \Lset, and let $\bar\chi(p)$ be a formula in $\Lset$ such that
 $\{ q \in P \mid \Lset(q) \not\subseteq \Lset(p) \} \subseteq \sat{\bar\chi(p)}$.
 Then, $\sat{\phi} \setminus \sat{\chi(p)} \subseteq \sat{\bar\chi(p)}$ holds.
\end{proposition}
 \begin{proof}
   Let us consider an element $q'$ such that $q' \in \sat{\phi}$ and $q' \notin \sat{\chi(p)}$.
   By the latter assumption, we have that $\Lset(p) \not\subseteq \Lset(q')$.
  Moreover, since $p$ is minimal in $\sat{\phi}$ and $q' \in \sat{\phi}$, it also holds that
  $\Lset(q') \not\subset \Lset(p)$. Thus, $\Lset(q') \not\subseteq \Lset(p)$, which implies
  $q' \in \{ q \in P \mid \Lset(q) \not\subseteq \Lset(p) \} \subseteq \sat{\bar\chi(p)}$,
  as we wanted to prove.
 \end{proof}

 In the next proposition, we build on the above result, and establish some
 conditions, which are met by the logics characterizing the semantics in van Glabbeek's linear time-branching time
 spectrum (see Section~\ref{sec:spectrum}), and which immediately lead to
 decomposability.
\begin{proposition} \label{prop:requirements}
 Let \Lset be a logic that features at least the Boolean connective $\wedge$ and such that:
 \begin{compactenum}[$(i)$]
  \item \label{item:req_all_p_char} $\chi(p)$ exists for each $p \in P$,
  \item \label{item:req_minimal} for each consistent $\phi$, the set $\sat{\phi}$ has a minimal element, and
  \item \label{item:req_chibar} for each $p\in P$, there exists in \Lset a formula $\bar\chi(p)$ such that either
    \begin{itemize}
    \item $\sat{\bar{\chi}(p)}= P \setminus \sat{\chi(p)}$ or
    \item $p \notin \sat{\bar\chi(p)}$ and $\{ q \in P \mid \Lset(q)
      \not\subseteq \Lset(p) \} \subseteq \sat{\bar\chi(p)}$.
    \end{itemize}
 \end{compactenum}
 Then, \Lset is decomposable.
\end{proposition}
 \begin{proof}
 Let $\phi \in \Lset$ be consistent. We choose a minimal element $p$ in $\sat{\phi}$,
 which exists by proviso (\ref{item:req_minimal}) of the proposition,
 and we define $\psi=\bar \chi(p)\wedge \phi$. Clearly, $p \notin \sat{\bar\chi(p) \wedge \phi}$
 since $p \notin \sat{\bar\chi(p)}$.
 We show that $\sat{\phi}=\sat{\chi(p)}\cup\sat{\psi}$. The inclusion from right to left
 immediately follows from the definition of $\psi$ and from the fact that $\sat{\phi}$ is upward closed.
 The converse inclusion is straightforward: if $\sat{\bar\chi(p)} = P \setminus
 \sat{\chi(p)}$, then it follows from the obvious observation that $\sat{\phi}
 \setminus \sat{\chi(p)} = \sat{\phi} \cap (P \setminus \sat{\chi(p)}) =
 \sat{\psi}$; otherwise, it immediately follows from
 Proposition~\ref{prop:barchi}.
 \end{proof}
 
In order to apply the above result to prove decomposability for a logic
\Lset, we now develop a general framework ensuring
conditions~(\ref{item:req_all_p_char}) and~(\ref{item:req_minimal})
in Proposition~\ref{prop:requirements}.
To this end, we exhibit a finite characterization of the (possibly) infinite
set $\Lset(p)$ of true facts associated with every $p \in P$.
(In order to ensure condition~(\ref{item:req_chibar}) of the proposition,
we will explicitly construct the formula $\bar\chi(p)$ in each of the languages considered in Section~\ref{sec:spectrum}.)

\begin{definition}[Characterization] \label{def:characterized}
 We say that the logic $\Lset$ is \emph{characterized} by a function $\B: P\rightarrow \pow(\Lset)$ iff for each $p\in P$ we have $\emptyset \subset \B(p)\subseteq\Lset(p)$ 
and for each  $\phi\in\Lset(p)$ it holds that
$\bigcap_{\psi\in\B(p)}\sat{\psi}\subseteq\sat{\phi}$.
%
We say that $\Lset$ is  \emph{finitely characterized} by $\B$ iff
$\Lset$ is characterized by a function $\B$ such that $\B(p)$ is finite for each $p\in P$.
Finally, we say that $\B$ is monotonic iff $\Lset(p)\subseteq\Lset(q)$ implies $\B(p)\subseteq\B(q)$ for all $p,q\in P$.
\end{definition}

In what follows, we show that if a logic \Lset features at least the Boolean
connective $\wedge$ and it is finitely characterized by $\B$, for some monotonic $\B$,
then it fulfils conditions~(\ref{item:req_all_p_char}) and~(\ref{item:req_minimal})
in Proposition~\ref{prop:requirements}
(see Proposition~\ref{prop:characteristic} and
Corollary~\ref{cor:inf_des_sequence} to follow).

\begin{proposition}\label{prop:characterized1}
If $\Lset$ is  characterized by $\B$, then for each $p,q \in P$,
$\B(p) \subseteq \Lset(q)$ implies
$\Lset(p) \subseteq \Lset(q)$.
\end{proposition}
 \begin{proof}
 Assume that $\B(p)\subseteq\Lset(q)$ and that $\phi\in\Lset(p)$.  As $\Lset$ is characterized by $\B$, $\bigcap_{\psi\in\B(p)}\sat{\psi}\subseteq\sat{\phi}$
 holds. Since $\B(p) \subseteq \Lset(q)$, we have $q \in \bigcap_{\psi\in\B(p)}\sat{\psi}\subseteq\sat{\phi}$, which means that $\phi\in\Lset(q)$, as we wanted to prove.
 \end{proof}
\begin{proposition}\label{prop:characterized2}
If $\Lset$ is  characterized by $\B$, then for each $p,q \in P$,
$\B(p) \subseteq \B(q)$ implies
$\Lset(p) \subseteq \Lset(q)$.
\end{proposition}
 \begin{proof}
 Assume that $\B(p)\subseteq\B(q)$. As $\B(q)\subseteq\Lset(q)$ the result follows by Proposition \ref{prop:characterized1}.
 \end{proof}

Every finitely
characterized logic featuring at least the Boolean connective
$\land$ enjoys
the pleasing property that every $p \in P$
admits a characteristic formula in $\Lset$.
\begin{proposition} \label{prop:characteristic}
 Let $\Lset$ be a logic  that features at least the Boolean connective $\wedge$ and that is finitely characterized by $\B$. Then, each $p \in P$ has a characteristic formula in $\Lset$ given by  $\chi(p) = \bigwedge_{\phi \in \B(p)} \phi$.
\end{proposition}
 \begin{proof}
 Assume that $p\in P$.
  First, we observe that $\bigwedge_{\phi \in \B(p)} \phi$
  is a well-formed formula, as $\B(p)$ is finite. 
  Next, we prove that $\sat{\bigwedge_{\phi \in \B(p)} \phi} = \clos{p}$ holds,
  from which the thesis immediately follows by Proposition~\ref{prop:general_properties_of_formulae}(\ref{item:char1}).
 
 Since $\B(p) \subseteq \Lset(p)$, it follows that $p\in\sat{\bigwedge_{\phi \in \B(p)} \phi}$. As $\sat{\bigwedge_{\phi \in \B(p)} \phi}$ is upwards closed, we have that $\clos{p}\subseteq\sat{\bigwedge_{\phi \in \B(p)} \phi}$.
 
  Towards proving the converse inclusion,
  let us assume that
  $q \in \sat{\bigwedge_{\phi \in \B(p)} \phi}$.
  Since $\sat{\bigwedge_{\phi \in \B(p)} \phi} =
  \bigcap_{\phi \in \B(p)} \sat{\phi}$, we have
  that $q \in\sat{\phi}$, for each $\phi \in \B(p)$.
  Thus, $\B(p) \subseteq \Lset(q)$ holds. From Proposition \ref{prop:characterized1},
  it follows that $\Lset(p) \subseteq \Lset(q)$ and, by the definition of
  upwards closure, we conclude that $q \in \clos{p}$.
 \end{proof}
 
\begin{proposition}\label{prop:inf_des_sequence}
 Let $\Lset$ be a logic that is finitely characterized by $\B$, for some monotonic $\B$.
 Then, for each $\phi \in \Lset$ and $q \in \sat{\phi}$, there exists some $p
 \in P$ such that $\Lset(p) \subseteq \Lset(q)$ and $p$ is minimal in
 $\sat{\phi}$.
\end{proposition}
\begin{proof}
Towards a contradiction, let us assume that there exist $\phi \in \Lset$ and $q\in\sat{\phi}$
such that, for each $p \in \sat{\phi}$, with $\Lset(p)\subseteq\Lset(q)$,
$p$ is not minimal in $\sat{\phi}$.
Notice that $q$ is not minimal in $\sat{\phi}$ itself.

Then, there exists an infinite sequence $q_0, q_1, \ldots \in \sat{\phi}$, with $q = q_0$, such that $\Lset(q_{i+1}) \subsetneq \Lset(q_{i})$ for each $i \geq 0$. As $\B$ is monotonic, we have that $\B(q_{i+1})\subseteq\B(q_{i})$ for each $i \geq 0$. Since $\B(q)$ is finite, there exists some $k\geq 0$ such that $\B(q_k)=\B(q_{k+\ell})$ for each $\ell > 0$.
By Proposition~\ref{prop:characterized2}, $\Lset(q_k) = \Lset(q_{k+\ell})$ holds for each $\ell > 0$.
This contradicts the fact that $\Lset(q_{i+1}) \subsetneq \Lset(q_{i})$ for each $i \geq 0$,
which means that for each $q \in \sat{\phi}$ there exists some $p \in \sat{\phi}$, with $\Lset(p)\subseteq\Lset(q)$,
such that $p$ is minimal in $\sat{\phi}$.
\end{proof}

\begin{corollary}\label{cor:inf_des_sequence}
 Let $\Lset$ be a logic that is finitely characterized by $\B$, for some monotonic $\B$.
 Then, for each consistent formula $\phi\in\Lset$, $\sat{\phi}$ has a minimal element.
\end{corollary}
\begin{proof}
The thesis immediately follows from Proposition~\ref{prop:inf_des_sequence}.
\end{proof}

The next result will be useful in the next section.

\begin{corollary}\label{cor:rep}
Let $\Lset$ be a logic that features at least the Boolean connective $\wedge$ and that is finitely characterized by $\B$, for some monotonic $\B$.
Then, each formula $\phi\in\Lset$ is represented (uniquely up to equivalence) by a set $\rep(\phi)$ (in the sense of Definition~\ref{def:representation}).
\end{corollary}
 \begin{proof}
 Let $\minimals{\phi} = \{ p \in P \mid p \text{ is minimal in } \sat{\phi} \}$.
 Notice that the elements of
 \minimals{\phi} are not necessarily pairwise incomparable (there can exist
 $p,q \in \minimals{\phi}$ such that $p \neq q$ and $p \equiv q$).
 
 Now, we define $\rep(\phi)$ as any subset of \minimals{\phi} such that
 \begin{inparaenum}[\it (i)]
  \item elements in $\rep(\phi)$ are pairwise incomparable and
  \item for each $p \in \minimals{\phi}$ there exists some $q \in \rep(\phi)$
  with $p \equiv q$.
 \end{inparaenum}
 Intuitively, $\rep(\phi)$ contains a representative element for each
 equivalence class of minimal elements in $\sat{\phi}$ modulo $\equiv$.
 First of all, we observe that, by Proposition~\ref{prop:characteristic},
 $\chi(p)$ exists for each $p \in P$.
 We show that $\sat{\phi} = \bigcup_{p\in\rep(\phi)}\clos{p}$, from which the
 thesis immediately follows.
 
 From Proposition~\ref{prop:inf_des_sequence}, we have that each $q \in \sat{\phi}$ belongs to $\clos{p}$,
 for some $p \in \rep(\phi)$, and thus $\sat{\phi} \subseteq \bigcup_{p\in\rep(\phi)}\clos{p}$.
 Moreover, as $\sat{\phi}$ is upwards closed, we get that $\sat{\phi} \supseteq \bigcup_{p\in\rep(\phi)}\clos{p}$.
 Therefore, we can conclude that $\sat{\phi} = \bigcup_{p\in\rep(\phi)}\clos{p}$,
 which completes the proof.
 \end{proof}



We summarize the results we provided so far in the following corollary. 

\begin{corollary} \label{cor:decomposability_condition_for_spectrum}
 Let $\Lset$ be a logic that features at least the Boolean connective $\wedge$
 and such that:
 \begin{compactenum}[$(i)$]
  \item \label{item:final_req_characteristic} \Lset is finitely characterized by $\B$, for some monotonic $\B$, and
  \item \label{item:final_req_chibar} for each $\chi(p)$, there exists in \Lset a formula $\bar\chi(p)$ such that either 
  \begin{itemize}
 \item $\sat{\bar\chi(p)}=P \setminus \sat{\chi(p)}$, or
 \item   $p \notin \sat{\bar\chi(p)}$ and $\{ q \in P \mid \Lset(q) \not\subseteq \Lset(p) \} \subseteq \sat{\bar\chi(p)}$.
 \end{itemize}
 \end{compactenum}
Then, \Lset is decomposable.
\end{corollary}
\begin{proof}
 By proviso (\ref{item:final_req_characteristic}) of the corollary and by
 Proposition~\ref{prop:characteristic}, $\chi(p)$ exists for each
 $p \in P$ (condition~(\ref{item:req_all_p_char}) in Proposition~\ref{prop:requirements}).
 By Corollary~\ref{cor:inf_des_sequence}, there exists an element $p \in P$ that is
 minimal in $\sat{\phi}$, for each consistent formula $\phi \in \Lset$
 (condition~(\ref{item:req_minimal}) in Proposition~\ref{prop:requirements}).
 The thesis immediately follows from Proposition~\ref{prop:requirements}, given
 that proviso (\ref{item:final_req_chibar}) of the corollary is equivalent to
 condition~(\ref{item:req_chibar}) in Proposition~\ref{prop:requirements}.
%
%
\end{proof}

 \begin{remark} \label{rem:disjunction}
It is worth pointing out that the Boolean connective $\wedge$ plays
a minor role in (the proof of) Proposition~\ref{prop:characteristic}
(and thus Corollaries~\ref{cor:rep} and~\ref{cor:decomposability_condition_for_spectrum}).
Indeed, it is applied to formulae in $\B(p)$ only.
Thus, such a result can be used also to deal with logics that allow for a
limited use of such a connective, such as the logic for trace semantics and
other linear-time semantics (see Section~\ref{sec:linear-time-spectrum}).
 \end{remark}




 As another path towards decomposability, we show the following result, which we
 will use in Section~\ref{sec:finitely_many_processes} to deal with logical
 settings requiring a special treatment.

\begin{proposition} \label{prop:decomposability-for-finite-modal-covariant}
  If $\Lset$ is a logic that
  \begin{inparaenum}[$(i)$]
  \item \label{item:or} features the Boolean connective $\vee$ and such that
  \item \label{item:finitely-represented} every $\phi \in \Lset$ is finitely
    represented and
  \item \label{item:characteristic} a characteristic formula exists in \Lset for
    every process $p \in P$,
  \end{inparaenum}
  then \Lset is decomposable.
\end{proposition}
\begin{proof}
  Let $\phi \in \Lset$ be a consistent formula which is not characteristic for
  any $p \in P$.
  We show that $\phi$ is decomposable, that is,
  $\sat{\phi}=\sat{\chi(p)}\cup\sat{\psi_p}$ for some $p\in P$ and
  $\psi_p\in\Lset$, with $p\not\in\sat{\psi_p}$.

  By assumption~(\ref{item:finitely-represented}), $\phi$ is finitely
  represented by some set $\rep(\phi) \subseteq P$; by
  assumption~(\ref{item:characteristic}), $\chi(p)$ exists in \Lset for all $p \in
  P$; thus, by
  Proposition~\ref{prop:general_properties_of_formulae}(\ref{item:char1}), we
  have:

  \smallskip

  {\centering

    $\sat{\phi} = \bigcup_{p \in \rep(\phi)} \clos{p} = \bigcup_{p \in
      \rep(\phi)} \sat{\chi(p)}$.

  }

  \smallskip

  Since $\phi$ is consistent, $\sat{\phi} \neq \emptyset$ and therefore
  $\rep(\phi) \neq \emptyset$.
  Pick some $p \in \rep(\phi)$.
  Notice that $|\rep(\phi)| > 1$ otherwise $\phi$ would be characteristic for
  $p$.
  Thus, using assumption~(\ref{item:or}), we have:

  \smallskip

  {\centering

    $
    \begin{array}{ll}
      \sat{\phi}
      & = \sat{\chi(p)} \cup \bigcup_{q \in \rep(\phi) \setminus \{ p \}}
        \sat{\chi(q)} \\
      & = \sat{\chi(p)} \cup \sat{\bigvee_{q \in \rep(\phi) \setminus \{ p \}}
        \chi(q)}.
    \end{array}$

  }

  \smallskip

  \noindent
  Since $p \notin \sat{\bigvee_{q \in \rep(\phi) \setminus \{ p \}} \chi(q)}$,
  we can conclude that \Lset is decomposable.
%
\end{proof}

In the remainder of the paper we will examine some applications for our general
results; in particular, we will use them to prove characterization by primality
for several well-known process semantics.
First, we consider some cases that can be dealt with using
Proposition~\ref{prop:decomposability-for-finite-modal-covariant}, and then we
analyze semantics in van Glabbeek's spectrum.


\section{Applications to finitely-represented logics}
%
%
\label{sec:finitely_many_processes}

As a first application of our results, we investigate three kinds of logics:
\begin{enumerate}
\item the set of processes $P$ is finite and the logic $\Lset$ features at least
  the Boolean connectives $\wedge$ and $\vee$;

\item the logic characterizing modal refinement semantics~\cite{BoudolL1992};

\item the logic characterizing covariant-contravariant simulation
  semantics~\cite{AcetoEtAl11b}.
\end{enumerate}
Note that in case 1 although $P$ itself is finite, it can contain processes with
infinite behaviours, e.g., when $p \in P$ represents a labelled transition
system with loops.
To deal with these logics we use
Proposition~\ref{prop:decomposability-for-finite-modal-covariant} and thus we
show that the logics satisfy its hypothesis.

\begin{proposition} \label{prop:finitely_represented_finite_processes}
  Let $\Lset$ be a logic that features at
  least the Boolean connective $\wedge$ and is interpreted over a finite set $P$.
  Then:
  \begin{compactenum}[$(a)$]
  \item \label{item:characterized}
    $\Lset$ is finitely characterized by $\B$, for some monotonic $\B$, and
  \item \label{item:represented}
    every formula $\phi \in \Lset$ is finitely represented.
  \end{compactenum}
\end{proposition}
\begin{proof}
  \begin{compactenum}[$(a)$]
  \item If $P$ is finite, so is $\Lset$, up to logical equivalence.
  Let $\Lset^{\fin}$ be a set of representatives of the equivalence classes of
  $\Lset$ modulo logical equivalence, and define
  $\B^{\fin}(p)=\Lset^{\fin}(p)=\Lset(p)\cap\Lset^{\fin}$, for each $p \in P$.
  It is easy to see that \Lset is finitely characterized by $\B^{\fin}$,
  according to Definition~\ref{def:characterized}.
  Moreover, $\B^{\fin}$ is clearly monotonic.
\item The claim immediately follows from Corollary~\ref{cor:rep}, making use of
  Proposition~\ref{prop:finitely_represented_finite_processes}(\ref{item:characterized})
  and observing that $\rep(\phi) \subseteq P$ is finite for each $\phi \in \Lset$.
  \qedhere
  \end{compactenum}
\end{proof}

\begin{proposition} \label{prop:finitely_represented_implies_decomposability}
  Every logic $\Lset$ that is interpreted over a finite set $P$ and that
  features at least the Boolean connectives $\wedge$ and $\vee$ is decomposable.
\end{proposition}
\begin{proof}
%
%
  By hypothesis, \Lset features the Boolean connective $\vee$; by
  Proposition~\ref{prop:finitely_represented_finite_processes}(\ref{item:represented}),
  every $\phi \in \Lset$ is finitely represented; by
  Proposition~\ref{prop:characteristic} and
  Proposition~\ref{prop:finitely_represented_finite_processes}(\ref{item:characterized}),
  $\chi(p)$ exists in $\Lset$ for all $p \in P$.
  Finally, by Proposition~\ref{prop:decomposability-for-finite-modal-covariant},
  \Lset is decomposable.
\end{proof}

\begin{proposition} \label{prop:modal_covariant_decomposability}
  The logics characterizing modal refinement semantics or
  covariant-contravariant simulation semantics given
  in~\cite{AcetoEtAl11b,BoudolL1992} are decomposable.
\end{proposition}
\begin{proof} The claim follows immediately from the definition of the logics and
  related results shown in~\cite{AcetoEtAl11b,BoudolL1992}, and by applying
  Proposition~\ref{prop:decomposability-for-finite-modal-covariant}.

  More precisely, for covariant-contravariant simulation, Lemma~2
  in~\cite{AcetoEtAl11b} yields the existence of characteristic formulae and
  Theorem~3 in~\cite{AcetoEtAl11b} yields finite representability of each
  formula.

  On the other hand, for modal refinement, Proposition~3.2 in~\cite{BoudolL1992}
  yields the existence of characteristic formulae and Proposition~4.2
  in~\cite{BoudolL1992} gives finite representability of each formula.
%
\end{proof}

Thus, we have the following theorem, resulting from
Theorems~\ref{theo:char_prime} and~\ref{theo:decomp}, along with
Propositions~\ref{prop:finitely_represented_implies_decomposability}
and~\ref{prop:modal_covariant_decomposability}.
%

\begin{theorem} [Characterization by primality] \label{theo:char_by_prim_finite}
  Let $\Lset$ be:
  \begin{itemize}
  \item the logic from~\cite{BoudolL1992} that characterizes modal refinement
    semantics,
  \item the one given in~\cite{AcetoEtAl11b} characterizing
    covariant-contravariant simulation semantics, or
  \item any logic that features at least the
    Boolean connectives $\wedge$ and $\vee$, and is interpreted over a finite set $P$ .

  \end{itemize}
  Then, each formula $\phi \in \Lset$ is consistent and prime if and only if
  $\phi$ is characteristic for some $p \in P$.
\end{theorem}


\section{Application to semantics in van Glabbeek's spectrum} \label{sec:spectrum}
Our next task is to apply the result described in Corollary \ref{cor:decomposability_condition_for_spectrum} to the semantics in the linear time-branching time spectrum~\cite{dFGPR13,VanGlabbeek01}, over finite trees and with a finite set of actions. All those semantics have been shown to be  characterized by specific logics and therefore inherit all the properties of logically defined preorders.
We reason about characterization by primality
(Theorems~\ref{theo:char_by_prim_branching_spectrum}
and~\ref{theo:char_by_prim_linear_spectrum}) by showing that each logic is
finitely characterized by some monotonic $\B$, and by building, for each
characteristic formula $\chi(p)$, a formula $\bar\chi(p)$ with the properties
specified in
Corollary~\ref{cor:decomposability_condition_for_spectrum}(\ref{item:final_req_chibar}).

\smallskip

\noindent{\bf Processes.} \label{page:process-definition}
To begin with, we give a formal definition of the notion of process.
%
The set of processes $P$ over a finite set of actions $Act$
is given by the following grammar:

{\centering{
$ p ::= \zeroProcess \mid \aProcess p \mid p + p,$

}}
\noindent where $\aProcess \in Act$.
Given a process $p$, we say that $p$ can perform the action
$\aProcess$ and evolve into $p'$, denoted $p \stackrel{a}{\rightarrow} p'$,
iff
\begin{inparaenum}[$(i)$]
 \item $p = \aProcess p'$ or
 \item $p = p_1 + p_2$ and $p_1 \stackrel{a}{\rightarrow} p'$
  or $p_2 \stackrel{a}{\rightarrow} p'$ holds.
\end{inparaenum}
Note that every process $p$ denotes a finite loop-free labelled transition
system whose states are those that are reachable from $p$ via transitions
$\stackrel{a}{\rightarrow}$, $a \in Act$, and whose initial state is
$p$~\cite{Keller76}.


We define the set of \emph{initials} of $p$, denoted $I(p)$,
as the set $\{ \aProcess \in Act \mid p~\stackrel{\aProcess}{\rightarrow}~p'$
for some $p' \in P \}$.
We write $p \stackrel{a}{\rightarrow}$ if $\aProcess \in I(p)$,
 $p \not \stackrel{a}{\rightarrow}$ if $\aProcess \not\in I(p)$, and $p \not \rightarrow$ if $I(p) = \emptyset$. We define $\traces(p)$ as follows (we use $\varepsilon$ to denote the empty string):

  \begin{equation}
    \label{eq:trace_definition}
    \traces(p) = \{ \varepsilon \} \cup \{ a\tau \mid \exists p' \in P \ . \
    p~\stackrel{\aProcess}{\rightarrow}~p' \text{ and } \tau \in \traces(p') \}.
  \end{equation}
%
\noindent For each trace $\tau = a_1 \ldots a_n$, we write
$p~\stackrel{\tau}{\rightarrow}~p'$ for
$p~\stackrel{a_1}{\rightarrow}~p_1~\stackrel{a_2}{\rightarrow}~p_2 \ldots
p_{n-1}~\stackrel{a_n}{\rightarrow}~p'$.
Finally, for each $p\in P$, $\depth(p)$ is the length of a longest trace in $\traces(p)$. 

\begin{figure}[ht!]
  \centering

  \begin{tikzpicture}
    \node at (0,0) (simulation){\simulation};
    \node at (simulation) [above=15mm] (completeSim){\completeSim};
    \node at (completeSim) [above=15mm] (readySim){\readySim};
    \node at (readySim) [above=15mm] (traceSim){\traceSim};
    \node at (traceSim) [above=15mm] (twoSim){\twoSim};
    \node at (twoSim) [above=15mm] (bisim){\bisim};

    \node at (simulation) [right=6cm] (trace){\trace};
    \node at (trace) [above=15mm] (completeTrace){\completeTrace};

    \node at (completeTrace) [above left=15mm and 16.5mm] (readyTrace){\readyTrace};
    \node at (readyTrace) [above right=2.5mm and 10mm] (failureTrace){\failureTrace};
    \node at (readyTrace) [below right=2.5mm and 25mm] (ready){\ready};
    \node at (readyTrace) [right=37.5mm] (failure){\failure};

    \path[draw,->] (readyTrace) -- (failureTrace);
    \path[draw,->] (readyTrace) -- (ready);
    \path[draw,->] (failureTrace) -- (failure);
    \path[draw,->] (ready) -- (failure);

    \node at (failureTrace) [draw=gray!70,rounded corners,above=15mm] (impossibleFutureTrace){\impossibleFutureTrace};
    \node at (readyTrace) [draw=gray!70,rounded corners,above=15mm] (possibleFutureTrace){\possibleFutureTrace};
    \node at (ready) [above=15mm] (possibleFuture){\possibleFuture};
    \node at (failure) [above=15mm] (impossibleFuture){\impossibleFuture};
    \path[draw,->] (possibleFutureTrace) -- (impossibleFutureTrace);
    \path[draw,->] (possibleFutureTrace) -- (possibleFuture);
    \path[draw,->] (impossibleFutureTrace) -- (impossibleFuture);
    \path[draw,->] (possibleFuture) -- (impossibleFuture);

    \node at (impossibleFutureTrace) [draw=gray!70,rounded corners,above=15mm] (impossibleSimulationTrace){\impossibleSimulationTrace};
    \node at (possibleFutureTrace) [draw=gray!70,rounded corners,above=15mm] (possibleSimulationTrace){\possibleSimulationTrace};
    \node at (possibleFuture) [draw=gray!70,rounded corners,above=15mm] (possibleSimulation){\possibleSimulation};
    \node at (impossibleFuture) [draw=gray!70,rounded corners,above=15mm] (impossibleSimulation){\impossibleSimulation};
    \path[draw,->] (possibleSimulationTrace) -- (impossibleSimulationTrace);
    \path[draw,->] (possibleSimulationTrace) -- (possibleSimulation);
    \path[draw,->] (impossibleSimulationTrace) -- (impossibleSimulation);
    \path[draw,->] (possibleSimulation) -- (impossibleSimulation);

    \path[draw,->] (bisim) -- (twoSim);
    \path[draw,->] (twoSim) -- (traceSim);
    \path[draw,->] (traceSim) -- (readySim);
    \path[draw,->] (readySim) -- (completeSim);
    \path[draw,->] (completeSim) -- (simulation);

    \path[draw,->] (twoSim) -- (possibleSimulationTrace);
    \path[draw,->] (traceSim) -- (possibleFutureTrace);
    \path[draw,->] (readySim) -- (readyTrace);
    \path[draw,->] (completeSim) -- (completeTrace);
    \path[draw,->] (simulation) -- (trace);

    \path[draw,->] (possibleSimulationTrace) -- (possibleFutureTrace);
    \path[draw,->] (possibleFutureTrace) -- (readyTrace);
    \path[draw,->] (impossibleSimulationTrace) -- (impossibleFutureTrace);
    \path[draw,->] (impossibleFutureTrace) -- (failureTrace);
    \path[draw,->] (possibleSimulation) -- (possibleFuture);
    \path[draw,->] (possibleFuture) -- (ready);
    \path[draw,->] (impossibleSimulation) -- (impossibleFuture);
    \path[draw,->] (impossibleFuture) -- (failure);

    \path[draw,->] (readyTrace) -- (completeTrace);
    \path[draw,->] (ready) -- (completeTrace);
    \path[draw,->] (failure) -- (completeTrace);
    \path[draw,->] (failureTrace) -- (completeTrace);

    \path[draw,->] (completeTrace) -- (trace);

    \draw[dashed] (2.5,-.25) -- ++(0,9.25);

    \draw[rounded corners,draw=black,draw opacity=.0,fill=gray!50,fill
    opacity=.3] (readyTrace) ++(-.5,.7) rectangle ++(8,-1.4);
    \draw (readyTrace) ++(5.3,0)node[right]{1st diamond};

    \draw[rounded corners,draw=black,draw opacity=.0,fill=gray!50,fill
    opacity=.3] (possibleFutureTrace) ++(-.5,.7) rectangle ++(8,-1.4);
    \draw (possibleFutureTrace) ++(5.3,0)node[right]{2nd diamond};

    \draw[rounded corners,draw=black,draw opacity=.0,fill=gray!50,fill
    opacity=.3] (possibleSimulationTrace) ++(-.5,.7) rectangle ++(8,-1.4);
    \draw (possibleSimulationTrace) ++(5.3,0)node[right]{3rd diamond};

  \end{tikzpicture}

  \caption{Semantic relations in van Glabbeek's linear time-branching time
    spectrum (branching semantics are on the left, linear ones are on the
    right---the six framed names are introduced here: these semantics were studied
    in~\cite{dFGPR13} but no name was assigned to them).  }
  \label{fig:VGspectrum}
\end{figure}
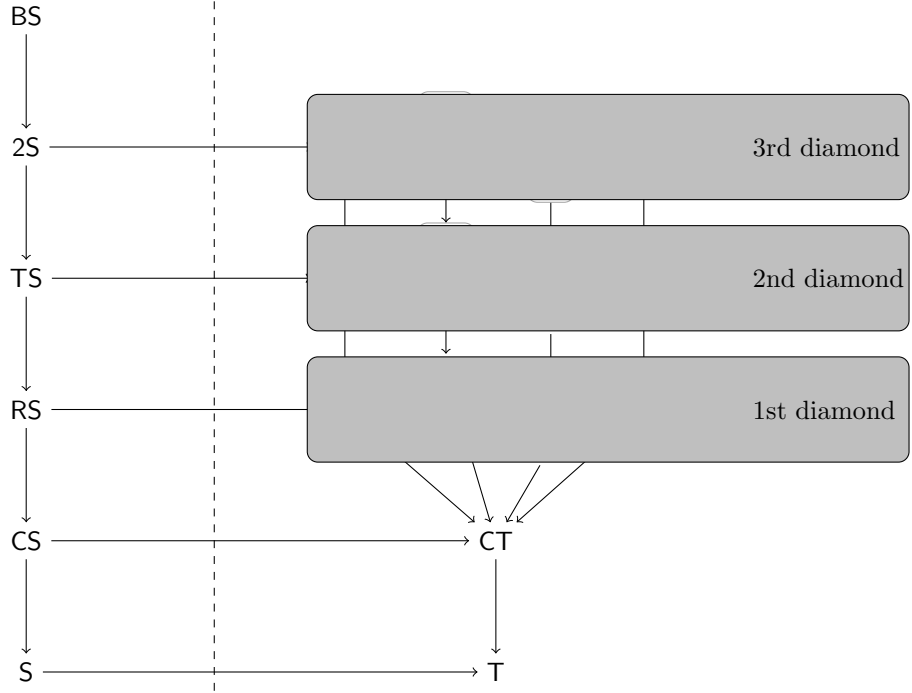

\noindent{\bf Behavioural preorders for process semantics.}
The semantics of processes is expressed by means of preorders,
which, intuitively, classify processes according to their possible
behaviours. Roughly speaking, a process \emph{follows} another in the preorder
(or it is \emph{above} it)
if it exhibits at least the same behaviours as the latter.
The semantic relations in van Glabbeek's linear time-branching time spectrum
present different levels of granularity: a finer relation is able to distinguish processes that are indistinguishable by a
coarser one.
Those semantics are as follows (see~\figurename~\ref{fig:VGspectrum}):
\begin{itemize}
\item branching time semantics (\figurename~\ref{fig:VGspectrum}, left-hand
side): \emph{simulation} (\simulation),
  \emph{complete simulation} (\completeSim),
  \emph{ready simulation} (\readySim),
  \emph{trace simulation} (\traceSim),
  \emph{2-nested simulation} (\twoSim), and
  \emph{bisimulation} (\bisim);
\item linear time semantics (\figurename~\ref{fig:VGspectrum}, right-hand
side): \emph{trace} (\trace),
  \emph{complete trace} (\completeTrace),
  \emph{failure} (\failure),
  \emph{failure trace} (\failureTrace),
  \emph{ready} (\ready),
  \emph{ready trace} (\readyTrace),
  \emph{impossible future} (\impossibleFuture),
  \emph{possible future} (\possibleFuture),
  \emph{always impossible future} (\impossibleFutureTrace),
  \emph{always possible future} (\possibleFutureTrace),
  \emph{impossible 2-simulation} (\impossibleSimulation),
  \emph{possible 2-simulation} (\possibleSimulation),
  \emph{always impossible 2-simulation} (\impossibleSimulationTrace),
  \emph{always possible 2-simulation} (\possibleSimulationTrace).
\end{itemize}
%
In the rest of this section $\bspectrumSet$ denotes the set $\{ \simulation,
\completeSim, \readySim, \traceSim, \twoSim, \bisim \}$, $\lspectrumSet$ \label{lspectrumSet} denotes
the set $\{ \trace, \completeTrace, \failure, \failureTrace, \ready,
\readyTrace, \impossibleFuture, \possibleFuture, \impossibleFutureTrace,
\possibleFutureTrace, \impossibleSimulation, \possibleSimulation,
\impossibleSimulationTrace, \possibleSimulationTrace \}$, and we let
$\vgspectrumSet = \bspectrumSet \cup \lspectrumSet$.


The remainder of the section is organized as follows.
In Section~\ref{sec:branch-time-spectr}, we establish characterization by
primality for the semantics in van Glabbeek's branching time spectrum.
Then, in Section~\ref{sec:linear-time-spectrum}, we deal with the ones in van
Glabbeek's linear time spectrum; at the beginning of this last section, we also
give a short account on the conceptual differences between the two semantics;
for a comprehensive account, we refer to~\cite{dFGPR13,VanGlabbeek01}.
%
%


\subsection{The branching time spectrum} \label{sec:branch-time-spectr}

In this sub-section, we focus on the semantics in \emph{van Glabbeek's branching
  time spectrum}, originally introduced in~\cite{VanGlabbeek01} and successively
generalized in~\cite{dFGPR13}, and their corresponding logical formalisms.

\begin{definition}[Branching time semantic relations~\cite{dFGPR13,VanGlabbeek01}]
  \label{def:branching-semantics}
For each $X \in \bspectrumSet$, $\<_X$ is the largest relation over $P$
satisfying the following conditions for each $p,q\in P$.

\begin{description}

\item[simulation (\simulation):] $p\<_\simulation q
  \Leftrightarrow$ for all $p\stackrel{a}{\rightarrow}p'$ there exists some
  $q\stackrel{a}{\rightarrow}q'$ such that $p'\<_\simulation q'$;

\item [complete simulation (\completeSim):] $p\<_\completeSim q \Leftrightarrow$
  \begin{inparaenum}[$(i)$]
  \item for all $p\stackrel{a}{\rightarrow}p'$ there exists some
    $q\stackrel{a}{\rightarrow}q'$ such that $p'\<_\completeSim q'$, and
  \item $I(p)=\emptyset$ iff $I(q)=\emptyset$;
  \end{inparaenum}

\item [ready simulation (\readySim):] $p\<_\readySim q \Leftrightarrow$
  \begin{inparaenum}[$(i)$]
  \item for all $p\stackrel{a}{\rightarrow}p'$ there exists some
    $q\stackrel{a}{\rightarrow}q'$ such that $p'\<_\readySim q'$, and
  \item $I(p)=I(q)$;
  \end{inparaenum}

\item [trace simulation (\traceSim):] $p\<_\traceSim q \Leftrightarrow$
  \begin{inparaenum}[$(i)$]
  \item for all $p\stackrel{a}{\rightarrow}p'$ there exists some
    $q\stackrel{a}{\rightarrow}q'$ such that $p'\<_\traceSim q'$, and
  \item $\traces(p)=\traces(q)$;
  \end{inparaenum}

\item [2-nested simulation (\twoSim):] $p\<_\twoSim q \Leftrightarrow$
  \begin{inparaenum}[$(i)$]
  \item for all $p\stackrel{a}{\rightarrow}p'$ there exists some
    $q\stackrel{a}{\rightarrow}q'$ such that $p'\<_\twoSim q'$, and
  \item $q\<_\simulation p$;
  \end{inparaenum}

\item [bisimulation (\bisim):] $p\<_\bisim q \Leftrightarrow$
  \begin{inparaenum}[$(i)$]
  \item for all $p\stackrel{a}{\rightarrow}p'$ there exists some
    $q\stackrel{a}{\rightarrow}q'$ such that $p'\<_\bisim q'$, and
  \item for all $q\stackrel{a}{\rightarrow}q'$ there exists some
    $p\stackrel{a}{\rightarrow}p'$ such that $p'\<_\bisim q'$.
  \end{inparaenum}
\end{description}

%
%
%
%
%
%

\noindent For each $X \in \bspectrumSet$, the equivalence relation $\equiv_X$ is
defined as expected, i.e.,

{\centering
  $p \equiv_X q \Leftrightarrow p \<_X q$ and $q \<_X p$.

}
\end{definition}

\smallskip

\noindent{\bf Branching time logics.}
The languages of the different logics yield the following chain of strict
inclusions (Table~\ref{fig:VGspectrum}, left-hand side):
$ \Lset_\simulation \subset \Lset_\completeSim \subset
\Lset_\readySim \subset \Lset_\traceSim \subset \Lset_\twoSim \subset
\Lset_\bisim,$
corresponding to a chain of formalisms with strictly increasing expressive power.
Notice that, as it will appear clear after the definition of the satisfaction relation below,
some of the languages present some redundancy, in the sense that they could be replaced
with smaller ones, without any loss in expressiveness.
For instance, a disjunction is expressible in $\Lset_\bisim$ using
conjunction and negation, and suitably replacing $\true$ with $\false$ and vice versa.
We followed this approach because we find it helpful to have
syntactically larger languages corresponding to more expressive semantics.

The syntax of the logics characterizing the semantics in van Glabbeek's
branching time spectrum is given in the following definition.
There, we treat formulae of the form $\univer{a} \psi$ and $\zeroFormula$ as
syntactic shorthands for, respectively, $\neg \exist{a} \neg \psi$ and $\bigwedge_{a \in Act} \univer{a} \false$.

\begin{definition}[Syntax~\cite{dFGPR13,VanGlabbeek01}]\label{def:syntax_branching_spectrum}
  For each $X \in \bspectrumSet$, $\Lset_X$ is the language defined by the
  corresponding grammar given below:

$\Lset_\simulation$:

{\centering
$\phi_\simulation ::= \true \mid \false \mid \phi_\simulation \wedge \phi_\simulation
\mid \phi_\simulation \vee \phi_\simulation \mid \exist{a} \phi_\simulation$.

}

$\Lset_\completeSim$:

{\centering
$\phi_\completeSim ::= \true \mid \false \mid \phi_\completeSim \wedge \phi_\completeSim
\mid \phi_\completeSim \vee \phi_\completeSim \mid \exist{a} \phi_\completeSim \mid
\zeroFormula$.

}

$\Lset_\readySim$:

{\centering
$\phi_\readySim ::= \true \mid \false \mid \phi_\readySim \wedge \phi_\readySim
\mid \phi_\readySim \vee \phi_\readySim \mid \exist{a} \phi_\readySim \mid
\univer{a} \false$.

}

$\Lset_\traceSim$:

{\centering
\begin{tabular}{lll}
$\phi_\traceSim$ & $::=$ & $\true \mid \false \mid \phi_\traceSim \wedge \phi_\traceSim
\mid \phi_\traceSim \vee \phi_\traceSim \mid \exist{a} \phi_\traceSim \mid
\psi_\traceSim$.\\

$\psi_\traceSim$ & $::=$ & $\false \mid \univer{a}\psi_\traceSim$.
\end{tabular}

}

$\Lset_\twoSim$:

{\centering
\begin{tabular}{lll}
$\phi_\twoSim$ & $::=$ & $\true \mid \false \mid \phi_\twoSim \wedge \phi_\twoSim
\mid \phi_\twoSim \vee \phi_\twoSim \mid \exist{a} \phi_\twoSim \mid
\neg\phi_\simulation$.

\end{tabular}

}

$\Lset_\bisim$:

{\centering
$\phi_\bisim ::= \true \mid \false \mid \phi_\bisim \wedge \phi_\bisim
\mid \phi_\bisim \vee \phi_\bisim \mid \exist{a} \phi_\bisim \mid
\neg\phi_\bisim$.

}
\end{definition}

\begin{table}[t!]
\scriptsize
\everymath{\displaystyle}
 \begin{tabular}{|l|l|@{\hspace{1mm}}l|}
   \hline
   & & \\ [-2mm]
& \multicolumn{1}{l|@{\hspace{1mm}}}{\small \specialitemThree Syntax}
& \multicolumn{1}{l@{\hspace{1mm}}|}{
\begin{tabular}{l@{\hspace{2mm}}l}
  \specialitemOne
  & \small Monotonic function \B for finite characterization \\
  \small \specialitemTwo
  & \small Formula $\bar\chi(p)$
\end{tabular}
} \\
   \hline
\multicolumn{3}{|l@{\hspace{1mm}}|}{} \\ [-1mm]
 \multicolumn{3}{|l@{\hspace{1mm}}|}{
   $\begin{array}{ll@{\hspace{10mm}}r}
      \specialitemThree
      & \phi^\exists_X ::= \true \mid \false \mid
    \phi_X \wedge \phi_X \mid
    \phi_X \vee \phi_X \mid
      \exist a \phi_X
      & X \in \bspectrumSet \\ [3mm]
      \text{\specialitemOne}
      & \B_X(p) = \B^+_X(p)\cup\B^-_X(p) \\
      \text{\specialitemOne}

      & \B^+_X(p)=\{\true\}\cup\{\exist{a}\varphi\mid a \in Act,
        \varphi=\bigwedge_{\psi\in\Psi}\psi, \Psi \subseteq \B_{\genericSemanticsX}(p'),
        p\stackrel{a}{\rightarrow}p'\} \\

    \end{array}$
   \hfill

   } \\
   \hline
   & & \\ [-2mm]
   \simulation
       & $
         \begin{array}[t]{@{\specialitemThree\hspace{1mm}}l}
           \phi_\simulation ::= \phi^\exists_\simulation
         \end{array}$
         & $\begin{array}[t]{l@{\hspace{1mm}}l}
              \text{\specialitemOne}
              &\B^-_{\simulation}(p) = \emptyset \\
              \specialitemTwo
              & \barchiS(p)= \bigvee_{a\in Act} \exist{a}\bigwedge_{p\stackrel{a}{\rightarrow}p'} \barchiS(p')
            \end{array}$ \\
   \hline
   & & \\ [-2.5mm]
   \completeSim
   &
     $
     \begin{array}[t]{@{\specialitemThree\hspace{1mm}}l}
       \phi_\completeSim ::= \phi^\exists_\completeSim \mid \phi^\forall_\completeSim \\
       \phi_\completeSim^\forall ::= \zeroFormula
     \end{array}
   $
   &
     $\begin{array}[t]{l@{\hspace{1mm}}l@{\hspace{1cm}}l}
       \text{\specialitemOne}
       & \B^-_{\completeSim}(p) = \{\zeroFormula\mid p\stackrel{a}{\not\rightarrow}\forall a\in Act\} \\

       \specialitemTwo
       & \barchiCS(p)= \big( \bigvee_{a\in Act} \exist{a}\bigwedge_{p\stackrel{a}{\rightarrow}p'} \barchiCS(p') \big)\vee \zeroFormula & \text{if } I(p)\neq\emptyset \\
       \specialitemTwo
       &   \barchiCS(p)= \bigvee_{a\in Act} \exist{a}\bigwedge_{p\stackrel{a}{\rightarrow}p'} \barchiCS(p') & \text{if } I(p)=\emptyset
     \end{array}$ \\
\hline
   & & \\ [-2.5mm]
   \readySim
   &
   $\begin{array}[t]{@{\specialitemThree\hspace{1mm}}l}
    \phi_\readySim ::= \phi_\readySim^\exists \mid \phi_\readySim^\forall \\
    \phi_\readySim^\forall ::= \univer a \false
   \end{array}$
   &
     $\begin{array}[t]{l@{\hspace{1mm}}l}
        \text{\specialitemOne}
        & \B^-_{\readySim}(p)= \{\univer{a}\false\mid a\in Act, p\stackrel{a}{\not\rightarrow} \} \\
        \specialitemTwo
        & \barchiRS(p)= \big( \bigvee_{a\in Act} \exist{a}\bigwedge_{p\stackrel{a}{\rightarrow}p'} \barchiRS(p') \big) \vee \bigvee_{a\in I(p)}\univer{a}\false
     \end{array}$ \\
\hline
   & & \\ [-2mm]
   \traceSim
   &
   $\begin{array}[t]{@{\specialitemThree\hspace{1mm}}l}
      \phi_\traceSim ::= \phi_\traceSim^\exists \mid \phi_\traceSim^\forall \\
      \phi_\traceSim^\forall ::= \false \mid \univer a \phi_\traceSim^\forall
   \end{array}$
   &
     $\begin{array}[t]{l@{\hspace{1mm}}l}
        \text{\specialitemOne}
        & \B^-_{\traceSim}(p) = \{\univer{\tau a}\false\mid \tau\in \traces(p), a \in Act, \tau a \notin \traces(p) \}
        \\
        \specialitemTwo
        & \barchiTS(p)= \big( \bigvee_{a\in Act} \exist{a}\bigwedge_{p\stackrel{a}{\rightarrow}p'} \barchiTS(p') \big) \vee
          \bigvee_{\tau \in \traces(p), \tau a \notin \traces(p)}\exist{\tau a}\true \vee \bigvee_{p\stackrel{\tau a}{\rightarrow}p'}\univer{\tau a}\false
      \end{array}$ \\
          \hline
   & & \\ [-2mm]
   \twoSim
   &
   $\begin{array}[t]{@{\specialitemThree\hspace{1mm}}l}
    \phi_\twoSim ::= \phi_\twoSim^\exists \mid \phi_\twoSim^\forall \\
    \phi_\twoSim^\forall ::= \neg \phi_\simulation
   \end{array}$
   &
     $\begin{array}[t]{l@{\hspace{1mm}}l}
        \text{\specialitemOne}
        & \B^-_{\twoSim}(p)= \{\univer{a}\varphi\in\Lset_{\twoSim}(p)\mid a \in Act, \varphi=\bigvee_{p' \in \maxsucc(p,a)} \bigwedge_{\psi \in \B^-_{\twoSim}(p')}\psi \} \\

        & \text{where } \maxsucc(p,a) = \{ p' \in P \mid
          p\stackrel{a}{\rightarrow}p' \text{ and } \nexists p'' .
          p\stackrel{a}{\rightarrow}p'' \text{ and } p' <_{\simulation} p'' \}
        \\


        \specialitemTwo
        & \barchitwoS(p)= \big( \bigvee_{a\in Act} \exist{a}\bigwedge_{p\stackrel{a}{\rightarrow}p'} \barchitwoS(p') \big) \vee \bar\Phi(p) \\

        & \text{where }
        \bar\Phi(p) = \bigvee_{a\in I(p)}\univer{a}\false \vee \bigvee_{a\in I(p)} \bigvee_{p\stackrel{a}{\rightarrow}p'}\univer{a} \bar\Phi(p')
      \end{array}$ \\
   \hline
   & & \\ [-2mm]
   \bisim
   &
   $\begin{array}[t]{@{\specialitemThree\hspace{1mm}}l}
    \phi_\bisim ::= \phi_\bisim^\exists \mid \phi_\bisim^\forall \\
    \phi_\bisim^\forall ::= \neg \phi_\bisim
   \end{array}$
   &

     $\begin{array}[t]{l@{\hspace{1mm}}l} \text{\specialitemOne}

        & \B^-_{\bisim}(p)= \{\univer{a}\varphi \in \Lset_{\bisim}(p) \mid a \in
          Act, \varphi= \bigvee_{p\stackrel{a}{\rightarrow}p'} \bigwedge_{\psi \in
          \B_{\bisim}(p')} \psi \} \\


%

        \specialitemTwo

        & \barchiBS(p)=\neg\chiBS(p) \text{\hspace{7mm} ($\chiBS(p)$ is
          characteristic for $p$ within $\Lset_{\bisim}$)}
      \end{array}$ \\
\hline
 \end{tabular}
 \caption{Syntax, monotonic function \B for finite characterization, and formula
   $\bar\chi(p)$, relative to the logics for the semantics in van Glabbeek's
   branching time spectrum.}
\label{tab:syntax_monotonicity_finite_char_branching_spectrum}
\end{table}

\begin{remark}
  It is important to notice that disjunction does not appear in the original
  formulation found in~\cite{VanGlabbeek01} of the logics characterizing the
  branching semantics in the spectrum.
  However, adding it does not affect the expressive power of these logics with
  respect to the corresponding semantics, as shown in~\cite[Proposition 6.2 at
  page 42]{dFGPR13}.
\end{remark}

We give here the semantics of the logics, by describing the satisfaction
relation for the most expressive one, namely $\Lset_\bisim$,
that characterizes bisimulation semantics.
%
The semantics for the other logics can be obtained by considering
the corresponding subset of clauses.

\begin{definition}[Satisfaction relation] \label{def:semantics_branching}
  The satisfaction relation for the logic $\Lset_\bisim$ is defined as follows:
  \begin{itemize}
  \item $p \in \sat{\true}$, for every $p \in P$,
  \item $p \notin \sat{\false}$, for every $p \in P$,
 \item $p \in \sat{\phi_1 \wedge \phi_2}$ iff $p \in \sat{\phi_1}$ and
  $p \in \sat{\phi_2}$,
 \item $p \in \sat{\phi_1 \vee \phi_2}$ iff $p \in \sat{\phi_1}$ or
  $p \in \sat{\phi_2}$,
 \item $p \in \sat{\exist{a} \phi}$ iff $p' \in \sat{\phi}$
  for some $p' \in P$ such that $p \stackrel{a}{\rightarrow} p'$,
 \item $p \in \sat{\neg \phi}$ iff $p \notin \sat{\phi}$.
\end{itemize}
\end{definition}
We say that a process $p$ \emph{satisfies} a formula $\phi \in \Lset_{\bisim}$
if, and only if, $p \in \sat{\phi}$.


%
%
%

\smallskip

The following well-known theorem states the relationship between logics and
process semantics that allows us to use our general results about logically characterized semantics.
\begin{theorem}[Logical characterization of branching time semantics~\cite{dFGPR13,VanGlabbeek01}]\label{theo:branching-semantics} 
For each $X\in\bspectrumSet$ and for all $p,q \in P$, $p\<_Xq$ iff $\Lset_X(p)\subseteq\Lset_X(q)$.
\end{theorem}

We observe that all the logics defined above feature the
Boolean connective $\wedge$, as required by one of the assumptions of
Corollary~\ref{cor:decomposability_condition_for_spectrum}.
In what follows, we show that every logic meets also the other conditions of the
corollary, that is, it is finitely characterized by some monotonic $\B$, and for
each $\chi(p)$ there exists a formula $\bar\chi(p)$ such that either
$\sat{\bar\chi(p)}=P \setminus \sat{\chi(p)}$ (as it is the case for
$\Lset_{\bisim}$) or $p \notin \sat{\bar\chi(p)}$ and $\{ q \in P \mid \Lset(q)
\not\subseteq \Lset(p) \} \subseteq \sat{\bar\chi(p)}$ (which holds in all the
other cases).
This yields the characterization by primality for the logics for the semantics
in \bspectrumSet (Theorem~\ref{theo:char_by_prim_branching_spectrum}).

To this end, we first summarize, in
Table~\ref{tab:syntax_monotonicity_finite_char_branching_spectrum}, both the functions $\B$
and the formulae $\bar\chi(p)$ for all the branching time semantics (rightmost
column),
%
%
and then we prove their correctness.
In particular, in Lemma~\ref{lem:characterization_branching}, we prove the
finite characterization result,
%
%
while in Lemma~\ref{lem:anti-char_branching} we show the correctness of the formula
$\bar\chi(p)$.
%
%
We have already pointed out the connection between our function \B and the
notion of characteristic formulae (see Proposition~\ref{prop:characteristic}).
As a matter of fact, roughly speaking, the definition of $\B_X(p)$ (with $X \in
\bspectrumSet$) provided in
Table~\ref{tab:syntax_monotonicity_finite_char_branching_spectrum} somehow
correspond to breaking the characteristic formula of $p$ (within logic
$\Lset_X$) into its conjuncts.
The same applies for the semantics in the linear time spectrum we will consider in Section~\ref{sec:linear-time-spectrum}.

For the sake of readability, besides including the functions $\B$ and the formulae
$\bar\chi(p)$,
Table~\ref{tab:syntax_monotonicity_finite_char_branching_spectrum} also recalls,
in a compact but equivalent way, the syntax of the different logical formalisms.
Roughly speaking, each language consists of an ``existential'' and a
``universal'' sub-language, as highlighted by the definitions in the second
column of \tablename~\ref{tab:syntax_monotonicity_finite_char_branching_spectrum} ($\phi_X
::= \phi_X^\exists \mid \phi_X^\forall$ for each $X \in \bspectrumSet$ apart
from simulation).
The ``existential'' sub-language (formulae derivable from the non-terminal
$\phi_X^\exists$) is common to all the logics and so is its definition (top of
\tablename~\ref{tab:syntax_monotonicity_finite_char_branching_spectrum}).
The ``universal'' sub-language (formulae derivable from
the non-terminal $\phi_X^\forall$) is what actually distinguishes
the various languages: its definition is provided for each logic
in the corresponding row.
%
The operators $\exist{\tau}$ and $\univer{\tau}$ (with $\tau = a_1 a_2 \ldots a_k
\in Act^*$),
occurring in the definition of $\B$ for trace simulation (\traceSim),
are abbreviations for
$\exist{a_1}\exist{a_2} \ldots \exist{a_k}$
and $\univer{a_1}\univer{a_2} \ldots \univer{a_k}$, respectively
(notice that, in particular, both $\exist{\varepsilon} \varphi$ and $\univer{\varepsilon} \varphi$ will be treated as
$\varphi$).
%
%
%
%

%
%


\newcounter{lemcharacterizationready}
\setcounter{lemcharacterizationready}{\value{lemma}}
\newcommand{\lemcharacterizationready}{
  Let $X \in \bspectrumSet$. $\Lset_X$ is finitely characterized by $\B_X$, for
  some monotonic $\B_X$ (see
  Table~\ref{tab:syntax_monotonicity_finite_char_branching_spectrum}).
}

\begin{lemma} \label{lem:characterization_branching}
  \lemcharacterizationready
\end{lemma}


\begin{proof}
  We detail the case of ready simulation only (the proof for the other cases can
  be found in~\ref{sec:full-proof-finite-char}).
  For the sake of clarity we recall
  from \tablename~\ref{tab:syntax_monotonicity_finite_char_branching_spectrum}
  that $\boundfun_{\readySim}$ is defined as $\boundfun^+_{\readySim}(p)\cup
  \boundfun^-_{\readySim}(p)$, where
  \begin{compactitem}
  \item $\boundfun^+_{\readySim}(p)=\{\true\}\cup\{\exist{a}\varphi\mid a \in
    Act, \varphi= \bigwedge_{\psi\in\Psi}\psi, \Psi \subseteq
    \B_{\genericSemanticsX}(p'), p\stackrel{a}{\rightarrow}p'\}$, and

  \item $\boundfun^-_{\readySim}(p)=\{\univer{a}\false\mid
    a\in Act, p\stackrel{a}{\not\rightarrow} \}$.
  \end{compactitem}

  For the sake of a lighter notation, we omit the subscript $_{\readySim}$,
  i.e., we write $\boundfun$ (resp., $\boundfun^+$, $\boundfun^-$, $\<$) for
  $\boundfun_{\readySim}$ (resp., $\boundfun^+_{\readySim}$,
  $\boundfun^-_{\readySim}$, $\<_{\readySim}$) since there is no risk of
  ambiguity.

We show that, for every $p \in P$,
\begin{enumerate}[label=(\emph{\roman*})]
\item \label{item:Bsoundness} $\emptyset \subset \B(p)\subseteq\Lset(p)$,
\item \label{item:Bentailment} for each $\phi\in\Lset(p)$, it holds
  $\bigcap_{\psi\in\B(p)}\sat{\psi}\subseteq\sat{\phi}$,
\item \label{item:Bfiniteness} $\B(p)$ is finite, and
\item \label{item:Bmonotonicity} for each $q\in P$, if $\Lset(p)\subseteq\Lset(q)$ then $\B(p)\subseteq\B(q)$. 
\end{enumerate}

To begin with, we prove property~\ref{item:Bfiniteness}, which also tells us
that $\B$ is well-defined.
It is immediate to see that, since $Act$ is finite, so is $\B^-(p)$ for all $p
\in P$.
We show that also $\B(p)$ is finite for every $p \in P$ by induction on the
depth of $p$.
When $I(p) = \emptyset$ (base case),
%
%
$\B(p)= \{ \true \} \cup \B^-(p)$, which is clearly finite.
Let us deal now with the inductive step ($I(p)\neq\emptyset$).
By the construction of $\B^+(p)$, a formula belongs to $\B^+(p)$ if, and only
if, it is either \true or $\exist{a}\varphi$, where $a \in Act$ and
$\varphi=\bigwedge_{\psi\in\Psi}\psi$, for some $\Psi \subseteq \B(p')$ and some
$p'$ such that $p\stackrel{a}{\rightarrow}p'$.
By the inductive hypothesis, $\B(p')$ is finite, and thus $\varphi$ is well
defined.
Since $Act$ is also finite and processes are finitely branching, there are only
finitely many such formulae $\exist{a}\varphi$, meaning that $\B^+(p)$ is
finite.
Therefore $\boundfun(p) = \boundfun^+(p)\cup \boundfun^-$ is finite as well.

In order to prove property~\ref{item:Bsoundness}, we preliminarily observe that $\true \in \B(p)$ for every $p \in P$, and thus $\emptyset \subset
\B(p)$.
Now, to prove that $\B(p) \subseteq \Lset(p)$ holds for every $p \in P$, we
first observe that $\B^-(p) \subseteq \Lset(p)$ trivially holds, by definition
of $\B^-(p)$, and then we show that $\B(p) \subseteq \Lset(p)$ also holds for
every $p \in P$, by induction on the depth of $p$.
When $I(p) = \emptyset$ (base case), we have $\B^+(p)=\{\true\} \subseteq
\Lset(p)$, and therefore $\B(p) \subseteq \Lset(p)$ holds as well.
To deal with the inductive step ($I(p) \neq \emptyset$), let $\phi\in\B^+(p)$.
If $\phi = \true$, then $\phi \in \Lset(p)$, and we are done.
Assume $\phi = \exist{a}\bigwedge_{\psi\in\Psi}\psi$, where $a \in Act$ and
$\Psi \subseteq \B(p')$ for some $p'$ such that $p\stackrel{a}{\rightarrow}p'$.
By the inductive hypothesis, we have that $\B(p')\subseteq\Lset(p')$, meaning
that $p' \in \sat{\psi}$ for all $\psi \in \Psi$.
Since $p\stackrel{a}{\rightarrow}p'$, we have that
$p\in\sat{\exist{a}\bigwedge_{\psi\in\Psi}\psi}$, which amounts to $\phi \in
\Lset(p)$.

%
%

In order to prove property~\ref{item:Bentailment}, we let $\phi \in \Lset(p)$,
for a generic $p \in P$, and we proceed by induction on the structure of $\phi$
(notice that we can ignore the case $\phi = \false$, as $\phi \in \Lset(p)$
implies $\phi \neq \false$).
\begin{compactitem}
\item $\phi=\true$ or $\phi=\univer{a}\false$: it is enough to observe that
  $\phi\in\B(p)$, which implies $\bigcap_{\psi\in\B(p)} \sat{\psi} \subseteq
  \sat{\phi}$.

\item $\phi= \varphi_1 \vee \varphi_2$: it holds that $\varphi_i\in\Lset(p)$ for
  some $i\in \{ 1,2 \}$. By the inductive hypothesis, we have that
  $\bigcap_{\psi\in\B(p)}\sat{\psi}\subseteq \sat{\varphi_i}$ and, since
  $\sat{\varphi_i}\subseteq\sat{\phi}$, we obtain the claim.

\item $\phi= \varphi_1 \wedge \varphi_2$: it holds that $\varphi_i\in\Lset(p)$
  for all $i\in \{ 1,2 \}$.
  By the inductive hypothesis, we have that
  $\bigcap_{\psi\in\B(p)}\sat{\psi}\subseteq \sat{\varphi_i}$ for all $i \in \{
  1,2 \}$.
  This implies that $\bigcap_{\psi \in \B(p)} \sat{\psi} \subseteq
  \sat{\varphi_1} \cap \sat{\varphi_2} = \sat{\phi}$.

\item $\phi=\exist{a}\varphi$: by definition we have that $\varphi\in\Lset(p')$
  for some $p\stackrel{a}{\rightarrow}p'$.
  By the inductive hypothesis, we have that
  $\bigcap_{\psi\in\B(p')}\sat{\psi}\subseteq \sat{\varphi}$.
  We define $\zeta=\exist{a}\bigwedge_{\psi\in \B(p')}\psi$.
  Clearly, $\zeta$ belongs to $\B^+(p)$ (by construction---notice that $\zeta$
  is well defined due to the finiteness of $\B(p')$) and
  $\sat{\zeta}\subseteq\sat{\phi}$ (because $\bigcap_{\psi\in\B(p')}\sat{\psi}
  \subseteq \sat{\varphi}$).
  Hence, $\bigcap_{\psi\in\B(p)}\sat{\psi}\subseteq \sat{\zeta} \subseteq
  \sat{\phi}$ holds.
\end{compactitem}

Finally, we show that $\B$ is monotonic (property~\ref{item:Bmonotonicity}).
Consider $p,q\in P$, with $\Lset(p)\subseteq\Lset(q)$.
We want to show that $\phi\in\B(p)$ implies $\phi\in\B(q)$, for each $\phi$.
Firstly, we observe that, by $\Lset(p)\subseteq\Lset(q)$ and
Theorem~\ref{theo:branching-semantics}, $p \< q$ holds. Thus, we have that $I(p) = I(q)$ and,
for each $a \in Act$ and $p' \in P$ with $p \stackrel{a}{\rightarrow} p'$,
there exists some $q' \in P$ such that $q \stackrel{a}{\rightarrow} q'$
and $p' \< q'$.
We also observe that $\B^-(p) = \B^-(q)$, since $I(p) = I(q)$.
In order to show that $\B(p) \subseteq \B(q)$, we proceed by induction on the depth of $p$.
If $I(p) = \emptyset$, then $I(q) = \emptyset$ as well. Thus, we have that
$\B^+(p) = \B^+(q) = \{ \true \}$, and the thesis follows.
Otherwise ($I(p) \neq \emptyset$), let us consider a formula
$\phi \in \B(p)$.
If $\phi \in \B^-(p)$, then the claim follows from $\B^-(p) = \B^-(q) \subseteq
\B(q)$.
If $\phi = \true$, then, by definition of $\B^+$, we have $\phi \in \B^+(q)
\subseteq \B(q)$.
Finally, if $\phi = \exist{a}\varphi \in \B^+(p)$, then, by definition of
$\B^+$, there exist $p' \in P$, with $p\stackrel{a}{\rightarrow}p'$, such that
$\varphi=\bigwedge_{\psi\in \Psi}\psi$ for some $\Psi \subseteq \boundfun(p')$.
This implies the existence of some $q' \in P$ such that $q \stackrel{a}{\rightarrow} q'$
and $p' \< q'$ (and therefore $\Lset(p') \subseteq \Lset(q')$ by Theorem~\ref{theo:branching-semantics}).
By the inductive hypothesis, $\boundfun(p') \subseteq \boundfun(q')$ holds as
well, which means that $\Psi \subseteq \boundfun(q')$. Hence, we have that
$\exist{a}\varphi \in \B^+(q) \subseteq \B(q)$.
\end{proof}

\newcounter{lemanticharready}
\setcounter{lemanticharready}{\value{lemma}}
\newcommand{\lemanticharready}{
 Let $X \in \bspectrumSet$.
 For each $p \in P$ and $\chiX{\genericSemanticsX}(p)$ characteristic within
 $\Lset_\genericSemanticsX$ for $p$, there exists a formula in
 $\Lset_\genericSemanticsX$, denoted by $\barchiX{\genericSemanticsX}(p)$, such
 that either
 \begin{itemize}
 \item $\sat{\bar\chi(p)}=P \setminus \sat{\chi(p)}$, or
 \item
%
   $p \not \in \sat{\barchiX{\genericSemanticsX}(p)}$ and
%
%
   $\{ p' \in P \mid p' \not \<_\genericSemanticsX p \} \subseteq \sat{\barchiX{\genericSemanticsX}(p)} $.
 \end{itemize}
}
\begin{lemma}\label{lem:anti-char_branching}
  \lemanticharready
\end{lemma}

%

\begin{proof}

  We detail the case of ready simulation only (the proof for the other cases can
  be found in~\ref{sec:full-proof-existence-of-chi}).
%
For the sake of clarity we recall here the definition of $\barchiRS$
from \tablename~\ref{tab:syntax_monotonicity_finite_char_branching_spectrum}:

{\centering
$\barchiRS(p)= \big(\bigvee_{a\in Act} \exist{a}\bigwedge_{p\stackrel{a}{\rightarrow}p'} \barchiRS(p')\big) \vee \bigvee_{a\in I(p)}\univer{a}\false$.

}

As we did for the proof of the previous lemma, we omit the subscript
$_{\readySim}$ with no risk of ambiguity, e.g., we write $\bar\chi$ (resp.,
$\equiv$) for $\barchiRS$ (resp., $\equiv_{\readySim}$).

Let us first show that for every $p\in P$ we have $p \not \in \sat{\bar \chi(p)}$. We proceed by induction on the depth of $p$. 
\begin{compactitem}
\item $I(p) = \emptyset$: we have $\bar
  \chi(p)=\bigvee_{a\in Act} \exist{a}\true$ (up to logical equivalence), and thus
  $p\notin\sat{\bar \chi(p)}$.

\item $I(p) \neq \emptyset$: obviously, $p\notin\sat{\univer{a}\false}$ holds for every $a\in I(p)$. Moreover, for every $a \in Act$ and every $p\stackrel{a}{\rightarrow}p''$, by the inductive hypothesis, $p'' \notin \sat{\bar \chi(p'')}$. Thus, $p''\notin \sat{\bigwedge_{p\stackrel{a}{\rightarrow}p'}\bar \chi(p')}$ and therefore  $p\notin \sat{\exist{a}\bigwedge_{p\stackrel{a}{\rightarrow}p'}\bar \chi(p')}$ for every $a \in Act$. Hence, we obtain that $p \not \in \sat{\bar \chi(p)}$.
\end{compactitem}

Now, let us show that $\{ p' \in P \mid p' \not \< p \} \subseteq \sat{\bar \chi(p)}$,
that is, $\sat{\bar \chi(p)}$ contains at least
the elements that are either strictly above $p$ or
incomparable with it.
The proof is by induction on the depth of $p$.
\begin{compactitem}
\item $I(p) = \emptyset$: we have that $\{ p' \in P \mid p' \not \< p \}=P\setminus\{p\in P\mid I(p)=\emptyset\}$ because in this case we have that $p \equiv \zeroProcess$ and thus $q\< \zeroProcess$ does not hold for any process $q$ with $I(q)\neq\emptyset$. It is easy to see that $P\setminus\{p\in P\mid I(p)=\emptyset\} \subseteq \sat{\bigvee_{a\in Act} \exist{a}\true}= \sat{\bar \chi(p)}$.

\item $I(p) \neq \emptyset$: let $q\not\< p$. Thus, either $I(q)\neq I(p)$ or there exists some $q'$, with $q\stackrel{a}{\rightarrow}q'$, such that, for every $p'$, $p\stackrel{a}{\rightarrow}p'$ implies $q'\not\< p'$. If it is the case that $I(q)\neq I(p)$, then either $q\in\sat{\exist{a}\true}$ holds for some $a\notin I(p)$, or $q\in\sat{\univer{a}\false}$ for some $a\in I(p)$. In either case, $q\in \sat{\bar \chi(p)}$ holds.

  Otherwise, if there exist $a\in Act$ and $q'\in P$, with $q\stackrel{a}{\rightarrow}q'$, such that $q'\not \< p'$ for every $p\stackrel{a}{\rightarrow}p'$, then, by the inductive hypothesis, $q'\in\sat{\bar \chi(p')}$ for every $p'$ such that $p\stackrel{a}{\rightarrow}p'$. Thus, $q'\in\sat{\bigwedge_{p\stackrel{a}{\rightarrow}p'}\bar \chi(p')}$ and therefore $q\in\sat{\exist{a}\bigwedge_{p\stackrel{a}{\rightarrow}p'}\bar \chi(p')}$. Hence, we conclude $q\in \sat{\bar \chi(p)}$.
  \qedhere
\end{compactitem}
\end{proof}

Finally, the following theorem states the main result of this section.
\begin{theorem}[Characterization by primality for the branching time
  spectrum] \label{theo:char_by_prim_branching_spectrum} Let $X \in \bspectrumSet$
  and $\phi \in \Lset_X$.
 Then, $\phi$ is consistent and prime if and only if $\phi$ is
 characteristic for some $p \in P$.
\end{theorem}
\begin{proof} The claim immediately follows from Theorem~\ref{theo:char_prime},
  Theorem~\ref{theo:decomp},
  Corollary~\ref{cor:decomposability_condition_for_spectrum},
  Lemma~\ref{lem:characterization_branching}, and Lemma~\ref{lem:anti-char_branching}.
 \end{proof}


\subsection{The linear time spectrum} \label{sec:linear-time-spectrum}
%

In this sub-section, we focus on the semantics in \emph{van Glabbeek's linear
  time spectrum},~\cite{dFGPR13,VanGlabbeek01}, and their corresponding logical
formalisms.

To begin with, we define sets $\genericSemanticsX(p)$, for $X \in \lspectrumSet$ (defined on page~\pageref{lspectrumSet})
and $p \in P$ (they were originally defined in~\cite{dFGPR13,VanGlabbeek01}).
Intuitively, $\genericSemanticsX(p)$ (with $\genericSemanticsX \in
\lspectrumSet$) establishes the granularity of \emph{process observations},
and thus the level of detail at which processes are compared.
For example, $\trace(p)$ contains all traces that $p$ is able to perform;
similarly, $\completeTrace(p)$ contains all traces that $p$ is able to
perform and that lead to a process where no action can be performed.
When we move up to semantics in the three diamonds (see
Figure~\ref{fig:VGspectrum}), things get a bit more involved: $\ready(p)$
contains pairs $\langle \tau, Y \rangle$ whose first element is a trace that
takes $p$ to a process $p'$ whose set of \emph{initials} (actions it can
perform) is exactly $Y$; on the contrary, $\langle \tau, Y \rangle \in
\failure(p)$ denotes the fact that $p$ reaches $p'$ through $\tau$ and $p'$ cannot
perform any action in $Y$; $\readyTrace(p)$ contains words over the alphabet
$(Act \cup \powerset{Act} )^*$, that is, $\sigma \in \readyTrace(p)$ is a finite
sequence $\sigma_1 \sigma_2 \ldots \sigma_k$ where each $\sigma_i$ is either an
action (\emph{action elements}) or a set of actions (\emph{set elements}):
action elements identify a trace $\tau$ that $p$ can perform, while set elements
represent initials of (some of) the processes reached while performing $\tau$; finally,
elements of $\failureTrace(p)$ differ from the ones in $\readyTrace(p)$ in that
set elements represent actions that cannot be performed (rather than initials).
Set $X(p)$, where $X$ is a semantics in the 2nd or the 3rd diamond, is defined
analogously to its counterpart in the 1st diamond (see
Figure~\ref{fig:VGspectrum}); the only difference is that set elements carry
information on traces that can/cannot be performed (2nd diamond) or on processes
that are/are not simulated (3rd diamond) rather than on actions that can/cannot
be performed.

The following observation will be useful later on, when we will finitely
characterize the logics characterizing the semantics in the linear time spectrum through some function $\B$.

\begin{remark} \label{rem:bounded-trace-length}
  We can safely reduce to words in $\readyTrace(p)$ and $\failureTrace(p)$ whose
  length is bounded by $2 \cdot \depth(p) + 1$.
  Indeed, the number of action elements occurring in any such word $\sigma$
  cannot be greater than the length of a longest trace $p$ can perform; moreover,
  two consecutive set elements can be suitably merged into one: for instance, word
  $a \{a\} \{b\} \in \failureTrace(p)$ says that $p$ can perform $a$ and reach
  process $p'$, which in turn can perform neither $a$ nor $b$; the same
  information is captured by $a \{a,b \}$, which is also an element of
  $\failureTrace(p)$.

  The same observation applies to corresponding sets for the semantics in the
  2nd and 3rd diamond, i.e., $\possibleFutureTrace(p)$,
  $\impossibleFutureTrace(p)$, $\possibleSimulationTrace(p)$, and
  $\impossibleSimulationTrace(p)$.
\end{remark}

Sets $X(p)$ (for $X \in \lspectrumSet$ and $p \in P$) are formally defined
below.
For every $p \in P$, $\eqClasspbisim$ is the equivalence class of $p$ with
respect to bisimulation equivalence, that is, $\eqClasspbisim = \{ q \in P \mid
q \equiv_{\bisim} p \}$; for every $Q \subseteq P$, we use $\eqClassQbisim$ to
denote the set of equivalence classes of processes in $Q$ with respect to
bisimulation equivalence, that is, $\eqClassQbisim = \{ \eqClasspbisim \mid p
\in Q \}$; moreover, we use $\Sequivclass{p}$ to denote the set of equivalence
classes of processes that are simulated by $p$, that is, $\Sequivclass{p} = \{
\eqClassqbisim \in \eqClassPbisim \mid q \<_\simulation p \}$.
Notice that \Sequivclass{p} is finite for all $p$.

\begin{description}

\item [trace (\trace):] $\trace(p) = \{ \varepsilon \} \cup \{ a\tau \mid
  \exists p' \in P \ . \ p~\stackrel{\aProcess}{\rightarrow}~p' \text{ and }
  \tau \in \trace(p') \}$;\footnote{$\trace(p)$ is defined in the same way as
    $\traces(p)$ (cf. equation~(\ref{eq:trace_definition}) at
    page~\pageref{eq:trace_definition}); we re-define it using a different
    notation to ease the reading and to be uniform with the notation used for
    the other linear time semantics.}

\item [complete trace (\completeTrace):] $\completeTrace(p) = \{ \tau \in
  \trace(p) \mid \exists p' . p \stackrel{\tau}{\rightarrow} p' \text{ and }
  p' \not \rightarrow
  \}$;

\item [ready (\ready):] $\ready(p) = \{ \langle \tau, Y \rangle \in Act^* \times
  \powerset{Act} \mid \ \exists p' \in P \ . \ p~\stackrel{\tau}{\rightarrow}~p'
  \text{ and } I(p') = Y \}$;

\item [ready trace (\readyTrace):] \hspace{1cm}

  \begin{compactitem}
  \item $p \readyarrow{\varepsilon} p$ for all $p \in P$,
  \item $p \readyarrow{Y} p$ for all $p \in P$ and $Y \subseteq Act$ such that
    $I(p) = Y$,
  \item if $p \xrightarrow{a} q$, then $p \readyarrow{a} q$, for all $p,q \in
    P$ and $a \in Act$,
  \item if $p \readyarrow{\sigma} q$ and $q \readyarrow{\rho} r$, then $p
    \readyarrow{\sigma \cdot \rho} r$, for all $p,q,r \in P$ and $\sigma, \rho
    \in (Act \cup \powerset{Act})^*$,
  \item $\readyTrace(p) = \{ \sigma \in (Act \cup \powerset{Act} )^* \mid
    \exists q \in P \ . \ p~\readyarrow{\sigma}~q \}$;
  \end{compactitem}

\item [failure (\failure):] $\failure(p) = \{ \langle \tau, Y \rangle \in Act^*
  \times \powerset{Act} \mid \ \exists p' \in P \ . \
  p~\stackrel{\tau}{\rightarrow}~p' \text{ and } I(p') \cap Y = \emptyset \}$;

\item [failure trace (\failureTrace):] \hspace{1cm}

  \begin{compactitem}


  \item $p \refusalarrow{\varepsilon} p$ for all $p \in P$,
  \item $p \refusalarrow{Y} p$ for all $p \in P$ and $Y \subseteq Act$ such that
    $I(p) \cap Y = \emptyset$,
  \item if $p \xrightarrow{a} q$, then $p \refusalarrow{a} q$ for all $p,q \in
    P$ and $a \in Act$,
  \item if $p \refusalarrow{\sigma} q$ and $q \refusalarrow{\rho} r$, then $p
    \refusalarrow{\sigma \cdot \rho} r$ for all $p,q,r \in P$ and $\sigma, \rho
    \in (Act \cup \powerset{Act})^*$,
  \item $\failureTrace(p) = \{ \sigma \in (Act \cup
    \powerset{Act} )^* \mid \exists q \in P \ . \ p~\refusalarrow{\sigma}~q \}$;

  \end{compactitem}

\item [possible future (\possibleFuture):] $\possibleFuture(p) = \{ \langle
  \tau, \Gamma \rangle \in Act^* \times \powerset{Act^*} \mid \ \exists p' \in P \
  . \ p~\stackrel{\tau}{\rightarrow}~p' \text{ and } \trace(p') = \Gamma \}$;

\item [possible-future trace (\possibleFutureTrace):] \hspace{1cm}

  \begin{compactitem}
  \item $p \possiblearrow{\varepsilon} p$ for all $p \in P$,
  \item $p \possiblearrow{\Gamma} p$ for all $p \in P$ and $\Gamma \subseteq
    Act^*$ such that $\trace(p) = \Gamma$,
  \item if $p \xrightarrow{a} q$, then $p \possiblearrow{a} q$, for all $p,q \in
    P$ and $a \in Act$,
  \item if $p \possiblearrow{\sigma} q$ and $q \possiblearrow{\rho} r$, then $p
    \possiblearrow{\sigma \cdot \rho} r$, for all $p,q,r \in P$ and $\sigma,
    \rho \in (Act \cup \powerset{Act^*})^*$,
  \item $\possibleFutureTrace(p) = \{ \sigma \in (Act \cup \powerset{Act^*} )^*
    \mid \exists q \in P \ . \ p~\possiblearrow{\sigma}~q \}$;
  \end{compactitem}

\item [impossible future (\impossibleFuture):] $\impossibleFuture(p) = \{
  \langle \tau, \Gamma \rangle \in Act^* \times \powerset{Act^*} \mid \ \exists p'
  \in P \ . \ p~\stackrel{\tau}{\rightarrow}~p' \text{ and } \trace(p') \cap
  \Gamma = \emptyset \}$;

\item [impossible-future trace (\impossibleFutureTrace):] \hspace{1cm}

  \begin{compactitem}
  \item $p \impossiblearrow{\varepsilon} p$ for all $p \in P$,
  \item $p \impossiblearrow{\Gamma} p$ for all $p \in P$ and $\Gamma \subseteq
    Act^*$ such that $\trace(p) \cap \Gamma = \emptyset$,
  \item if $p \xrightarrow{a} q$, then $p \impossiblearrow{a} q$, for all $p,q
    \in P$ and $a \in Act$,
  \item if $p \impossiblearrow{\sigma} q$ and $q \impossiblearrow{\rho} r$, then
    $p \impossiblearrow{\sigma \cdot \rho} r$, for all $p,q,r \in P$ and $\sigma,
    \rho \in (Act \cup \powerset{Act^*})^*$,
  \item $\impossibleFutureTrace(p) = \{ \sigma \in (Act \cup \powerset{Act^*} )^*
    \mid \exists q \in P \ . \ p~\impossiblearrow{\sigma}~q \}$.
  \end{compactitem}

\item [possible 2-simulation (\possibleSimulation):] $\possibleSimulation(p) =
  \{ \langle \tau, \mathbb P \rangle \in Act^* \times \powerset{\eqClassPbisim}
  \mid \ \exists p' \in P \ . \ p~\stackrel{\tau}{\rightarrow}~p' \text{ and }
  \Sequivclass{p'} = \mathbb P \}$;

\item [possible-2-simulation trace (\possibleSimulationTrace):] \hspace{1cm}

  \begin{compactitem}
  \item $p \possiblesimulationarrow{\varepsilon} p$ for all $p \in P$,
  \item $p \possiblesimulationarrow{\mathbb P} p$ for all $p \in P$ and $\mathbb
    P \subseteq \eqClassPbisim$ such that $\Sequivclass{p} = \mathbb P$,
  \item if $p \xrightarrow{a} q$, then $p \possiblesimulationarrow{a} q$, for
    all $p,q \in P$ and $a \in Act$,
  \item if $p \possiblesimulationarrow{\sigma} q$ and $q
    \possiblesimulationarrow{\rho} r$, then $p \possiblesimulationarrow{\sigma \cdot
      \rho} r$, for all $p,q,r \in P$ and $\sigma, \rho \in (Act \cup
    \powerset{\eqClassPbisim})^*$,
  \item $\possibleSimulationTrace(p) = \{ \sigma \in (Act \cup
    \powerset{\eqClassPbisim} )^* \mid \exists q \in P \ . \
    p~\possiblesimulationarrow{\sigma}~q \}$;
  \end{compactitem}

\item [impossible 2-simulation (\impossibleSimulation):]
  $\impossibleSimulation(p) = \{ \langle \tau, \mathbb P \rangle \in Act^* \times
  \powerset{\eqClassPbisim} \mid \ \exists p' \in P \ . \
  p~\stackrel{\tau}{\rightarrow}~p' \text{ and } \Sequivclass{p'} \cap \mathbb P =
  \emptyset \}$;

\item [impossible-2-simulation trace (\impossibleSimulationTrace):]
  \hspace{1cm}

  \begin{compactitem}
  \item $p \impossiblesimulationarrow{\varepsilon} p$ for all $p \in P$,
  \item $p \impossiblesimulationarrow{\mathbb P} p$ for all $p \in P$ and
    $\mathbb P \subseteq \eqClassPbisim$ such that $\Sequivclass{p} \cap \mathbb P =
    \emptyset$,
  \item if $p \xrightarrow{a} q$, then $p \impossiblesimulationarrow{a} q$, for
    all $p,q \in P$ and $a \in Act$,
  \item if $p \impossiblesimulationarrow{\sigma} q$ and $q
    \impossiblesimulationarrow{\rho} r$, then $p \impossiblesimulationarrow{\sigma
      \cdot \rho} r$, for all $p,q,r \in P$ and $\sigma, \rho \in (Act \cup
    \powerset{\eqClassPbisim})^*$,
  \item $\impossibleSimulationTrace(p) = \{ \sigma \in (Act \cup
    \powerset{\eqClassPbisim} )^* \mid \exists q \in P \ . \
    p~\impossiblesimulationarrow{\sigma}~q \}$.
  \end{compactitem}

\end{description}

\begin{definition}[Linear time semantic relations~\cite{dFGPR13,VanGlabbeek01}]
  \label{def:linear-semantics}
  For each $X \in \lspectrumSet$, $\<_X$ is defined as follows:

  {\centering $p\<_X q \Leftrightarrow X(p) \subseteq X(q)$,

  }

  \noindent
  while $\equiv_X$ is defined as

  {\centering $p \equiv_X q \Leftrightarrow p \<_X q$ and $q \<_X p
    \Leftrightarrow X(p) = X(q)$.

  }

%
%
%
%
%
%
%
%
%

\end{definition}

\noindent{\bf Linear time logics.}
Unlike the logics for branching time semantics, the linear time logics do not yield a
chain of strict inclusions; their expressive power is captured by a partial
ordering relation, as shown in \figurename~\ref{fig:VGspectrum} (right-hand
side).

The syntax defining the logics characterizing the semantics in van Glabbeek's
linear time spectrum is as follows.
Let us recall here the meaning of the following abbreviations:
\begin{itemize}
\item $\univer{a} \psi$ stands for $\neg \exist{a} \neg \psi$,
\item $\zeroFormula$ stands for $\bigwedge_{a \in Act} \univer{a} \false$,
\item $\exist{\tau}$ (with $\tau = a_1 a_2 \ldots a_k \in Act^*$) stands for
  $\exist{a_1}\exist{a_2} \ldots \exist{a_k}$, and
\item $\univer{\tau}$ (with $\tau = a_1 a_2 \ldots a_k \in Act^*$) stands for
  $\univer{a_1}\univer{a_2} \ldots \univer{a_k}$.
\end{itemize}
Additionally, for $Y \subseteq Act$ (resp., finite set $\Gamma \subseteq
Act^*$), we use $\existuniver{Y}$ (resp., $\existuniver{\Gamma}$) as an
abbreviation for $\bigwedge_{a \in Y} \exist{a}\true \wedge \bigwedge_{a \in Act
  \setminus Y} \univer{a}\false$ (resp., $\bigwedge_{\tau \in \Gamma}
\exist{\tau}\true \wedge \bigwedge_{\tau \in Act^*_{|_\Gamma} \setminus \Gamma}
\univer{\tau}\false$), where $Act^*_{|_\Gamma} = \{ \tau \in Act^* \mid |\tau|
\leq (\max_{\tau' \in \Gamma} |\tau'| )+ 1 \}$

\begin{definition}[Syntax~\cite{dFGPR13,VanGlabbeek01}]\label{def:syntax_linear_spectrum}
  Let $X \in \lspectrumSet$, $\Lset_X$ is the language defined by the
  corresponding grammar among the following ones:

$\Lset_\trace$:

{\centering
\begin{tabular}{lll}
$\phi_\trace$ & $::=$ & $\true \mid \false \mid \phi_\trace \wedge \phi_\trace
\mid \phi_\trace \vee \phi_\trace \mid \psi_\trace$ \\

$\psi_\trace$ & $::=$ & $\true \mid \exist{a}\psi_\trace$,
\end{tabular}

}

$\Lset_\completeTrace$:

{\centering
\begin{tabular}{lll}
$\phi_\completeTrace$ & $::=$ & $\true \mid \false \mid \phi_\completeTrace \wedge \phi_\completeTrace
\mid \phi_\completeTrace \vee \phi_\completeTrace \mid \psi_\completeTrace$ \\

$\psi_\completeTrace$ & $::=$ & $\true \mid \zeroFormula \mid \exist{a}\psi_\completeTrace$,
\end{tabular}

}

$\Lset_\failure$:

{\centering
\begin{tabular}{lll}
%

$\phi_\failure$ & $::=$ & $\true \mid \false \mid \phi_\failure \wedge \phi_\failure
\mid \phi_\failure \vee \phi_\failure \mid \psi_\failure$ \\

$\psi_\failure$ & $::=$ & $\true \mid \exist{a}\psi_\failure \mid \gamma_\failure$ \\

$\gamma_\failure$ & $::=$ & $\univer{a}\false \mid \gamma_\failure \wedge \gamma_\failure$,

\end{tabular}

}

$\Lset_\ready$:

{\centering
\begin{tabular}{lll}
$\phi_\ready$ & $::=$ & $\true \mid \false \mid \phi_\ready \wedge \phi_\ready
\mid \phi_\ready \vee \phi_\ready \mid \psi_\ready$ \\

$\psi_\ready$ & $::=$ & $\true \mid \existuniver{Y} \mid \exist{a}\psi_\ready$,
\end{tabular}

}

$\Lset_\failureTrace$:

{\centering
\begin{tabular}{lll}
$\phi_\failureTrace$ & $::=$ & $\true \mid \false \mid \phi_\failureTrace \wedge \phi_\failureTrace
\mid \phi_\failureTrace \vee \phi_\failureTrace \mid \psi_\failureTrace$ \\

$\psi_\failureTrace$ & $::=$ & $\true \mid \univer{a}\false \mid \exist{a}\psi_\failureTrace \mid \univer{a}\false \wedge
                                 \psi_\failureTrace$,

\end{tabular}

}

$\Lset_\readyTrace$:

{\centering
\begin{tabular}{lll}
$\phi_\readyTrace$ & $::=$ & $\true \mid \false \mid \phi_\readyTrace \wedge \phi_\readyTrace
\mid \phi_\readyTrace \vee \phi_\readyTrace \mid \psi_\readyTrace$ \\

$\psi_\readyTrace$ & $::=$ & $\true \mid
                        \existuniver{Y} \wedge \psi_\readyTrace \mid
                             \exist{a}\psi_\readyTrace$,

\end{tabular}

}

$\Lset_\impossibleFuture$:

{\centering
\begin{tabular}{lll}
$\phi_\impossibleFuture$ & $::=$ & $\true \mid \false \mid \phi_\impossibleFuture \wedge \phi_\impossibleFuture
\mid \phi_\impossibleFuture \vee \phi_\impossibleFuture \mid \psi_\impossibleFuture$ \\

$\psi_\impossibleFuture$ & $::=$ & $\true \mid
                                   \exist{a}\psi_\impossibleFuture \mid \gamma_\impossibleFuture$ \\

$\gamma_\impossibleFuture$ & $::=$ & $\univer{\tau}\false \mid \gamma_\impossibleFuture \wedge \gamma_\impossibleFuture$,
\end{tabular}

}

$\Lset_\possibleFuture$:

{\centering
\begin{tabular}{lll}

$\phi_\possibleFuture$ & $::=$ & $\true \mid \false \mid \phi_\possibleFuture \wedge \phi_\possibleFuture
\mid \phi_\possibleFuture \vee \phi_\possibleFuture \mid \psi_\possibleFuture$ \\

$\psi_\possibleFuture$ & $::=$ & $\true \mid \existuniver{\Gamma} \mid
                             \exist{a}\psi_\possibleFuture$,

\end{tabular}

}

$\Lset_\impossibleFutureTrace$:

{\centering
\begin{tabular}{lll}

$\phi_\impossibleFutureTrace$ & $::=$ & $\true \mid \false \mid \phi_\impossibleFutureTrace \wedge \phi_\impossibleFutureTrace
\mid \phi_\impossibleFutureTrace \vee \phi_\impossibleFutureTrace \mid \psi_\impossibleFutureTrace$ \\

$\psi_\impossibleFutureTrace$ & $::=$ & $\true \mid \univer{\tau}\false \mid
                             \exist{a}\psi_\impossibleFutureTrace \mid \univer{\tau}\false \wedge
                             \psi_\impossibleFutureTrace$,

\end{tabular}

}

$\Lset_\possibleFutureTrace$:

{\centering
\begin{tabular}{lll}

$\phi_\possibleFutureTrace$ & $::=$ & $\true \mid \false \mid \phi_\possibleFutureTrace \wedge \phi_\possibleFutureTrace
\mid \phi_\possibleFutureTrace \vee \phi_\possibleFutureTrace \mid \psi_\possibleFutureTrace$ \\

$\psi_\possibleFutureTrace$ & $::=$ & $\true \mid \existuniver{\Gamma} \wedge \psi_\possibleFutureTrace \mid
                             \exist{a}\psi_\possibleFutureTrace$,

\end{tabular}

}

$\Lset_\impossibleSimulation$:

{\centering
\begin{tabular}{lll}
$\phi_\impossibleSimulation$ & $::=$ & $\true \mid \false \mid \phi_\impossibleSimulation \wedge \phi_\impossibleSimulation
\mid \phi_\impossibleSimulation \vee \phi_\impossibleSimulation \mid \psi_\impossibleSimulation$ \\

$\psi_\impossibleSimulation$ & $::=$ & $\true \mid
                                   \exist{a}\psi_\impossibleSimulation \mid \gamma_\impossibleSimulation$ \\

$\gamma_\impossibleSimulation$ & $::=$ & $\true \mid \false \mid \univer{a}\gamma_\impossibleSimulation \mid \gamma_\impossibleSimulation \wedge \gamma_\impossibleSimulation \mid \gamma_\impossibleSimulation \vee \gamma_\impossibleSimulation$,
\end{tabular}

}

$\Lset_\possibleSimulation$:

{\centering
\begin{tabular}{lll}

$\phi_\possibleSimulation$ & $::=$ & $\true \mid \false \mid \phi_\possibleSimulation \wedge \phi_\possibleSimulation
\mid \phi_\possibleSimulation \vee \phi_\possibleSimulation \mid \psi_\possibleSimulation$ \\

$\psi_\possibleSimulation$ & $::=$ & $\true \mid
                                     \exist{a}\psi_\possibleSimulation \mid \gamma^\exists_\possibleSimulation \wedge \gamma^\forall_\possibleSimulation$ \\

  $\gamma^\exists_\possibleSimulation$ & $::=$ & $\true \mid \false \mid \exist{a}\gamma^{\exists}_\possibleSimulation \mid \gamma^{\exists}_\possibleSimulation \wedge \gamma^{\exists}_\possibleSimulation \mid \gamma^{\exists}_\possibleSimulation \vee \gamma^{\exists}_\possibleSimulation$ \\

  $\gamma^\forall_\possibleSimulation$ & $::=$ & $\true \mid \false \mid \univer{a}\gamma^{\forall}_\possibleSimulation \mid \gamma^{\forall}_\possibleSimulation \wedge \gamma^{\forall}_\possibleSimulation \mid \gamma^{\forall}_\possibleSimulation \vee \gamma^{\forall}_\possibleSimulation$,

\end{tabular}

}

$\Lset_\impossibleSimulationTrace$:

{\centering
\begin{tabular}{lll}
$\phi_\impossibleSimulation$ & $::=$ & $\true \mid \false \mid \phi_\impossibleSimulation \wedge \phi_\impossibleSimulation
\mid \phi_\impossibleSimulation \vee \phi_\impossibleSimulation \mid \psi_\impossibleSimulation$ \\

$\psi_\impossibleSimulation$ & $::=$ & $\true \mid
                                   \exist{a}\psi_\impossibleSimulation \mid \gamma_\impossibleSimulation \wedge \psi_\impossibleSimulation$ \\

$\gamma_\impossibleSimulation$ & $::=$ & $\true \mid \false \mid \univer{a}\gamma_\impossibleSimulation \mid \gamma_\impossibleSimulation \wedge \gamma_\impossibleSimulation \mid \gamma_\impossibleSimulation \vee \gamma_\impossibleSimulation$,
\end{tabular}

}

$\Lset_\possibleSimulationTrace$:

{\centering
\begin{tabular}{lll}

$\phi_\possibleSimulation$ & $::=$ & $\true \mid \false \mid \phi_\possibleSimulation \wedge \phi_\possibleSimulation
\mid \phi_\possibleSimulation \vee \phi_\possibleSimulation \mid \psi_\possibleSimulation$ \\

$\psi_\possibleSimulation$ & $::=$ & $\true \mid
                                     \exist{a}\psi_\possibleSimulation \mid \gamma^\exists_\possibleSimulation \wedge \gamma^\forall_\possibleSimulation \wedge \psi_\possibleSimulation$ \\

  $\gamma^\exists_\possibleSimulation$ & $::=$ & $\true \mid \false \mid \exist{a}\gamma^{\exists}_\possibleSimulation \mid \gamma^{\exists}_\possibleSimulation \wedge \gamma^{\exists}_\possibleSimulation \mid \gamma^{\exists}_\possibleSimulation \vee \gamma^{\exists}_\possibleSimulation$ \\

  $\gamma^\forall_\possibleSimulation$ & $::=$ & $\true \mid \false \mid \univer{a}\gamma^{\forall}_\possibleSimulation \mid \gamma^{\forall}_\possibleSimulation \wedge \gamma^{\forall}_\possibleSimulation \mid \gamma^{\forall}_\possibleSimulation \vee \gamma^{\forall}_\possibleSimulation$.

\end{tabular}

}

\end{definition}

\begin{remark}
  $\Lset_\ready$ (resp., $\Lset_\readyTrace$, $\Lset_\possibleFuture$,
  $\Lset_\possibleFutureTrace$) is strictly more expressive than $\Lset_\failure$
  (resp., $\Lset_\failureTrace$, $\Lset_\impossibleFuture$,
  $\Lset_\impossibleFutureTrace$).
  However, we notice that, unlike the other cases, the embedding
  of the latter into the former is not succinct in general: translating
  $\Lset_\failure$ formulae (resp., $\Lset_\failureTrace$) into $\Lset_\ready$
  (resp., $\Lset_\readyTrace$) ones might cause an exponential growth (in the size
  of the set of actions) of the formula size; the situation is even worse when
  translating $\Lset_\impossibleFuture$ formulae (resp.,
  $\Lset_\impossibleFutureTrace$) into $\Lset_\possibleFuture$ (resp.,
  $\Lset_\possibleFutureTrace$) ones, which might cause a doubly exponential
  growth (in the size of the set of actions and in the size of the formula).
  For instance, formula $\bigwedge_{a \in Y}\univer{a}\false \in \Lset_\failure$
  (for $Y \subseteq Act$) is translated into formula $\bigvee_{Y' \subseteq Act
    \setminus Y} \existuniver{Y'} = \bigvee_{Y' \subseteq Act \setminus
    Y}(\bigwedge_{a \in Y'} \exist{a}\true \wedge \bigwedge_{a \in Act \setminus Y'}
  \univer{a}\false)\in \Lset_\ready$, whose size is exponential in the size of
  $Act$.
%
\end{remark}

\begin{remark}
%
%
  We notice that we allow for larger languages than the ones proposed
  in~\cite{dFGPR13,VanGlabbeek01}, e.g., we allow for the use of $\wedge$ and
  $\vee$ as outermost operator.
  However, it is not difficult to see that our logics do not distinguish more
  processes than the original ones defined~in \cite{dFGPR13,VanGlabbeek01}.
%
%
  We need such extended languages because we use $\vee$ as outermost operator in
  order to define the formula $\bar\chi$ (in Lemma~\ref{lem:anti-char_linear}) and
  we need $\wedge$ to apply
  Corollary~\ref{cor:decomposability_condition_for_spectrum} (in the proof of
  Theorem~\ref{theo:char_by_prim_linear_spectrum}).
\end{remark}

The satisfaction relation and semantics associated to logics for the linear time
semantics are the same as the ones for logics for branching time semantics (see
Definition~\ref{def:semantics_branching}).

The following well-known  theorem is the linear time counterpart of
Theorem~\ref{theo:branching-semantics} for branching time semantics.
It states the relationship between logics and process semantics that allows us
to use our general results about logically characterized semantics.
\begin{theorem}[Logical characterization of linear time semantics~\cite{dFGPR13,VanGlabbeek01}]\label{theo:linear-semantics} 
For each $X\in\lspectrumSet$ and for all $p,q \in P$, $p\<_Xq$ iff $\Lset_X(p)\subseteq\Lset_X(q)$.
\end{theorem}

We proceed now to prove our characterization by primality result for the above defined
logics.
Following the same approach adopted in the previous section for branching time
semantics, we show that also the logics considered in this section meet the
conditions of Corollary~\ref{cor:decomposability_condition_for_spectrum}, that
is, they are finitely characterized by some monotonic $\B$, and for each
$\chi(p)$ there exists a formula $\bar\chi(p)$ such that $p \notin
\sat{\bar\chi(p)}$ and $\{ q \in P \mid \Lset(q) \not\subseteq \Lset(p) \}
\subseteq \sat{\bar\chi(p)}$.

An elucidation is in order: even though, according to
Corollary~\ref{cor:decomposability_condition_for_spectrum}, logics are required
to feature (arbitrary) disjunction, we have already observed
(Remark~\ref{rem:disjunction} at page~\pageref{rem:disjunction}) that the
Boolean connective $\wedge$ plays a minor role in (the proof of)
Corollary~\ref{cor:decomposability_condition_for_spectrum}, and thus we can use
it to deal with the logics studied in this section, that allow for a limited use
of such a connective.

\begin{table}[t!]

  \footnotesize

  \centering

  \scalebox{.9}{
  $\begin{array}[h]{|l|l|}


     \hline

     \multirow{2}{*}{\genericSemanticsX=\trace}

     & \traceFinite(p) = \trace(p) \hspace{25mm} \text{ (as } \trace(p) \text{
       is already finite)} \\

     & \getFormulatrace{\tau} = \exist{\tau}\true \\

\hline

     \multirow{2}{*}{\genericSemanticsX=\completeTrace}

     & \completeTraceFinite(p) = \completeTrace(p) \hspace{22mm} \text{ (as }
       \completeTrace(p) \text{ is already finite)} \\

     & \getFormulacompleteTrace{\tau} = \exist{\tau}\zeroFormula \\

\hline

     \multirow{2}{*}{\genericSemanticsX=\failure}

     & \failureFinite(p) = \failure(p) \hspace{26mm} \text{ (as } \failure(p)
       \text{ is already finite)} \\

     & \getFormulafailure{\langle \tau, Y \rangle} = \exist{\tau} \bigwedge_{a
       \in Y} \univer{a}\false \\

\hline

     \multirow{2}{*}{\genericSemanticsX=\ready}

     & \readyFinite(p) = \ready(p) \hspace{25.5mm} \text{ (as } \ready(p) \text{
       is already finite)} \\

     & \getFormulaready{\langle \tau, Y \rangle} = \exist{\tau} \existuniver{Y}
     \\

\hline

    \multirow{2}{*}[-1em]{\genericSemanticsX=\failureTrace}

     & \failureTraceFinite(p) = \smallset{\failureTrace(p)}{2 \cdot \depth(p) +
       1} \\

     & \getFormulafailureTrace{\sigma} = \left\{
    \begin{array}[h]{ll} \true & \sigma = \varepsilon \\
      \exist{a}\getFormulafailureTrace{\sigma'} & \sigma = a\sigma', a \in Act,
                                                  \sigma' \in (Act \cup \powerset{Act} )^* \\ \bigwedge_{a \in
      Y}{\univer{a}\false} \wedge \getFormulafailureTrace{\sigma'} & \sigma =
                                                                     Y\sigma', Y \in \powerset{Act}, \sigma' \in (Act \cup \powerset{Act} )^*
    \end{array} \right. \\

\hline

    \multirow{2}{*}[-1em]{\genericSemanticsX=\readyTrace}

     & \readyTraceFinite(p) = \smallset{\readyTrace(p)}{2 \cdot \depth(p) + 1}
     \\

     & \getFormulareadyTrace{\sigma} = \left\{
    \begin{array}[h]{ll} \true & \sigma = \varepsilon \\

      \exist{a}\getFormulareadyTrace{\sigma'} & \sigma = a\sigma', a \in Act,
                                                \sigma' \in (Act \cup \powerset{Act} )^* \\

      \existuniver{Y} \wedge \getFormulareadyTrace{\sigma'} & \sigma = Y\sigma',
                                                              Y \in \powerset{Act}, \sigma' \in (Act \cup \powerset{Act} )^*
    \end{array} \right.
  \\

\hline

    \multirow{2}{*}{\genericSemanticsX=\impossibleFuture}

     & \impossibleFutureFinite(p) = \impossibleFuture(p) \cap
       (\smallset{Act^*}{\depth(p)+1} \times \powerset{\smallset{Act^*}{\depth(p)+1}})
%
%
 \\

    & \getFormulaimpossibleFuture{\langle \tau, \Gamma \rangle} = \exist{\tau}
      \bigwedge_{\tau' \in \Gamma} \univer{\tau'}\false \\

\hline

    \multirow{2}{*}{\genericSemanticsX=\possibleFuture}

     & \possibleFutureFinite(p) = \possibleFuture(p) \hspace{22.6mm} \text{ (as
       } \possibleFuture(p) \text{ is already finite)} \\

    & \getFormulapossibleFuture{\langle \tau, \Gamma \rangle} = \exist{\tau}
      \existuniver{\Gamma} \\

\hline

    \multirow{2}{*}[-1em]{\genericSemanticsX=\impossibleFutureTrace}

     & \impossibleFutureTraceFinite(p) = \impossibleFutureTrace(p) \cap
       \smallset{[(Act \cup \powerset{\smallset{Act^*}{\depth(p)+1}})^*]}{2 \cdot
       \depth(p) + 1} \\

    & \getFormulaimpossibleFutureTrace{\sigma} = \left\{
    \begin{array}[h]{ll}
      \true & \sigma = \varepsilon \\

      \exist{a}\getFormulaimpossibleFutureTrace{\sigma'} & \sigma = a\sigma', a \in Act, \sigma' \in (Act \cup
                                         \powerset{Act^*} )^* \\

      \bigwedge_{\tau \in \Gamma}{\univer{\tau}\false} \wedge
      \getFormulaimpossibleFutureTrace{\sigma'} & \sigma = \Gamma\sigma', \Gamma \in \powerset{Act^*}, \sigma' \in (Act \cup \powerset{Act^*}
                                )^*
    \end{array}
  \right.
\\

\hline

    \multirow{2}{*}[-1em]{\genericSemanticsX=\possibleFutureTrace}

     & \possibleFutureTraceFinite(p) = \smallset{\possibleFutureTrace(p)}{2
       \cdot \depth(p) + 1} \\

    & \getFormulapossibleFutureTrace{\sigma} = \left\{
    \begin{array}[h]{ll}
      \true & \sigma = \varepsilon \\

      \exist{a}\getFormulapossibleFutureTrace{\sigma'} & \sigma = a\sigma', a \in Act, \sigma' \in (Act \cup
                                        \powerset{Act^*} )^* \\

      \existuniver{\Gamma} \wedge
      \getFormulapossibleFutureTrace{\sigma'} & \sigma = \Gamma\sigma', \Gamma \in \powerset{Act^*}, \sigma' \in (Act \cup \powerset{Act^*}
                                )^*
    \end{array}
  \right.
\\

\hline

    \multirow{2}{*}{\genericSemanticsX=\impossibleSimulation}

     & \impossibleSimulationFinite(p) = \impossibleSimulation(p) \cap
       (\smallset{Act^*}{\depth(p)+1} \times
       \powerset{\eqClassbisim{\smallset{P}{\depth(p)+1}}})
%
%
 \\

     & \getFormulaimpossibleSimulation{\langle \tau, \mathbb P \rangle} =
       \exist{\tau} \bigwedge_{\eqClassbisim{p'} \in \mathbb P} \neg \chiS(p') \\


\hline

    \multirow{3}{*}[-.5em]{\genericSemanticsX=\possibleSimulation}

     & \possibleSimulationFinite(p) = \possibleSimulation(p) \hspace{22.6mm}
       \text{ (as } \possibleSimulation(p) \text{ is already finite)}\\

     & \getFormulapossibleSimulation{\langle \tau, \mathbb P \rangle} =
       \exist{\tau} \alpha(\mathbb P )

\hfill

       \begin{array}[t]{r@{\hspace{1mm}}l@{\hspace{0mm}}l@{\hspace{0mm}}l}
         \text{with}
         & \alpha (\mathbb P) & = & \bigwedge_{\eqClassbisim{p'} \in \mathbb P}
                                    \chiS(p') \wedge \bigvee_{\eqClassbisim{p'} \in \mathbb P} \simulatedby{p'} \\
         \text{and}
         & \simulatedby{p'} & = & \bigwedge_{a \in Act}\univer{a} \bigvee_{p'
                                 \stackrel{a}{\rightarrow} p''}\simulatedby{p''}
       \end{array} \\

\hline

    \multirow{2}{*}[-1em]{\genericSemanticsX=\impossibleSimulationTrace}

     & \impossibleSimulationTraceFinite(p) = \impossibleSimulationTrace(p) \cap
       \smallset{[(Act \cup \powerset{\eqClassbisim{\smallset{P}{\depth(p)+1}}})^*]}{2
       \cdot \depth(p) + 1} \\

    & \getFormulaimpossibleSimulationTrace{\sigma} = \left\{
    \begin{array}[h]{ll}
      \true & \sigma = \varepsilon \\

      \exist{a}\getFormulaimpossibleSimulationTrace{\sigma'} & \sigma = a\sigma', a \in Act, \sigma' \in (Act \cup
                                         \powerset{\eqClassPbisim} )^* \\

      \bigwedge_{\eqClasspbisim \in \mathbb P}{\neg \chiS(p)} \wedge
      \getFormulaimpossibleSimulationTrace{\sigma'} & \sigma = \mathbb P\sigma',
                                                      \mathbb P \in \powerset{\eqClassPbisim}, \sigma' \in (Act \cup \powerset{\eqClassPbisim} )^*
    \end{array}
  \right.
 \\

\hline

    \multirow{2}{*}[-1em]{\genericSemanticsX=\possibleSimulationTrace}

     & \possibleSimulationTraceFinite(p) =
       \smallset{\possibleSimulationTrace(p)}{2 \cdot \depth(p) + 1} \\

    & \getFormulapossibleSimulationTrace{\sigma} = \left\{
    \begin{array}[h]{ll}
      \true & \sigma = \varepsilon \\

      \exist{a}\getFormulapossibleSimulationTrace{\sigma'} & \sigma = a\sigma', a \in Act, \sigma' \in (Act \cup
                                        \powerset{\eqClassPbisim} )^* \\

      \alpha (\mathbb P) \wedge \getFormulapossibleSimulationTrace{\sigma'} &
                                                                              \sigma = \mathbb P\sigma', \mathbb P \in \powerset{\eqClassPbisim}, \sigma' \in
                                                                              (Act \cup \powerset{\eqClassPbisim} )^*

    \end{array}
  \right. \\

\hline


  \end{array}$
} 

  \caption{Instantiations of $\genericSemanticsFiniteX(\cdot)$ and
    \getFormula{\genericSemanticsX}{\cdot} for all $X \in \lspectrumSet$.
    For a set $S$ and $n \in \mathbb N$, we let $\smallset{S}{n} = \{ x \in S
    \mid |x| \leq n \}$ (we adopt the convention that $|x| = \depth(x)$, for $x \in
    P$).}

  \label{tab:linearB}
\end{table}

In order to finitely characterize the logics for the semantics in the linear
time spectrum, we base our definition of $\B_{\genericSemanticsX}(p)$ on
suitably defined finite versions $\genericSemanticsFiniteX(p)$ of sets
$\genericSemanticsX(p)$ (with $\genericSemanticsX \in \lspectrumSet$ and $p \in
P$) that carry all significant pieces of information, in the sense of
Proposition~\ref{prop:X-and-Xfinite}.
Thus, we define $\B_{\genericSemanticsX}$ as follows: for all
$\genericSemanticsX \in \lspectrumSet$ and $p \in P$

{\centering

  $\B_\genericSemanticsX(p) = \{ \true \} \cup \{
  \getFormula{\genericSemanticsX}{x} \mid x \in \genericSemanticsFiniteX(p) \}$

}

\noindent
where $\getFormula{\genericSemanticsX}{x}$ and $\genericSemanticsFiniteX(p)$ are
instantiated as in Table~\ref{tab:linearB}.
Intuitively, $\getFormula{\genericSemanticsX}{x}$ is a logical characterization
of a process observation $x$ according to semantics $\genericSemanticsX$, as
stated by Proposition~\ref{prop:p-sat-formula-x-for-x-in-Xp} below.
Thus, $\B_\genericSemanticsX(p)$ contains, besides $\true$, a formula for each
element of $\genericSemanticsFiniteX(p)$, meaning that its finiteness
immediately follows from the one of $\genericSemanticsFiniteX(p)$.
The construction of the finite sets $\genericSemanticsFiniteX(p)$ uses the following
notation: $\setWithRelation{S}{\sim}{n} = \{ x \in S \mid |x| \sim n \}$ for
every set $S$, $n \in \mathbb N$, and $\sim {\in}~\{ \leq, = \}$ (we adopt the
convention that $|x| = \depth(x)$, for $x \in P$), and is based on the following
considerations:
\begin{enumerate}[label={\it (\alph*)}]
\item $Act$ is finite;
\item the set of traces that a process can perform is finite;
\item \label{item:finite-set-of-failure-trace} even though the set of traces
  that a process $p$ cannot perform is infinite, it is characterized by the finite
  set of minimal traces that $p$ cannot perform, and the length of the longest
  trace in such a characterizing set is $\depth(p)+1$ (e.g., if a process $p$
  cannot perform $a$, then it clearly cannot perform any trace that starts with
  $a$; even if there are infinitely many such traces, it is enough to keep track
  of the fact that $a$ is a minimal trace that cannot be performed by $p$); thus,
  using the notation introduced above, we can focus on
  $\smallset{Act^*}{\depth(p)+1} = \{ \tau \in Act^* \mid |\tau| \leq \depth(p)+1
  \}$ (rather than full $Act^*$) when we need to characterize traces that $p$ can
  or cannot do;

\item even though the set of processes that are simulated by a process $p$ is
  infinite, it is actually finite up to bisimilarity (i.e., \Sequivclass{p} is
  finite);

\item the set of processes that are not simulated by a process $p$ is infinite;
  however, using an argument similar to the one from
  item~\ref{item:finite-set-of-failure-trace}, it is possible to identify a finite
  (up to bisimilarity) set of minimal (wrt. depth) processes that characterizes
  it;
%
%
%
%
%
  for a process $p$, such a set is defined as

  {\centering

%
    $\smallset{P}{\depth(p)+1} = \{ p' \in P \mid \depth(p') \leq \depth(p)+1
    \}$;

  }

\item \label{item:linear_word_bound} as already observed in Remark~\ref{rem:bounded-trace-length}, as far as
  sets $\failureTrace(p)$, $\readyTrace(p)$, $\impossibleFutureTrace(p)$,
  $\possibleFutureTrace(p)$, $\impossibleSimulationTrace(p)$, and
  $\possibleSimulationTrace(p)$ are concerned, we can restrict ourselves to considering words whose length
  is bounded by $2 \cdot \depth(p) + 1$.
%
\end{enumerate}

\newcounter{lemcharacterizationLinear}
\setcounter{lemcharacterizationLinear}{\value{lemma}}
\newcommand{\lemcharacterizationLinear}{
  Let $\genericSemanticsX \in \lspectrumSet$. $\Lset_{\genericSemanticsX}$ is
  finitely characterized by $\B_{\genericSemanticsX}$, for some monotonic
  $\B_{\genericSemanticsX}$.
}

The following two results follow from the definition of the semantics
$\genericSemanticsX$ and of \getFormula{\genericSemanticsX}{x} in
Table~\ref{tab:linearB}, and their proofs are omitted.

\begin{proposition} \label{prop:X-and-Xfinite}
  For all $\genericSemanticsX \in \lspectrumSet$ and $p,p' \in P$, the following
  statements are equivalent:

  \begin{enumerate}[label=\it{(\alph*)}]
  \item \label{item:Xp-subset-Xp} $\genericSemanticsX(p) \subseteq
    \genericSemanticsX(p')$,
  \item \label{item:Xfinitep-subset-Xfinitep} $\genericSemanticsFiniteX(p)
    \subseteq \genericSemanticsFiniteX(p')$, and
  \item \label{item:Xfinitep-subset-Xp} $\genericSemanticsFiniteX(p) \subseteq
    \genericSemanticsX(p')$.
  \end{enumerate}
\end{proposition}

\begin{proposition} \label{prop:p-sat-formula-x-for-x-in-Xp}
  For all $\genericSemanticsX \in \lspectrumSet$, $p \in P$, and $x \in
  \bigcup_{q \in P} \genericSemanticsFiniteX(q)$, we have:

  {\centering

    $x \in \genericSemanticsX(p)$ if and only if $p \in
    \sat{\getFormula{\genericSemanticsX}{x}}$.

  }

\end{proposition}

\begin{lemma} \label{lem:characterization_linear}
  \lemcharacterizationLinear
\end{lemma}
\begin{proof}
  For every $\genericSemanticsX \in \lspectrumSet$ and $p \in P$, we show that,
  \begin{enumerate}[label=(\emph{\roman*})]
  \item \label{item:linearBsoundness} $\emptyset \subset
    \B_{\genericSemanticsX}(p)\subseteq\Lset_{\genericSemanticsX}(p)$,
  \item \label{item:linearBentailment} for each
    $\phi\in\Lset_{\genericSemanticsX}(p)$, it holds that
    $\bigcap_{\psi\in\B_{\genericSemanticsX}(p)}\sat{\psi}\subseteq\sat{\phi}$,
  \item \label{item:linearBfiniteness} $\B_{\genericSemanticsX}(p)$ is finite,
    and
  \item \label{item:linearBmonotonicity} for each $q\in P$, if
    $\Lset_{\genericSemanticsX}(p)\subseteq\Lset_{\genericSemanticsX}(q)$ then
    $\B_{\genericSemanticsX}(p)\subseteq\B_{\genericSemanticsX}(q)$. 
  \end{enumerate}

  Let $p \in P$ and $\genericSemanticsX \in \lspectrumSet$.

  Property~\ref{item:linearBfiniteness} immediately follows from the definition
  of $\B_{\genericSemanticsX}(p)$, since the set $\genericSemanticsFiniteX(p)$ is
  finite.

  As for property~\ref{item:linearBsoundness}, we observe that $\true \in
  \B_\genericSemanticsX(p)$, and thus it suffices to prove
  $\B_{\genericSemanticsX}(p)\subseteq\Lset_{\genericSemanticsX}(p)$.
  To this end, we need to show that $\getFormula{\genericSemanticsX}{x} \in
  \Lset_{\genericSemanticsX}(p)$ for every $x \in \genericSemanticsFiniteX(p)$.
  Intuitively, this is very easy to see: an element $x$ of
  $\genericSemanticsFiniteX(p)$ carries information on how $p$ can/cannot evolve;
  the formula \getFormula{\genericSemanticsX}{x} express exactly the same information
  about $p$, and thus it is satisfied by $p$.
  We refer the reader to~\ref{app:linear},
  Lemma~\ref{lem:soundness-of-B-for-linear}, for a formal proof.

  Property~\ref{item:linearBentailment} follows from the observation that
  $\bigwedge_{\psi \in \B(p)} \psi$ is characteristic for $p$ within
  $\Lset_{\genericSemanticsX}$ (a detailed proof of this claim is given
  in~\ref{app:linear}, Lemma~\ref{lem:linear-B-is-characteristic}).
  Indeed, let $\phi \in \Lset_{\genericSemanticsX}(p)$.
  By the definition of characteristic formula
  (Definition~\ref{def:characteristic_formula}), we have $\sat{\bigwedge_{\psi \in
      \B_{\genericSemanticsX}(p)} \psi} = \{ q \in P \mid
  \Lset_{\genericSemanticsX}(p) \subseteq \Lset_{\genericSemanticsX}(q) \}$, and
  thus we have $\bigcap_{\psi \in \B_{\genericSemanticsX}(p)}\sat{\psi} =
  \sat{\bigwedge_{\psi \in \B_{\genericSemanticsX}(p)} \psi} = \{ q \in P \mid
  \Lset_{\genericSemanticsX}(p) \subseteq \Lset_{\genericSemanticsX}(q) \}
  \subseteq \sat{\phi}$ since $\sat{\phi}$ is upwards closed by
  Proposition~\ref{prop:upward_closure}(\ref{item:upclose}).

  Finally, in order to prove property~\ref{item:linearBmonotonicity}, let $q \in
  P$ be such that $\Lset_{\genericSemanticsX}(p) \subseteq
  \Lset_{\genericSemanticsX}(q)$.
  By Theorem~\ref{theo:linear-semantics}, we have $p \<_{\genericSemanticsX} q$.
  By Corollary~\ref{cor:spectrum_depth} in~\ref{app:linear}, we know that
  $\depth(p)\leq\depth(q)$.
  Using this property, it is easy to see that $\genericSemanticsX(p) \subseteq
  \genericSemanticsX(q)$ implies $\genericSemanticsFiniteX(p) \subseteq
  \genericSemanticsFiniteX(q)$.
  Then, we have $p \<_{\genericSemanticsX} q \Leftrightarrow
  \genericSemanticsX(p) \subseteq \genericSemanticsX(q) \Rightarrow
  \genericSemanticsFiniteX(p) \subseteq \genericSemanticsFiniteX(q) \Rightarrow
  \B_{\genericSemanticsX}(p) \subseteq \B_{\genericSemanticsX}(q)$.
%
\end{proof}

\begin{lemma} \label{lem:anti-char_linear}
  Let $\genericSemanticsX \in \lspectrumSet$.
  For each $p \in P$ and $\chiX{\genericSemanticsX}(p)$ characteristic within
  $\Lset_X$ for $p$, there exists a formula in $\Lset_{\genericSemanticsX}$,
  denoted by $\barchiX{\genericSemanticsX}(p)$, such that
  \begin{inparaenum}[(i)]
  \item $p \not \in \sat{\barchiX{\genericSemanticsX}(p)}$ and
  \item $\{ p' \in P \mid p' \not \<_{\genericSemanticsX} p \} \subseteq
    \sat{\barchiX{\genericSemanticsX}(p)} $.
  \end{inparaenum}
\end{lemma}
\begin{proof}
  First, we define, for $\genericSemanticsX \in \lspectrumSet$,

  {\centering

    $\genericSemanticsFiniteXbar(p) = \big( \bigcup_{\eqClassbisim{p'} \in
      \eqClassbisim{\smallset{P}{\depth(p)}}} \genericSemanticsFiniteX(p') \big)
    \setminus \genericSemanticsX(p)$,

  }

  \noindent
  and then we define $\barchiX{\genericSemanticsX}$ as follows: for each $p \in
  P$

  {\centering

    $\barchiX{\genericSemanticsX}(p) = \bigvee_{x \in
      \genericSemanticsFiniteXbar(p)} \getFormula{\genericSemanticsX}{x} \vee
    \bigvee_{\substack{\tau \in \equalset{Act^*}{\depth(p)+1}}}
    \exist{\tau}\true$

  }

  Let $p \in P$ and $\genericSemanticsX \in \lspectrumSet$.
  We have to show that
  \begin{inparaenum}[(i)]
  \item \label{item:proof-of-lem:anti-char_linear-p-not-in-barchi} $p \not \in
    \sat{\barchiX{\genericSemanticsX}(p)}$ and
  \item \label{item:proof-of-lem:anti-char_linear-barchi-condition} $\{ p' \in P
    \mid p' \not \<_{\genericSemanticsX} p \} \subseteq
    \sat{\barchiX{\genericSemanticsX}(p)} $.
  \end{inparaenum}

  In order to prove~(\ref{item:proof-of-lem:anti-char_linear-p-not-in-barchi}),
  assume, towards a contradiction, that $p \in
  \sat{\barchiX{\genericSemanticsX}(p)}$.
  Then, we distinguish two possibilities:
  \begin{enumerate}
  \item if $p \in \sat{\bigvee_{\substack{\tau \in
          \equalset{Act^*}{\depth(p)+1}}} \exist{\tau}\true}$, then there is $\tau \in
    \trace(p)$ whose length is greater than the depth of $p$, which is a
    contradiction;
  \item if $p \in \sat{\getFormula{\genericSemanticsX}{x}}$ for some $x \in
    \genericSemanticsFiniteXbar(p)$, then, by
    Proposition~\ref{prop:p-sat-formula-x-for-x-in-Xp}, we have that $x \in
    \genericSemanticsX(p)$, which is in contradiction with $x$ being an element of
    $\genericSemanticsFiniteXbar(p)$.
  \end{enumerate}

  In order to prove~(\ref{item:proof-of-lem:anti-char_linear-barchi-condition}),
  let $p' \in P$ be such that $p' \not\<_{\genericSemanticsX} p$.
  If $\depth(p') > \depth(p)$, then there is a trace $\tau \in \trace(p')$ with
  length $\depth(p)+1$, and thus $p' \in \sat{\bigvee_{\substack{\tau \in
        \equalset{Act^*}{\depth(p)+1}}} \exist{\tau}\true} \subseteq
  \sat{\barchiX{\genericSemanticsX}(p)}$.
  Otherwise ($\depth(p') \leq \depth(p)$), we proceed as follows.
  By Definition~\ref{def:linear-semantics}, $\genericSemanticsX(p')
  \not\subseteq \genericSemanticsX(p)$, and, by
  Proposition~\ref{prop:X-and-Xfinite} (not \ref{item:Xp-subset-Xp} $\Rightarrow$
  not \ref{item:Xfinitep-subset-Xp}), we have $\genericSemanticsFiniteX(p')
  \not\subseteq \genericSemanticsX(p)$.
  Thus, there is some $x \in \genericSemanticsFiniteX(p') \setminus
  \genericSemanticsX(p)$, which means that $x \in \genericSemanticsFiniteXbar(p)$.
  By Proposition~\ref{prop:p-sat-formula-x-for-x-in-Xp}, $x \in
  \genericSemanticsFiniteX(p') \subseteq \genericSemanticsX(p')$ implies $p' \in
  \sat{\getFormula{\genericSemanticsX}{x}}$ and, since $x \in
  \genericSemanticsFiniteXbar(p)$, we have $p' \in
  \sat{\barchiX{\genericSemanticsX}(p)}$.
%
\end{proof}

Finally, the following theorem states the main result of this section.
\begin{theorem}[Characterization by primality for the linear time
  spectrum] \label{theo:char_by_prim_linear_spectrum} Let $X \in \lspectrumSet$
  and $\phi \in \Lset_X$.
 Then, $\phi$ is consistent and prime if and only if $\phi$ is
 characteristic for some $p \in P$.
\end{theorem}
\begin{proof}
  The claim immediately follows from Theorem~\ref{theo:char_prime},
  Theorem~\ref{theo:decomp},
  Corollary~\ref{cor:decomposability_condition_for_spectrum},
  Lemma~\ref{lem:characterization_linear}, and Lemma~\ref{lem:anti-char_linear}.
 \end{proof}


\section{Conclusions and future directions} \label{sec:conclusions}

In this paper, we have provided general sufficient conditions guaranteeing that formulae for which model checking can be reduced to equivalence/preorder
checking, that is, the characteristic formulae, are exactly the consistent and prime ones.
We have applied our framework to show that characteristic formulae are exactly
the consistent and prime ones when the set of processes is finite, as well as
for modal refinement semantics~\cite{BoudolL1992}, covariant-contravariant
semantics~\cite{AcetoEtAl11b} (the result was known for these last
two semantics), and all the semantics in van Glabbeek's
spectrum~\cite{VanGlabbeek01} and those considered in~\cite{dFGPR13}.
%
%
%
%
Our results indicate that the ``characterization by primality'' result, first proved by Boudol and 
Larsen~\cite{BoudolL1992} in the context of the modal logic that characterizes modal refinement over modal transition systems, holds in a wide variety of 
settings in concurrency theory. We feel, therefore, that this study reinforces the view that there is a very close connection between the behavioural and 
logical view of processes: not only do the logics characterize processes up to the chosen notion of behavioural relation, but processes characterize all 
the prime and consistent formulae.
%
%


\subsection{Applying the theory to conformance simulation} \label{sec:conformance}

As a future work, we would like to provide a characterization by primality for
the logics characterizing \emph{conformance simulation}
(\conformance)~\cite{FabregasEtAl10-logics}.
Here, we give evidence (Proposition~\ref{prop:no_characterization_conformance}
below) that Corollary~\ref{cor:decomposability_condition_for_spectrum} cannot be
used to show the decomposability of the corresponding logic
$\Lset_{\conformance}$ (defined below).
However, we are confident that the alternative path to decomposability we
provide (through
Proposition~\ref{prop:decomposability-for-finite-modal-covariant}) might serve
the purpose.

Conformance simulation~\cite{FabregasEtAl10-logics} is defined as the
largest relation $\<_{\conformance}$ satisfying:

{\centering

  $\begin{array}{llll}
     p\<_\conformance q
     & \Leftrightarrow
     & I(p)\subseteq I(q) \text{ and}
     & \text{for all } q',a \text{ such that } q\stackrel{a}{\rightarrow} q'
       \text{ and } p\stackrel{a}{\rightarrow} \\
     &
     &
     & \text{there exists } p' \text{ such that } p\stackrel{a}{\rightarrow}p'
       \text{ and } p'\<_\conformance q'.

   \end{array}$

}

The logic $\Lset_\conformance$, characterizing conformance simulation,
is interpreted over set of processes $P$, defined as in
Section~\ref{page:process-definition} on page~\pageref{page:process-definition}.
Its language includes the Boolean constants $\true$ and $\false$, the Boolean
connectives $\wedge$ and $\vee$, as well as the modality $\conf{a}$ ($a \in
Act$), whose semantic interpretation is as follows:

{\centering

  $p \in \sat{\conf{a} \phi}$ iff $p\stackrel{a}{\rightarrow}$ and $p' \in
  \sat{\phi}$ for all $p'$ such that $p \stackrel{a}{\rightarrow} p'$.

}

\noindent The semantic clauses for the other operators are given in Definition~\ref{def:semantics_branching}.

Analogously to the case of the semantics in the linear time-branching time
spectrum, the logic $\Lset_\conformance$ captures exactly conformance simulation, as stated by the following result.

\begin{theorem}[Logical characterization of conformance simulation~\cite{FabregasEtAl10-logics}]\label{theo:conformance-semantics}
  For all $p,q \in P$, $p\<_\conformance q$ iff
  $\Lset_\conformance(p)\subseteq\Lset_\conformance(q)$.
\end{theorem}

The next result shows that
Corollary~\ref{cor:decomposability_condition_for_spectrum} cannot be used to
show the decomposability of logic $\Lset_{\conformance}$
and thus to prove its characterization by primality.

\begin{proposition}\label{prop:no_characterization_conformance}
  $\Lset_{\conformance}$ is not finitely characterized by $\B$, for
  any monotonic $\B$ (see Definition~\ref{def:characterized}).
%
\end{proposition}

\begin{proof}
  Suppose, towards a contradiction, that $\Lset_{\conformance}$ is finitely
  characterized by $\B$, for some monotonic $\B$, and let us define the following
  processes, for $k >0$: $p_k= \sum_{i=1}^{k}a^ib \zeroProcess$.

  Since $I(p_k)=I(p_{k+1})$ and $p_k$ is a subtree of $p_{k+1}$, it is clear
  that $p_{k+1}\<_\conformance p_{k}$ holds, for all $k$ (see the picture below to
  verify the claim for $k=1$, i.e., $p_2 \<_{\conformance} p_1$).

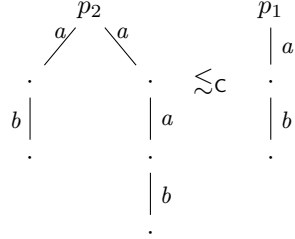
{\centering
\begin{tikzpicture}[>=stealth',shorten >=1pt,auto]
\matrix [matrix of math nodes, column sep={.8cm,between origins},row sep={1cm,between origins}]
{
& \node (C0) {p_2}; &
& &  \node (D0) {p_1};
 \\
\node (C11) {\cdot}; & & \node (C12) {\cdot}; &
\node (sim) {\<_\conformance}; & \node (D1) {\cdot}; \\
\node (C21) {\cdot}; & & \node (C22) {\cdot}; &
 & \node (D2) {\cdot};  \\
& & \node (C3) {\cdot}; &
 &  \\
};
\begin{scope}[every node/.style={font=\small\itshape}]
\path   (C0)   edge        node [left,above]  {a} (C11)
		(C0)   edge        node [right,above] {a} (C12)
		(C11)  edge        node [left]        {b} (C21)
		(C12)  edge        node [right]       {a} (C22)
		(C22)  edge        node [right]       {b} (C3)
		(D0)   edge        node    {a} (D1)
		(D1)   edge        node    {b} (D2);
\end{scope}
\end{tikzpicture}

}

On the other hand, $p_{k}\not\<_\conformance p_{k+1}$, because the branch
$a^{k+1}b\zeroProcess$ cannot be matched by any branch in $p_k$, where the
longest sequence of consecutive $a$'s is $a^{k}$.
%
%
Thus, we have a non-well-founded sequence of strictly decreasing processes
$\ldots\strict_\conformance p_{k+1}\strict_\conformance
p_{k}\ldots\strict_\conformance p_{1}$.

For all $k$, $p_{k+1}\strict_\conformance p_{k}$ implies
$\Lset_{\conformance}(p_{k+1}) \subset \Lset_{\conformance}(p_{k})$, which, by
monotonicity of \B (Definition~\ref{def:characterized}), in turn implies
$\B(p_{k+1}) \subset \B(p_{k})$.
Thus, there exist infinitely many formulae $\psi_1, \psi_2, \ldots, \psi_k,
\ldots$ such that $\psi_{k} \in \B(p_{k}) \setminus \B(p_{k+1})$ for all $k$.
We show that the formulae $\{ \psi_k \}_{k \geq 1}$ are pairwise different, that
is, $\psi_k \neq \psi_j$ for each $k \neq j$.
To this end, let us suppose, towards a contradiction, that there exist $j,k$
such that $j < k$ and $\psi_k = \psi_j$.
By construction, we have that $\psi_j \in \B(p_j) \setminus \B(p_{j+1})$
and $\psi_k = \psi_j \in \B(p_k) \setminus \B(p_{k+1})$.
Since $j+1 \leq k$, we have that $\B(p_k) \subseteq \B(p_{j+1})$,
which means that $\psi_k = \psi_j \in \B(p_{j+1})$, leading to a contradiction.

Therefore the sequence $\{ \psi_k \}_{k \geq 1}$ includes infinitely many different formulae.
Since $\B(p_1)$ contains them all, the hypothesis of finiteness
of $\B(p_1)$ is contradicted, and the thesis follows.
%
\end{proof}

\paragraph{Acknowledgements}{} We thank the anonymous referees for their
comments, which led to improvements in the paper.

\section*{References}


\newpage

\appendix
\section{Proofs for semantics in van Glabbeek's branching time spectrum}
\label{app:branching}

\subsection{Finite characterization} \label{sec:full-proof-finite-char}

We prove here that logics in van Glabbeek's branching-time spectrum
are finitely characterized by some monotonic $\B$
(condition~(\ref{item:final_req_characteristic})
in Corollary~\ref{cor:decomposability_condition_for_spectrum}).

\setcounter{resume}{\value{lemma}}
\setcounter{lemma}{\value{lemcharacterizationready}}
\begin{lemma}
 \lemcharacterizationready
\end{lemma}
\setcounter{lemma}{\value{resume}}

\begin{proof}

  We consider each semantics $\genericSemanticsX \in \bspectrumSet$ and we show
  that $\B_{\genericSemanticsX}$ (as defined in
  Table~\ref{tab:syntax_monotonicity_finite_char_branching_spectrum}) is such
  that for each $p\in P$:
  \begin{enumerate}[label=(\emph{\roman*}),ref=\emph{\roman*}]
  \item \label{item:BsoundnessApp} $\emptyset \subset
    \B_{\genericSemanticsX}(p)\subseteq\Lset_{\genericSemanticsX}(p)$,
  \item \label{item:BentailmentApp} for each
    $\phi\in\Lset_{\genericSemanticsX}(p)$, it holds
    $\bigcap_{\psi\in\B_{\genericSemanticsX}(p)}\sat{\psi}\subseteq\sat{\phi}$,
  \item \label{item:BfinitenessApp} $\B_{\genericSemanticsX}(p)$ is finite, and
  \item \label{item:BmonotonicityApp} for each $q\in P$, if
    $\Lset_{\genericSemanticsX}(p)\subseteq\Lset_{\genericSemanticsX}(q)$ then
    $\B_{\genericSemanticsX}(p)\subseteq\B_{\genericSemanticsX}(q)$. 
  \end{enumerate}

  Before considering each semantics separately, we notice that, for every
  $\genericSemanticsX \in \bspectrumSet$, we have that $\true \in
  \B_{\genericSemanticsX}(p)$, and thus $\emptyset \subset
  \B_{\genericSemanticsX}(p)$.
  Consequently, we can focus on showing that
  $\B_{\genericSemanticsX}(p)\subseteq\Lset_{\genericSemanticsX}(p)$ holds when
  proving property~\eqref{item:BsoundnessApp}.

  In addition, we find it convenient to partially factorize the proofs of
  properties~\eqref{item:BsoundnessApp} and~\eqref{item:BfinitenessApp}.
  To this end, we show that, for every $\genericSemanticsX \in \bspectrumSet$,

  \noindent
  \begin{align}
    \label{eq:BsoundnessApp-factorized} \tag{\ref*{item:BsoundnessApp}$'$}
    & \text{if $\B_{\genericSemanticsX}^-(p) \subseteq
      \Lset_{\genericSemanticsX}(p)$ for all $p \in P$, then
      $\B_{\genericSemanticsX}(p) \subseteq \Lset_{\genericSemanticsX}(p)$ for all $p
      \in P$, and} \\
    \label{eq:BfinitenessApp-factorized} \tag{\ref*{item:BfinitenessApp}$'$}
    & \text{if $\B_{\genericSemanticsX}^-(p)$ is finite for all $p \in P$, then
      $\B_{\genericSemanticsX}(p)$ is finite for all $p \in P$.}
  \end{align}
  Let $X \in \bspectrumSet$.

  First, let us assume that $\B_{\genericSemanticsX}^-(p) \subseteq
  \Lset_{\genericSemanticsX}(p)$ holds for all $p \in P$; we show that
  $\B_{\genericSemanticsX}(p) \subseteq \Lset_{\genericSemanticsX}(p)$ holds for
  all $p \in P$ as well, by induction on the depth of $p$.
  When $I(p) = \emptyset$ (base case), we have
  $\B_{\genericSemanticsX}^+(p)=\{\true\} \subseteq
  \Lset_{\genericSemanticsX}(p)$, and therefore $\B_{\genericSemanticsX}(p)
  \subseteq \Lset_{\genericSemanticsX}(p)$ holds as well.
  To deal with the inductive step ($I(p) \neq \emptyset$), let
  $\phi\in\B_{\genericSemanticsX}^+(p)$.
  If $\phi = \true$, then $\phi \in \Lset_{\genericSemanticsX}(p)$, and we are
  done.
  Assume $\phi = \exist{a}\bigwedge_{\psi\in\Psi}\psi$, where $a \in Act$ and
  $\Psi \subseteq \B_{\genericSemanticsX}(p')$ for some $p'$ such that
  $p\stackrel{a}{\rightarrow}p'$.
  By the inductive hypothesis, we have that
  $\B_{\genericSemanticsX}(p')\subseteq\Lset_{\genericSemanticsX}(p')$, meaning
  that $p' \in \sat{\psi}$ for all $\psi \in \Psi$.
  Since $p\stackrel{a}{\rightarrow}p'$, we have that
  $p\in\sat{\exist{a}\bigwedge_{\psi\in\Psi}\psi}$, which amounts to $\phi \in
  \Lset_{\genericSemanticsX}(p)$.

  Now, let us assume $\B_{\genericSemanticsX}^-(p)$ to be finite for every $p
  \in P$; we show that also $\B_{\genericSemanticsX}(p)$ is finite for every $p
  \in P$, by induction on the depth of $p$.
  When $I(p) = \emptyset$ (base case), $\B_{\genericSemanticsX}(p)= \{ \true \}
  \cup \B_{\genericSemanticsX}^-(p)$, which is clearly finite.
  Let us deal now with the inductive step ($I(p)\neq\emptyset$).
  By the construction of $\B_{\genericSemanticsX}^+(p)$, a formula belongs to
  $\B_{\genericSemanticsX}^+(p)$ if, and only if, it is either \true or
  $\exist{a}\varphi$, where $a \in Act$ and $\varphi=\bigwedge_{\psi\in\Psi}\psi$,
  for some $\Psi \subseteq \B_{\genericSemanticsX}(p')$ and some $p'$ such that
  $p\stackrel{a}{\rightarrow}p'$.
  By the inductive hypothesis, $\B_{\genericSemanticsX}(p')$ is finite.
  Since $Act$ is also finite and processes are finitely branching, there are
  only finitely many such formulae $\exist{a}\varphi$, meaning that
  $\B_{\genericSemanticsX}^+(p)$ is finite.
  Therefore $\boundfun_{\genericSemanticsX}(p) =
  \boundfun_{\genericSemanticsX}^+(p)\cup \boundfun_{\genericSemanticsX}^-$ is
  finite as well.

  As a consequence of~\eqref{eq:BsoundnessApp-factorized}, showing that
  $\B_{\genericSemanticsX}^-(p) \subseteq \Lset_{\genericSemanticsX}(p)$ holds for
  all $p \in P$ is enough to prove property~\eqref{item:BsoundnessApp}; similarly,
  thanks to~\eqref{eq:BfinitenessApp-factorized}, showing finiteness of $\B^-(p)$
  for all $p \in P$ is enough to prove property~\eqref{item:BfinitenessApp}.
  We will use~\eqref{eq:BsoundnessApp-factorized} and
  \eqref{eq:BfinitenessApp-factorized} to prove, respectively,
  properties~\eqref{item:BsoundnessApp} and~\eqref{item:BfinitenessApp} for all
  semantics.

  To keep the notation light, we often omit the subscripts identifying the
  semantics as it is clear from the context, e.g., when working out the case of
  simulation semantics, we write $\Lset$, $\B$, and $p\< q$ instead of
  $\Lset_{\simulation}$, $\B_{\simulation}$, and $p\<_{\simulation} q$,
  respectively, and similarly for the other semantics.

\noindent{\bf Case simulation (\simulation).}
For the sake of clarity we recall from
\tablename~\ref{tab:syntax_monotonicity_finite_char_branching_spectrum} that
$\B$ is defined as

{\centering
$\boundfun(p)=\{\true\}\cup\{\exist{a}\varphi\mid a\in Act, \varphi=\bigwedge_{\psi\in\Psi}\psi, \Psi\subseteq\boundfun(p'), p\stackrel{a}{\rightarrow}p'\}$.

}

Since $\B^-(p) = \emptyset$ for all $p \in P$,
properties~\eqref{item:BsoundnessApp} and~\eqref{item:BfinitenessApp}
immediately follow from~\eqref{eq:BsoundnessApp-factorized}
and~\eqref{eq:BfinitenessApp-factorized}, respectively.
Notice also that property~\eqref{item:BfinitenessApp} implies that $\B$ is well
defined.

In order to prove property~\eqref{item:BentailmentApp}, we let $\phi \in
\Lset(p)$, for a generic $p \in P$, and we proceed by induction on the structure
of $\phi$ (notice that we can ignore the case $\phi = \false$, as $\phi \in
\Lset(p)$ implies $\phi \neq \false$).
\begin{itemize}
\item $\phi=\true$: the claim follows trivially.

\item $\phi= \varphi_1 \vee \varphi_2$: it holds that $\varphi_i\in\Lset(p)$ for
  some $i\in \{ 1,2 \}$. By the inductive hypothesis, we have that
  $\bigcap_{\psi\in\B(p)}\sat{\psi}\subseteq \sat{\varphi_i}$ and, since
  $\sat{\varphi_i}\subseteq\sat{\phi}$, we obtain the claim.

\item $\phi= \varphi_1 \wedge \varphi_2$: it holds that $\varphi_i\in\Lset(p)$
  for all $i\in \{ 1,2 \}$.
  By the inductive hypothesis, we have that
  $\bigcap_{\psi\in\B(p)}\sat{\psi}\subseteq \sat{\varphi_i}$ for all $i \in \{
  1,2 \}$.
  This implies that $\bigcap_{\psi \in \B(p)} \sat{\psi} \subseteq
  \sat{\varphi_1} \cap \sat{\varphi_2} = \sat{\phi}$.

\item $\phi=\exist{a}\varphi$: by definition we have that $\varphi\in\Lset(p')$
  for some $p\stackrel{a}{\rightarrow}p'$.
  By the inductive hypothesis, we have that
  $\bigcap_{\psi\in\B(p')}\sat{\psi}\subseteq \sat{\varphi}$.
  We define $\zeta=\exist{a}\bigwedge_{\psi\in \B(p')}\psi$.
  Clearly, $\zeta$ belongs to $\B^+(p)$ (by construction---notice that $\zeta$
  is well defined due to the finiteness of $\B(p')$) and
  $\sat{\zeta}\subseteq\sat{\phi}$ (because $\bigcap_{\psi\in\B(p')}\sat{\psi}
  \subseteq \sat{\varphi}$).
  Hence, $\bigcap_{\psi\in\B(p)}\sat{\psi}\subseteq \sat{\zeta} \subseteq
  \sat{\phi}$ holds.
\end{itemize}

Finally, we show that $\B$ is monotonic
(property~\eqref{item:BmonotonicityApp}).
Consider $p,q\in P$, with $\Lset(p)\subseteq\Lset(q)$.
We want to show that $\phi\in\B(p)$ implies $\phi\in\B(q)$, for each $\phi$.
Firstly, we observe that, by $\Lset(p)\subseteq\Lset(q)$ and
Theorem~\ref{theo:branching-semantics}, $p \< q$ holds. Thus, for each $a \in
Act$ and $p' \in P$ with $p \stackrel{a}{\rightarrow} p'$, there exists some $q'
\in P$ such that $q \stackrel{a}{\rightarrow} q'$ and $p' \< q'$.
We also observe that $\B^-(p) = \B^-(q) = \emptyset$.
In order to show that $\B(p) \subseteq \B(q)$, we proceed by induction on the depth of $p$.
If $I(p) = \emptyset$, then $\B^+(p) = \{ \true \} \subseteq \B^+(q)$, and the
thesis follows.
Otherwise ($I(p) \neq \emptyset$), let us consider a formula
$\phi \in \B(p)$.
If $\phi \in \B^-(p)$, then the claim follows from $\B^-(p) = \B^-(q) \subseteq
\B(q)$.
If $\phi = \true$, then, by definition of $\B^+$, we have $\phi \in \B^+(q)
\subseteq \B(q)$.
Finally, if $\phi = \exist{a}\varphi \in \B^+(p)$, then, by definition of
$\B^+$, there exist $p' \in P$, with $p\stackrel{a}{\rightarrow}p'$, such that
$\varphi=\bigwedge_{\psi\in \Psi}\psi$ for some $\Psi \subseteq \boundfun(p')$.
This implies the existence of some $q' \in P$ such that $q \stackrel{a}{\rightarrow} q'$
and $p' \< q'$ (and therefore $\Lset(p') \subseteq \Lset(q')$ by Theorem~\ref{theo:branching-semantics}).
By the inductive hypothesis, $\boundfun(p') \subseteq \boundfun(q')$ holds as
well, which means that $\Psi \subseteq \boundfun(q')$. Hence, we have that
$\exist{a}\varphi \in \B^+(q) \subseteq \B(q)$.

\noindent{\bf Case complete simulation (\completeSim).}
For the sake of clarity we recall from \tablename~\ref{tab:syntax_monotonicity_finite_char_branching_spectrum}
that $\B$ is defined as
$\boundfun(p) = \boundfun^+(p)\cup \boundfun^-(p)$, where
\begin{itemize}

\item $\boundfun^+(p)=\{\true\}\cup\{\exist{a}\varphi\mid a\in Act,
  \varphi=\bigwedge_{\psi\in\Psi}\psi, \Psi\subseteq\boundfun(p'),
  p\stackrel{a}{\rightarrow}p'\}$, and
\item $\boundfun^-(p)=\{\zeroFormula\mid p\stackrel{a}{\not\rightarrow}, \forall a \in Act\}$.
\end{itemize}

For all $p \in P$, if $I(p) = \emptyset$, then $\B^-(p) = \{ \zeroFormula \}$
and $\zeroFormula \in \Lset(p)$, otherwise $\B^-(p) = \emptyset$; thus,
$\B^-(p)$ is finite and $\B^-(p) \subseteq \Lset(p)$.
Hence, properties~\eqref{item:BsoundnessApp} and~\eqref{item:BfinitenessApp}
immediately follow from~\eqref{eq:BsoundnessApp-factorized}
and~\eqref{eq:BfinitenessApp-factorized}, respectively.
Notice also that property~\eqref{item:BfinitenessApp} implies that $\B$ is well
defined.

In order to prove property~\eqref{item:BentailmentApp}, we let $\phi \in
\Lset(p)$, for a generic $p \in P$, and we proceed by induction on the structure
of $\phi$.
Apart from a new base case ($\phi = \zeroFormula$, which we deal with below),
the proof is the same as in the case of simulation semantics above.
\begin{itemize}
\item $\phi=\zeroFormula$: it is enough to observe that $\phi\in\B(p)$, which
  implies $\bigcap_{\psi\in\B(p)} \sat{\psi} \subseteq \sat{\phi}$.
\end{itemize}

Finally, we show that $\B$ is monotonic
(property~\eqref{item:BmonotonicityApp}).
Consider $p,q\in P$, with $\Lset(p)\subseteq\Lset(q)$.
We want to show that $\phi\in\B(p)$ implies $\phi\in\B(q)$, for each $\phi$.
Firstly, we observe that, by $\Lset(p)\subseteq\Lset(q)$ and
Theorem~\ref{theo:branching-semantics}, $p \< q$ holds. Thus, $I(p) = \emptyset$
if and only if $I(q) = \emptyset$, which implies $\B^-(p) = \B^-(q)$; moreover,
we have that for each $a \in Act$ and $p' \in P$ with $p
\stackrel{a}{\rightarrow} p'$, there exists some $q' \in P$ such that $q
\stackrel{a}{\rightarrow} q'$ and $p' \< q'$.
In order to show that $\B(p) \subseteq \B(q)$, we proceed by induction on the
depth of $p$.
If $I(p) = \emptyset$, then $I(q) = \emptyset$ as well. Thus, we have that
$\B^+(p) = \B^+(q) = \{ \true \}$, and the thesis follows.
Otherwise ($I(p) \neq \emptyset$), let us consider a formula $\phi \in \B(p)$.
If $\phi \in \B^-(p)$, then the claim follows from $\B^-(p) = \B^-(q) \subseteq
\B(q)$.
If $\phi = \true$, then, by definition of $\B^+$, we have $\phi \in \B^+(q)
\subseteq \B(q)$.
Finally, if $\phi = \exist{a}\varphi \in \B^+(p)$, then, by definition of
$\B^+$, there exist $p' \in P$, with $p\stackrel{a}{\rightarrow}p'$, such that
$\varphi=\bigwedge_{\psi\in \Psi}\psi$ for some $\Psi \subseteq \boundfun(p')$.
This implies the existence of some $q' \in P$ such that $q
\stackrel{a}{\rightarrow} q'$ and $p' \< q'$ (and therefore $\Lset(p') \subseteq
\Lset(q')$ by Theorem~\ref{theo:branching-semantics}).
By the inductive hypothesis, $\boundfun(p') \subseteq \boundfun(q')$ holds as
well, which means that $\Psi \subseteq \boundfun(q')$. Hence, we have that
$\exist{a}\varphi \in \B^+(q) \subseteq \B(q)$.

\noindent{\bf Case ready simulation (\readySim).}
See the proof of Lemma~\ref{lem:characterization_branching} at
page~\pageref{lem:characterization_branching}.

\noindent{\bf Case trace simulation (\traceSim).}
For the sake of clarity we recall from \tablename~\ref{tab:syntax_monotonicity_finite_char_branching_spectrum}
that $\B$ is defined as
$\boundfun(p) = \boundfun^+(p)\cup \boundfun^-(p)$, where
\begin{itemize}
\item $\boundfun^+(p)=\{\true\}\cup\{\exist{a}\varphi\mid a\in Act,
  \varphi=\bigwedge_{\psi\in\Psi}\psi, \Psi\subseteq\boundfun(p'),
  p\stackrel{a}{\rightarrow}p'\}$, and

\item $\B^-(p) = \{\univer{\tau a}\false\mid \tau\in \traces(p), a \in Act, \tau a \notin \traces(p) \}$.

\end{itemize}

It is easy to see that, for all $p$,
\begin{inparaenum}[$(a)$]
\item $\B^-(p)$ is finite, because $Act$ is finite and so is $\traces(p)$, and
\item $\B^-(p) \subseteq \Lset(p)$ trivially holds, by definition of $\B^-(p)$.
\end{inparaenum}
Hence, properties~\eqref{item:BsoundnessApp} and~\eqref{item:BfinitenessApp}
immediately follow from~\eqref{eq:BsoundnessApp-factorized}
and~\eqref{eq:BfinitenessApp-factorized}, respectively.
Notice also that property~\eqref{item:BfinitenessApp} implies that $\B$ is well
defined.

In order to prove property~\eqref{item:BentailmentApp}, we let $\phi \in
\Lset(p)$, for a generic $p \in P$, and we proceed by induction on the structure
of $\phi$.
Inductive steps work as in previous cases; the following additional base case,
with respect to previous cases, must be considered.
\begin{itemize}
\item $\phi=\univer{\tau a}\false$: if $\phi\in\B^-(p)$, the property holds
  trivially.
  Otherwise, $\phi\notin\B^-(p)$ implies $\tau \notin \traces(p)$ (notice that
  $\tau a \notin \traces(p)$: indeed, if $\tau a \in \traces(p)$ then
  $\phi\notin\Lset(p)$, thus raising a contradiction).
  There must be two prefixes $\tau' = a_1 \ldots a_k$ and $\tau'' = a_1 \ldots
  a_k a_{k+1}$ of $\tau$ (i.e., $\tau'$ is the largest proper prefix of
  $\tau''$---possibly $\tau' = \varepsilon$ and/or $\tau'' = \tau$) such that
  $p\stackrel{\tau'}{\longrightarrow}$ and $p\stackrel{\tau''}{\not\rightarrow}$.
  The formula $\psi=\univer{\tau''}\false$ is such that $\psi \in \B^-(p)$ and
  $\sat{\psi}\subseteq\sat{\phi}$.
  The thesis follows immediately.
\end{itemize}

Finally, we show that $\B$ is monotonic
(property~\eqref{item:BmonotonicityApp}).
Consider $p,q\in P$, with $\Lset(p)\subseteq\Lset(q)$.
We want to show that $\phi\in\B(p)$ implies $\phi\in\B(q)$, for each $\phi$.
Firstly, we observe that, by $\Lset(p)\subseteq\Lset(q)$ and
Theorem~\ref{theo:branching-semantics}, $p \< q$ holds. Thus, we have that
$\traces(p) = \traces(q)$, which implies $\B^-(p) = \B^-(q)$; moreover, we have
that for each $a \in Act$ and $p' \in P$ with $p \stackrel{a}{\rightarrow} p'$,
there exists some $q' \in P$ such that $q \stackrel{a}{\rightarrow} q'$ and $p'
\< q'$.
In order to show that $\B(p) \subseteq \B(q)$, we proceed by induction on the
depth of $p$.
If $I(p) = \emptyset$, then $I(q) = \emptyset$ as well. Thus, we have that
$\B^+(p) = \B^+(q) = \{ \true \}$, and the thesis follows.
Otherwise ($I(p) \neq \emptyset$), let us consider a formula $\phi \in \B(p)$.
If $\phi \in \B^-(p)$, then the claim follows from $\B^-(p) = \B^-(q) \subseteq
\B(q)$.
If $\phi = \true$, then, by definition of $\B^+$, we have $\phi \in \B^+(q)
\subseteq \B(q)$.
Finally, if $\phi = \exist{a}\varphi \in \B^+(p)$, then, by definition of
$\B^+$, there exist $p' \in P$, with $p\stackrel{a}{\rightarrow}p'$, such that
$\varphi=\bigwedge_{\psi\in \Psi}\psi$ for some $\Psi \subseteq \boundfun(p')$.
This implies the existence of some $q' \in P$ such that $q
\stackrel{a}{\rightarrow} q'$ and $p' \< q'$ (and therefore $\Lset(p') \subseteq
\Lset(q')$ by Theorem~\ref{theo:branching-semantics}).
By the inductive hypothesis, $\boundfun(p') \subseteq \boundfun(q')$ holds as
well, which means that $\Psi \subseteq \boundfun(q')$. Hence, we have that
$\exist{a}\varphi \in \B^+(q) \subseteq \B(q)$.

\noindent{\bf Case 2-nested simulation (\twoSim).}
For the sake of clarity we recall from \tablename~\ref{tab:syntax_monotonicity_finite_char_branching_spectrum}
that $\B$ is defined as
$\boundfun(p) = \boundfun^+(p)\cup \boundfun^-(p)$, where
\begin{compactitem}
\item $\boundfun^+(p)=\{\true\}\cup\{\exist{a}\varphi\mid a\in Act,
  \varphi=\bigwedge_{\psi\in\Psi}\psi, \Psi\subseteq\boundfun(p'),
  p\stackrel{a}{\rightarrow}p'\}$, and

\item $\B^-(p)= \{\univer{a}\varphi\in\Lset(p)\mid a \in Act,
  \varphi=\bigvee_{p' \in \maxsucc(p,a)} \bigwedge_{\psi \in \B^-(p')}\psi
  \}$,

\end{compactitem}
where $\maxsucc(p,a) = \{ p' \in P \mid p\stackrel{a}{\rightarrow}p' \text{ and
} \nexists p'' .
          p\stackrel{a}{\rightarrow}p'' \text{ and } p' <_{\simulation} p'' \}$.

By definition, $\B^-(p) \subseteq \Lset(p)$ for every $p \in P$; moreover, it is
easy to verify that $\B^-(p)$ is finite, by induction on the depth on $p$ and by
recalling that $Act$ is finite and processes are finitely branching.
Hence, properties~\eqref{item:BsoundnessApp} and~\eqref{item:BfinitenessApp}
immediately follow from~\eqref{eq:BsoundnessApp-factorized}
and~\eqref{eq:BfinitenessApp-factorized}, respectively.
Notice also that property~\eqref{item:BfinitenessApp} implies that $\B$ is well
defined.

In order to prove property~\eqref{item:BentailmentApp}, we let $\phi \in
\Lset(p)$, for a generic $p \in P$, and we proceed, as usual, by induction on
the structure of $\phi$.
With respect to the previous semantics, we have to consider new formulae of the
form $\neg \varphi$.
They are treated as an additional base case, which, in turn, is handled by means
of (inner) induction on the structure of $\varphi$.
\begin{itemize}
\item $\phi=\true$ (base case): the claim follows trivially.

\item $\phi=\neg \varphi$ (base case): according to the definition of the syntax
  of $\Lset_{\twoSim}$, $\varphi$ must belong to $\Lset_{\simulation}$.
  We show that $\bigcap_{\psi \in \B_{\twoSim}(p)} \sat{\psi} \subseteq
  \sat{\neg \varphi}$ holds for all $\varphi \in \Lset_{\simulation}$.
  We proceed by (inner) induction on the structure of $\varphi$ (notice that we
  can ignore the case $\varphi = \true$, as $\phi \in \Lset(p)$ implies $\varphi
  \neq \true$).
  \begin{itemize}
  \item $\varphi=\false$: the claim follows trivially.

  \item $\varphi=\varphi_1 \wedge \varphi_2$: then, $\neg \varphi$ is logically
    equivalent to $\neg \varphi_1 \vee \neg \varphi_2$.
    Thus, $\phi \in \Lset(p)$ implies $\neg \varphi_i \in \Lset(p)$ for some $i
    \in \{ 1,2 \}$.
    By inductive hypothesis, $\bigcap_{\psi \in \B(p)} \sat{\psi} \subseteq
    \sat{\neg \varphi_i}$ and, since $\sat{\neg \varphi_i} \subseteq \sat{\neg
      \varphi}$, the thesis follows.

  \item $\varphi=\varphi_1 \vee \varphi_2$: then, $\neg \varphi$ is logically
    equivalent to $\neg \varphi_1 \wedge \neg \varphi_2$.
    Thus, $\phi \in \Lset(p)$ implies $\neg \varphi_i \in \Lset(p)$ for every $i
    \in \{ 1,2 \}$.
    By inductive hypothesis, $\bigcap_{\psi \in \B(p)} \sat{\psi} \subseteq
    \sat{\neg \varphi_i}$ for every $i \in \{ 1,2 \}$, which means that
    $\bigcap_{\psi \in \B(p)} \sat{\psi} \subseteq \sat{\neg \varphi_1} \cap
    \sat{\neg \varphi_2} = \sat{\neg \varphi}$, hence the thesis.

  \item $\varphi=\exist{a} \varphi_1$: then, $\neg \varphi$ is logically
    equivalent to $\univer{a} \neg \varphi_1$.
    Thus, $\phi \in \Lset(p)$ implies $\neg \varphi_1 \in \Lset(p')$ for every
    $p' \in P$ such that $p\stackrel{a}{\rightarrow}p'$.
    By the inductive hypothesis, $\bigcap_{\psi \in \B(p')} \sat{\psi} \subseteq
    \sat{\neg \varphi_1}$ for every $p' \in P$ such that
    $p\stackrel{a}{\rightarrow}p'$, from which it follows $\bigcup_{p' \in
      \maxsucc(p,a)}\bigcap_{\psi \in \B(p')} \sat{\psi} \subseteq \sat{\neg
      \varphi_1}$.
    We define $\zeta = \univer{a}\bigvee_{p' \in \maxsucc(p,a)} \bigwedge_{\psi
    \in \B^-(p')} \psi$ (notice that $\zeta = \univer{a}\false$ if
  $p\stackrel{a}{\not\rightarrow}$).
  If is not difficult to show, by induction on the depth of $p$, that, for all
  $p,q \in P$, if $p \<_{\simulation} q$ then $p \in \sat{\psi}$ for all $\psi \in
  \B^-(q)$.
  Moreover, we observe that for every $p'$ such that
  $p\stackrel{a}{\rightarrow}p'$, there is $q \in \maxsucc(p,a)$ with $p'
  \<_{\simulation} q$.
  Therefore, $\zeta \in \Lset(p)$, because for all $p'$ with
  $p\stackrel{a}{\rightarrow}p'$ we have that $p' \in \sat{\bigwedge_{\psi \in
      \B^-(q)} \psi}$, where $q$ is such that $q \in \maxsucc(p,a)$ and $p'
  \<_{\simulation} q$.
    Moreover, we have that $\zeta \in \B^-(p)$ (by definition of $\B^-(p)$) and
    therefore $\bigcap_{\psi \in \B(p)} \sat{\psi} \subseteq \sat{\zeta}$.
    By $\bigcup_{p' \in \maxsucc(p,a)}\bigcap_{\psi \in \B(p')}
    \sat{\psi} \subseteq \sat{\neg \varphi_1}$, we have $\sat{\zeta} \subseteq
    \sat{\univer{a} \neg \varphi_1} = \sat{\neg \varphi}$, and the thesis
    immediately follows.
  \end{itemize}

\item $\phi \in \{ \varphi_1 \vee \varphi_2, \varphi_1 \wedge \varphi_2,
  \exist{a}\varphi \}$ (inductive step): the proof is exactly the same as in the
  previous cases.

\end{itemize}

Before showing that $\B$ is monotonic (property~\eqref{item:BmonotonicityApp}),
we prove, as a preliminary result, that if $p \equiv_{\simulation} q$ then
$\B^-(p) = \B^-(q)$.
To this end, we first prove that if $p \equiv_{\simulation} q$, then for every
$a \in Act$ and $p' \in \maxsucc(p,a)$ there is $q' \in \maxsucc(q,a)$ such
that $p' \equiv_{\simulation} q'$.
(Notice that, since $\equiv_{\simulation}$ is symmetric, this also implies that
for every $a \in Act$ and $q' \in \maxsucc(q,a)$ there is $p' \in
\maxsucc(p,a)$ such that $p' \equiv_{\simulation} q'$.)
Let $p,q \in P$ be such that $p \equiv_{\simulation} q$, and let $a \in Act$ and
$p' \in \maxsucc(p,a)$.
By $p \<_{\simulation} q$, there is $q'$ such that
$q\stackrel{a}{\rightarrow}q'$ and $p' \<_{\simulation} q'$.
We show that $q' \in \maxsucc(q,a)$ and $p' \equiv_{\simulation} q'$.
In order to show the former, assume, towards a contradiction, that $q' \notin
\maxsucc(q,a)$; thus, there is $q''$ such that $q\stackrel{a}{\rightarrow}q''$
and $q' <_{\simulation} q''$.
By $q \<_{\simulation} p$, we have that there is $p''$ such that
$p\stackrel{a}{\rightarrow}p''$ and $q'' \<_{\simulation} p''$.
Thus, $p''$ is such that $p\stackrel{a}{\rightarrow}p''$ and $p' <_{\simulation}
p''$, thus getting a contradiction with the hypothesis that $p' \in
\maxsucc(p,a)$. Therefore, $q' \in \maxsucc(q,a)$.
Now, in order to show that $p' \equiv_{\simulation} q'$, it is enough to prove
that $q' \<_{\simulation} p'$ (as we already know that $p' \<_{\simulation}
q'$).
Assume, towards a contradiction, that $q' \not\<_{\simulation} p'$ (i.e., $p'
<_{\simulation} q'$).
By $q \<_{\simulation} p$, there is $p''$ such that
$p\stackrel{a}{\rightarrow}p''$ and $q' \<_{\simulation} p''$.
Therefore, $p' <_{\simulation} p''$, thus getting a contradiction with the
hypothesis that $p' \in \maxsucc(p,a)$. Therefore, $q' \<_{\simulation} p'$,
which means $q' \equiv_{\simulation} p'$.
We turn to proving that if $p \equiv_{\simulation} q$ then $\B^-(p) = \B^-(q)$.
By the symmetry of $\equiv_{\simulation}$, it is enough to prove that $\B^-(p)
\subseteq \B^-(q)$.
We proceed by induction on the depth of $p$.
If $I(p) = \emptyset$ (base case), then $I(q) = \emptyset$ as well, and $\B^-(p)
= \B^-(q) = \{ \univer{a} \false \mid a \in Act \}$.
If $I(p) \neq \emptyset$ (inductive step), then let $\univer{a}\varphi \in
\B^-(p)$.
We distinguish two cases.
\begin{itemize}
\item If $a \not \in I(p)$, then $\varphi = \false$ and, by $q \<_{\simulation}
  p$, $a \not \in I(q)$, which implies $\univer{a}\varphi \in \B^-(q)$.

\item If $a \in I(p)$, then $\varphi = \bigvee_{p' \in \maxsucc(p,a)}
  \bigwedge_{\psi \in \B^-(p')} \psi$.
  By $p \equiv_{\simulation} q$, we have that for every $p' \in \maxsucc(p,a)$
  there is $q' \in \maxsucc(q,a)$ with $p' \equiv_{\simulation} q'$, which, by
  inductive hypothesis, implies $\B^-(p') = \B^-(q')$.
  The converse direction holds as well, that is, for every $q' \in
  \maxsucc(q,a)$ there is $p' \in \maxsucc(p,a)$ with $p' \equiv_{\simulation}
  q'$, which, by inductive hypothesis, implies $\B^-(p') = \B^-(q')$.
  Consequently, $\univer{a}\varphi \in \B^-(q)$, and we are done.
\end{itemize}

We are now ready to show that $\B$ is monotonic
(property~\eqref{item:BmonotonicityApp}).
Consider $p,q\in P$, with $\Lset(p)\subseteq\Lset(q)$.
We want to show that $\phi\in\B(p)$ implies $\phi\in\B(q)$, for each $\phi$.
Firstly, we observe that, by $\Lset(p)\subseteq\Lset(q)$ and
Theorem~\ref{theo:branching-semantics}, $p \<_{\twoSim} q$ holds. Thus, we have
that $p \equiv_{\simulation} q$, which implies $\B^-(p) = \B^-(q)$ (thanks to
the previous result); moreover, we have that for each $a \in Act$ and $p' \in P$
with $p \stackrel{a}{\rightarrow} p'$, there exists some $q' \in P$ such that $q
\stackrel{a}{\rightarrow} q'$ and $p' \<_{\twoSim} q'$.
In order to show that $\B(p) \subseteq \B(q)$, we proceed by induction on the
depth of $p$.
If $I(p) = \emptyset$, then $I(q) = \emptyset$ as well. Thus, we have that
$\B^+(p) = \B^+(q) = \{ \true \}$, and the thesis follows.
Otherwise ($I(p) \neq \emptyset$), let us consider a formula $\phi \in \B(p)$.
If $\phi \in \B^-(p)$, then the claim follows from $\B^-(p) = \B^-(q) \subseteq
\B(q)$.
If $\phi = \true$, then, by definition of $\B^+$, we have $\phi \in \B^+(q)
\subseteq \B(q)$.
Finally, if $\phi = \exist{a}\varphi \in \B^+(p)$, then, by definition of
$\B^+$, there exist $p' \in P$, with $p\stackrel{a}{\rightarrow}p'$, such that
$\varphi=\bigwedge_{\psi\in \Psi}\psi$ for some $\Psi \subseteq \boundfun(p')$.
This implies the existence of some $q' \in P$ such that $q
\stackrel{a}{\rightarrow} q'$ and $p' \<_{\twoSim} q'$ (and therefore $\Lset(p') \subseteq
\Lset(q')$ by Theorem~\ref{theo:branching-semantics}).
By the inductive hypothesis, $\boundfun(p') \subseteq \boundfun(q')$ holds as
well, which means that $\Psi \subseteq \boundfun(q')$. Hence, we have that
$\exist{a}\varphi \in \B^+(q) \subseteq \B(q)$.

\noindent{\bf Case bisimulation (\bisim).}
For the sake of clarity we recall from \tablename~\ref{tab:syntax_monotonicity_finite_char_branching_spectrum}
that $\B$ is defined as
$\boundfun^+(p)\cup \boundfun^-(p)$, where
\begin{compactitem}
\item $\boundfun^+(p)=\{\true\}\cup\{\exist{a}\varphi\mid a\in Act,
  \varphi=\bigwedge_{\psi\in\Psi}\psi, \Psi\subseteq\boundfun(p'),
  p\stackrel{a}{\rightarrow}p'\}$, and

\item $\B^-(p)= \{\univer{a}\varphi \in \Lset(p) \mid a \in Act, \varphi=
  \bigvee_{p\stackrel{a}{\rightarrow}p'} \bigwedge_{\psi \in \B(p')} \psi \}$.

\end{compactitem}

By definition, $\B^-(p) \subseteq \Lset(p)$ for every $p \in P$; moreover, it is
easy to verify that $\B^-(p)$ is finite, by induction on the depth on $p$ and by
recalling that $Act$ is finite and processes are finitely branching.
Hence, properties~\eqref{item:BsoundnessApp} and~\eqref{item:BfinitenessApp}
immediately follow from~\eqref{eq:BsoundnessApp-factorized}
and~\eqref{eq:BfinitenessApp-factorized}, respectively.
Notice also that property~\eqref{item:BfinitenessApp} implies that $\B$ is well
defined.

In order to prove property~\eqref{item:BentailmentApp}, we let $\phi \in
\Lset(p)$, for a generic $p \in P$, and we proceed, as usual, by induction on
the structure of $\phi$.
With respect to the previous semantics, bisimulation one is characterized by the
use of free negation (formula of the form $\neg \varphi$, for every $\varphi \in
\Lset$).
By using the fact that $\neg \exist{a}\varphi$ is equivalent to $\univer{a} \neg
\varphi$, it suffices to show how to deal with formulae in the form $\univer{a}
\varphi$, for every $\varphi \in \Lset$; the other cases are dealt with as in
the previous cases.

\begin{compactitem}

\item $\phi=\univer{a}\varphi$: by $\phi \in \Lset(p)$, we have that
  $\varphi\in\Lset(p')$ for all $p'$ such that $p\stackrel{a}{\rightarrow}p'$.
%
%
  By the inductive hypothesis, $\bigcap_{\psi\in\B(p')}\sat{\psi} \subseteq
  \sat{\varphi}$, for each $p'$ such that $p\stackrel{a}{\rightarrow}p'$, from
  which it follows $\bigcup_{p\stackrel{a}{\rightarrow}p'}\bigcap_{\psi \in \B(p')}
  \sat{\psi} \subseteq \sat{\varphi}$.
  We define $\zeta = \univer{a}\bigvee_{p \stackrel{a}{\rightarrow} p'}
  \bigwedge_{\psi \in \B(p')} \psi$ (notice that $\zeta = \univer{a}\false$ if
  $p\stackrel{a}{\not\rightarrow}$).
  Clearly, for all $p'$ such that $p\stackrel{a}{\not\rightarrow} p'$, it holds
  $p' \in \sat{\psi}$ for all $\psi \in \B(p')$.
  Therefore, $p \in \sat{\zeta}$, which means $\zeta \in \Lset(p)$.
  Moreover, we have that $\zeta \in \B^-(p)$ (by definition of $\B^-(p)$) and
  therefore $\bigcap_{\psi \in \B(p)} \sat{\psi} \subseteq \sat{\zeta}$.
  By $\bigcup_{p\stackrel{a}{\rightarrow}p'}\bigcap_{\psi \in \B(p')} \sat{\psi}
  \subseteq \sat{\varphi}$, we have $\sat{\zeta} \subseteq \sat{\univer{a}
    \varphi} = \sat{\phi}$, and the thesis immediately follows.

\end{compactitem}

Finally, we show that $\B$ is monotonic
(property~\eqref{item:BmonotonicityApp}).
Consider $p,q\in P$, with $\Lset(p)\subseteq\Lset(q)$.
We want to show that $\B(p) \subseteq \B(q)$.
Firstly, we observe that, by $\Lset(p)\subseteq\Lset(q)$ and
Theorem~\ref{theo:branching-semantics}, $p \< q$ holds. Moreover, by the
definition of bisimulation semantics, $p \< q$ implies $q \< p$, and thus $p
\equiv q$.
It is straightforward to show, by induction on the depth of $p$, that if $p
\equiv q$, then $\B(p) = \B(q)$, hence the thesis.
\end{proof}

%
%
%



\subsection{Existence of \texorpdfstring{$\bar \chi(\cdot)$}{bar chi(.)}} \label{sec:full-proof-existence-of-chi}


In what follows, we show that it is possible to build, for each
$\chi(p)$, a formula $\bar\chi(p)$, with the properties described in 
Corollary~\ref{cor:decomposability_condition_for_spectrum}(\ref{item:final_req_chibar}).

\setcounter{resume}{\value{lemma}}
\setcounter{lemma}{\value{lemanticharready}}
\begin{lemma}
  \lemanticharready
\end{lemma}
\setcounter{lemma}{\value{resume}}

\begin{proof}

  We consider each semantics $X \in \bspectrumSet$ in turn.
  As in the previous section, for the sake of a lighter notation, we often omit
  the subscripts identifying the semantics as it is clear from the context, e.g.,
  we write $\Lset$, $\bar\chi$, and $p\not \< q$ instead of $\Lset_{\simulation}$,
  $\bar\chi_{\simulation}$, and $p\not \<_{\simulation} q$, respectively.

\noindent{\bf Case simulation (\simulation).}
For the sake of clarity we recall here the definition of $\bar \chi$
from \tablename~\ref{tab:syntax_monotonicity_finite_char_branching_spectrum}:

{\centering
$\bar \chi(p)= \bigvee_{a\in Act} \exist{a}\bigwedge_{p\stackrel{a}{\rightarrow}p'} \bar\chi(p')$.

}

Let us first show that for every $p\in P$ we have $p \not \in \sat{\bar \chi(p)}$. We proceed by induction on the depth of $p$. 
\begin{compactitem}
\item $I(p) = \emptyset$: we have $\bar \chi(p)=\bigvee_{a\in Act} \exist{a}\true$, and thus $p\notin\sat{\bar \chi(p)}$.

\item $I(p) \neq \emptyset$: obviously, $p\notin\sat{\exist{a}\true}$ holds for every $a\notin I(p)$. Now, for every $p\stackrel{a}{\rightarrow}p'$, by induction hypothesis, $p'\notin \sat{\bar \chi(p')}$. Thus, $p'\notin \sat{\bigwedge_{p\stackrel{a}{\rightarrow}p'}\bar \chi(p')}$ and therefore  $p\notin \sat{\exist{a}\bigwedge_{p\stackrel{a}{\rightarrow}p'}\bar \chi(p')}$. Hence, we obtain that $p \not \in \sat{\bar \chi(p)}$.
\end{compactitem}

Now, let us show that $\{ p' \in P \mid p' \not \< p \} \subseteq \sat{\bar \chi(p)}$.
The proof is by induction on the depth of $p$.
\begin{compactitem}
\item $I(p) = \emptyset$: we have that $\{ p' \in P \mid p' \not \< p \}=P\setminus\{p'\in P\mid I(p')=\emptyset\}$. It is easy to see that $P\setminus\{p'\in P\mid I(p')=\emptyset\} = \sat{\bigvee_{a\in Act} \exist{a}\true}= \sat{\bar \chi(p)}$.

\item $I(p) \neq \emptyset$: let $q\not\< p$. Thus there exists some $q'$, with $q\stackrel{a}{\rightarrow}q'$, such that, for every $p'$, $p\stackrel{a}{\rightarrow}p'$ implies $q'\not\< p'$. Then, by inductive hypothesis, $q'\in\sat{\bar \chi(p')}$ for every $p'$ such that $p\stackrel{a}{\rightarrow}p'$. Thus, $q'\in\sat{\bigwedge_{p\stackrel{a}{\rightarrow}p'}\bar \chi(p')}$ and therefore $q\in\sat{\exist{a}\bigwedge_{p\stackrel{a}{\rightarrow}p'}\bar \chi(p')}$. Hence, we conclude $q\in \sat{\bar \chi(p)}$.
\end{compactitem}

 \medskip

\noindent{\bf Case complete simulation (\completeSim).}
For the sake of clarity we recall here the definition of $\bar \chi$
from \tablename~\ref{tab:syntax_monotonicity_finite_char_branching_spectrum}:

{\centering
$
\begin{array}{lll@{\hspace{20mm}}l}
\bar \chi(p) & = & \big( \bigvee_{a\in Act} \exist{a}\bigwedge_{p\stackrel{a}{\rightarrow}p'} \bar\chi(p') \big)\vee \zeroFormula & \text{if $I(p)\neq\emptyset$}\\

\bar \chi(p) & = & \bigvee_{a\in Act} \exist{a}\bigwedge_{p\stackrel{a}{\rightarrow}p'} \bar\chi(p') & \text{if $I(p)=\emptyset$}
\end{array}$

}

Let us first show that for every $p\in P$ we have $p \not \in \sat{\bar \chi(p)}$. We proceed by induction on the depth of $p$. 
\begin{compactitem}
\item $I(p) = \emptyset$: we have $\bar \chi(p)=\bigvee_{a\in Act} \exist{a}\true$, and obviously $p\notin\sat{\bar \chi(p)}$.

\item $I(p) \neq \emptyset$: obviously, $p\notin\sat{\zeroFormula}$ and $p\notin\sat{\exist{a}\true}$ for every $a\notin I(p)$. Now, for every $p\stackrel{a}{\rightarrow}p_1$, by the
inductive hypothesis, $p_1\notin \sat{\bar \chi(p_1)}$. Thus, $p_1\notin \sat{\bigwedge_{p\stackrel{a}{\rightarrow}p'}\bar \chi(p')}$ and 
therefore  $p\notin \sat{\exist{a}\bigwedge_{p\stackrel{a}{\rightarrow}p'}\bar \chi(p')}$. Hence, we obtain that $p \not \in \sat{\bar \chi(p)}$.
\end{compactitem}

Now, let us show that $\{ p' \in P \mid p' \not \< p \} \subseteq \sat{\bar \chi(p)}$. The proof is by induction on the depth of $p$.
\begin{compactitem}
\item $I(p) = \emptyset$: we have that $\{ p' \in P \mid p' \not \< p \}=P\setminus\{p'\in P\mid I(p')=\emptyset\}$ because, when $I(p) = \emptyset$, $q \not\< p$ holds exactly for processes $q$ with $I(q)\neq\emptyset$. It is easy to see that $P\setminus\{p'\in P\mid I(p')=\emptyset\} = \sat{\bigvee_{a\in Act} \exist{a}\true}= \sat{\bar \chi(p)}$.

\item $I(p) \neq \emptyset$: let $q\not\< p$. Thus, either $I(q)=\emptyset$ and $I(p)\neq\emptyset$ (or analogously $I(p)=\emptyset$ and $I(q)\neq\emptyset$), or there exists some $q'$, with $q\stackrel{a}{\rightarrow}q'$, such that, for every $p'$, $p\stackrel{a}{\rightarrow}p'$ implies $q'\not\< p'$. If it is the case that $I(q)=\emptyset$ and $I(p)\neq\emptyset$ (respectively, $I(p)=\emptyset$ and $I(q)\neq\emptyset$), then $q\in\sat{\zeroFormula}$ (respectively, $q\in\sat{\exist{a}\true}$ holds for some $a\notin I(p)$). In any case, $q\in \sat{\bar \chi(p)}$ holds.

Otherwise, if there exist $a\in Act$ and $q'\in P$, with $q\stackrel{a}{\rightarrow}q'$, such that $q'\not \< p'$ for every $p\stackrel{a}{\rightarrow}p'$, then, by the inductive hypothesis, $q'\in\sat{\bar \chi(p')}$ for every $p'$. Thus, $q'\in\sat{\bigwedge_{p\stackrel{a}{\rightarrow}p'}\bar \chi(p')}$ and therefore $q\in\sat{\exist{a}\bigwedge_{p\stackrel{a}{\rightarrow}p'}\bar \chi(p')}$. Hence, we conclude $q\in \sat{\bar \chi(p)}$.
\end{compactitem}

\medskip

\noindent{\bf Case ready simulation: (\readySim).}
See the proof of Lemma~\ref{lem:anti-char_branching} at
page~\pageref{lem:anti-char_branching}.
\\

\noindent{\bf Case trace simulation (\traceSim).}
For the sake of clarity we recall here the definition of $\bar \chi$
from \tablename~\ref{tab:syntax_monotonicity_finite_char_branching_spectrum}:

{\centering
$\bar\chi(p)= \big( \bigvee_{a\in act} \exist{a}\bigwedge_{p\stackrel{a}{\rightarrow}p'} \bar\chi(p') \big) \vee
     \bigvee_{\tau \in \traces(p), \tau a \notin \traces(p)}\exist{\tau a}\true \vee \bigvee_{p\stackrel{\tau a}{\rightarrow}p'}\univer{\tau a}\false$.

}

Let us first show that for every $p\in P$ we have $p \not \in \sat{\bar \chi(p)}$. We proceed by induction on the depth of $p$. 
\begin{compactitem}
\item $I(p) = \emptyset$: we have $\bar \chi(p)=\bigvee_{a\in Act} \exist{a}\true$, and thus $p\notin\sat{\bar \chi(p)}$.

\item $I(p) \neq \emptyset$: obviously, $p\notin\sat{\exist{\tau a}\true}$ holds for every $\tau\in \traces(p)$ and $a\in Act$ such that $\tau a \notin \traces(p)$. Moreover, $p\notin\sat{\univer{\tau a}\false}$ holds for every $\tau a\in \traces(p)$. Now, for every $p\stackrel{a}{\rightarrow}p_1$, by the inductive hypothesis, $p_1\notin \sat{\bar \chi(p_1)}$. Thus, $p_1\notin \sat{\bigwedge_{p\stackrel{a}{\rightarrow}p'}\bar \chi(p')}$ and therefore  
$p\notin \sat{\exist{a}\bigwedge_{p\stackrel{a}{\rightarrow}p'}\bar \chi(p')}$. Hence, we obtain that $p \not \in \sat{\bar \chi(p)}$.
\end{compactitem}

Now, let us show that $\{ p' \in P \mid p' \not \< p \} \subseteq \sat{\bar \chi(p)}$,
that is, $\sat{\bar \chi(p)}$ contains at least
the elements that are either strictly above $p$ or
incomparable with it.
The proof is by induction on the depth of $p$.
\begin{compactitem}
\item $I(p) = \emptyset$: we have that $\{ p' \in P \mid p' \not \< p \}=P\setminus\{p'\in P\mid I(p')=\emptyset\}$ because, when $I(p) = \emptyset$, $q \not\< p$ holds exactly for processes $q$ with $I(q)\neq\emptyset$. It is easy to see that $P\setminus\{p'\in P\mid I(p')=\emptyset\} = \sat{\bigvee_{a\in Act} \exist{a}\true}= \sat{\bar \chi(p)}$.

\item $I(p) \neq \emptyset$: let $q\not\< p$. Thus, either $\traces(q)\neq \traces(p)$ or there exists some $q'$, with $q\stackrel{a}{\rightarrow}q'$, such that, for every $p'$, $p\stackrel{a}{\rightarrow}p'$ implies $q'\not\< p'$. If it is the case that $\traces(q)\neq \traces(p)$, then either $q\in\sat{\exist{\tau a}\true}$ for some $\tau a$ such that $\tau \in \traces(p)$ and $\tau a\notin \traces(p)$, or $q\in\sat{\univer{\tau a}\false}$ for some $\tau a\in \traces(p)$. In either case, $q\in \sat{\bar \chi(p)}$ holds.

If, on the other hand, $\traces(p) = \traces(q)$, then there exist $a\in Act$ and $q'\in P$, with $q\stackrel{a}{\rightarrow}q'$, such that $q'\not \< p'$ for every $p\stackrel{a}{\rightarrow}p'$, then, by the inductive hypothesis, $q'\in\sat{\bar \chi(p')}$ for every $p'$. Thus, $q'\in\sat{\bigwedge_{p\stackrel{a}{\rightarrow}p'}\bar \chi(p')}$ and therefore $q\in\sat{\exist{a}\bigwedge_{p\stackrel{a}{\rightarrow}p'}\bar \chi(p')}$ (notice that $a \in I(p)$, due to $a \in I(q)$ and $\traces(p) = \traces(q)$). Hence, we conclude $q\in \sat{\bar \chi(p)}$.
\end{compactitem}

\medskip

\noindent{\bf Case 2-Nested simulation (\twoSim).}
For the sake of clarity we recall here the definition of $\bar \chi$
from \tablename~\ref{tab:syntax_monotonicity_finite_char_branching_spectrum}:

{\centering
$
\begin{array}{lll}
\bar \chi(p) & = & \big( \bigvee_{a\in Act} \exist{a}\bigwedge_{p\stackrel{a}{\rightarrow}p'} \bar\chi(p') \big) \vee \bar\Phi(p) \\

\bar \Phi(p) & = & \bigvee_{a\in I(p)}\univer{a}\false \vee \bigvee_{a\in I(p)} \bigvee_{p\stackrel{a}{\rightarrow}p'} \univer{a}\bar\Phi(p')
\end{array}$

}

Let us first show that for every $p\in P$ we have $p \not \in \sat{\bar \chi(p)}$. We proceed by induction on the depth of $p$. 
\begin{compactitem}
\item $I(p) = \emptyset$: we have $\bar \chi(p)=\bigvee_{a\in Act} \exist{a}\true$, and thus $p\notin\sat{\bar \chi(p)}$.

\item $I(p) \neq \emptyset$: obviously, $p\notin\sat{\exist{a}\true}$ holds for every $a\notin I(p)$. Next, for every $p\stackrel{a}{\rightarrow}p'$, by the inductive hypothesis, $p'\notin \sat{\bar \chi(p')}$. Thus, $p\notin \sat{\exist{a}\bar \chi(p')}$, for each $p'$ such that $p\stackrel{a}{\rightarrow}p'$. Therefore $p\notin \sat{\exist{a}\bigwedge_{p\stackrel{a}{\rightarrow}p'}\bar\chi(p')}$.
We are left with showing that $p \notin \sat{\bar\Phi(p)}$.
To this end, observe that $p\notin\sat{\univer{a}\false}$ for all $a\in I(p)$; moreover, for all $p\stackrel{a}{\rightarrow}p'$, since $p'\notin \sat{\bar \chi(p')}$ holds by the inductive hypothesis, in particular we have $p'\notin \sat{\bar \Phi(p')}$. Thus, $p\notin \sat{\univer{a}\bar \Phi(p')}$ for each $p'$ such that $p\stackrel{a}{\rightarrow}p'$ and $p\notin \sat{\bigvee_{p\stackrel{a}{\rightarrow}p'}\univer{a}\bar\Phi(p')}$. Hence, we obtain that $p \not \in \sat{\bar \chi(p)}$.
\end{compactitem}

Now, let us show that $\{ p' \in P \mid p' \not \< p \} \subseteq \sat{\bar \chi(p)}$.
The proof is by induction on the depth of $p$.
\begin{compactitem}
\item $I(p) = \emptyset$: we have that $\{ p' \in P \mid p' \not \< p \}=P\setminus\{p'\in P\mid I(p')=\emptyset\}$. It is easy to see that $P\setminus\{p'\in P\mid I(p')=\emptyset\} = \sat{\bigvee_{a\in Act} \exist\true}= \sat{\bar \chi(p)}$.

\item $I(p) \neq \emptyset$: let $q\not\<_{\twoSim} p$. Thus, by definition, either there exists $a\in Act$ and $q'$, with $q\stackrel{a}{\rightarrow}q'$, such that, for every $p'$, $p\stackrel{a}{\rightarrow}p'$ implies $q'\not\<_{\twoSim} p'$, or $p\not\<_\simulation q$.

First, if there exist $a\in Act$ and $q'\in P$, with $q\stackrel{a}{\rightarrow}q'$, such that $q'\not \<_{\twoSim} p'$ for every $p\stackrel{a}{\rightarrow}p'$, by the inductive hypothesis, $q'\in\sat{\bar \chi(p')}$ for every $p'$ such that $p\stackrel{a}{\rightarrow}p'$. Thus, $q'\in\sat{\bigwedge_{p\stackrel{a}{\rightarrow}p'}\bar \chi(p')}$ and therefore $q\in\sat{\exist{a}\bigwedge_{p\stackrel{a}{\rightarrow}p'}\bar \chi(p')}$.

Otherwise, if $p\not\<_\simulation q$, then there exist $a\in Act$ and $p'$ with $p\stackrel{a}{\rightarrow}p'$, such that $p'\not \<_\simulation q'$ for every $q\stackrel{a}{\rightarrow}q'$. We show by induction on the depth of $q$ that if $p\not\<_\simulation q$ then $q\in\bar\Phi(p)$. If $I(q)=\emptyset$, then $q\in\sat{\univer{a}\false}$ for some $a\in I(p)$ (notice that $I(p) \neq \emptyset$ by assumption) and thus $q\in\bar\Phi(p)$. Now, when $I(q)\neq\emptyset$, by the inductive hypothesis, $q'\in\sat{\bar\Phi(p')}$ for every $q\stackrel{a}{\rightarrow}q'$. Thus $q\in\sat{\univer{a}\bar\Phi(p')}$, which implies that $q\in\sat{\bar\Phi(p)}$.
Since $p\not\<_\simulation q$ by assumption, we have that
$q\in\sat{\bar\Phi(p)}$. Therefore $q\in\sat{\bar\chi(p)}$ holds by
definition of $\bar\chi(p)$, and we are done.

\medskip

\noindent{\bf Case bisimulation (\bisim).}
We simply define $\bar\chi(p)=\neg\chi(p)$ ($\chi(p)$ exists thanks to
Proposition~\ref{prop:characteristic}) and $\sat{\bar\chi(p)}=P \setminus
\sat{\chi(p)}$ holds trivially.
\qedhere
\end{compactitem}
\end{proof}


\section{Proofs for semantics in van Glabbeek's linear time spectrum}
\label{app:linear}

\begin{lemma}\label{lem:sim_depth}
For all $p, q\in P$ such that $p\<_\trace q$, we have $\depth(p)\leq\depth(q)$.
\end{lemma}

\begin{proof}
  The claim immediately follows from the definition of $\depth(p)$ (i.e.,
  $\depth(p)$ is the length of a longest trace in $\trace(p)$) and the fact that
  $p\<_\trace q$ if and only if $\trace(p) \subseteq \trace(q)$.
\end{proof}

\begin{corollary}\label{cor:spectrum_depth}
Let $X \in \lspectrumSet$. For all $p, q\in P$ such that $p\<_X q$, we have $\depth(p)\leq\depth(q)$.
\end{corollary}

\begin{lemma}\label{lem:soundness-of-B-for-linear}
  $\B_{\genericSemanticsX}(p)\subseteq\Lset_{\genericSemanticsX}(p)$ holds for
  every $\genericSemanticsX \in \lspectrumSet$ and $p \in P$.
\end{lemma}
\begin{proof}
  Let $\genericSemanticsX \in \lspectrumSet$, $p \in P$, and $\phi \in
  \B_{\genericSemanticsX}(p)$.
  If $\phi = \true$, then the thesis follows trivially.
  If $\phi = \getFormula{\genericSemanticsX}{x}$ for some $x \in
  \genericSemanticsFiniteX(p)$, then we distinguish the following cases.

  \begin{itemize}
  \item If $\genericSemanticsX = \trace$, then $\phi = \exist{\tau}\true$ for
    $\tau \in \traceFinite(p) \subseteq \trace(p)$.
    Since $\tau \in \trace(p)$, it clearly holds $p \in \sat{\phi}$, which means
    $\phi \in \Lset_{\trace}(p)$.

  \item If $\genericSemanticsX = \completeTrace$, then $\phi =
    \exist{\tau}\zeroFormula$ for $\tau \in \completeTraceFinite(p) \subseteq
    \completeTrace(p)$.
    Since $\tau \in \completeTrace(p)$, it clearly holds $p \in \sat{\phi}$,
    which means $\phi \in \Lset_{\completeTrace}(p)$.

  \item If $\genericSemanticsX = \failure$, then $\phi = \exist{\tau}
    \bigwedge_{a \in Y} \univer{a}\false$ for $\langle \tau, Y \rangle \in
    \failureFinite(p) \subseteq \failure(p)$.
    From the definition of $\failure(p)$, it immediately follows $p \in
    \sat{\phi}$, and we are done.

  \item If $\genericSemanticsX = \ready$, then $\phi = \exist{\tau}
    \existuniver{Y}$ for $\langle \tau, Y \rangle \in \readyFinite(p) \subseteq
    \ready(p)$.
    From the definition of $\ready(p)$, it immediately follows $p \in
    \sat{\phi}$, and we are done.

  \item If $\genericSemanticsX = \failureTrace$, then let $\phi =
    \getFormula{\failureTrace}{\sigma}$ for $\sigma \in \failureTraceFinite(p)$; we
    proceed by induction on $\sigma$.
    \begin{itemize}
    \item If $\sigma = \varepsilon$, then $\phi = \true$, and the thesis
      follows trivially;
    \item If $\sigma = a\sigma'$, with $a \in Act$ and $\sigma' \in
      \failureTraceFinite(p')$ for some $p'$ such that $p \stackrel{a}{\rightarrow}
      p'$, then $\phi = \exist{a}\getFormula{\failureTrace}{\sigma'}$ and
      $\getFormula{\failureTrace}{\sigma'} \in \B_{\failureTrace}(p')$\footnote{Notice
        that, while it is clear that $\sigma' \in \failureTrace(p')$, we are not
        guaranteed in general that $|\sigma'| \leq 2 \cdot \depth(p')$ holds (which is
        needed for $\sigma'$ to belong to $\failureTraceFinite(p')$).
        However, it can be shown that we can always reduce to the case where
        this holds (see item~\ref{item:linear_word_bound} at
        page~\pageref{item:linear_word_bound} for a detailed explanation).}.
      By the inductive hypothesis, we have that
      $\getFormula{\failureTrace}{\sigma'} \in \Lset_{\failureTrace}(p')$, that is,
      $p' \in \sat{\getFormula{\failureTrace}{\sigma'}}$.
      It follows that $p \in \sat{\phi}$, which equals to $\phi \in
      \Lset_{\failureTrace}(p)$.
    \item If $\sigma = Y\sigma'$, with $Y \in \powerset{Act}$ and $\sigma' \in
      \failureTraceFinite(p)$, then $\phi = \bigwedge_{a \in Y}{\univer{a}\false}
      \wedge \getFormulafailureTrace{\sigma'}$.
      By $Y\sigma' \in \failureTrace(p')$, we have that $p \in \sat{\bigwedge_{a
          \in Y}{\univer{a}\false}}$.
      By $\sigma' \in \failureTraceFinite(p)$, it holds
      $\getFormula{\failureTrace}{\sigma'} \in \B_{\failureTrace}(p)$ and thus, by the
      inductive hypothesis, $\getFormula{\failureTrace}{\sigma'} \in
      \Lset_{\failureTrace}(p)$, that is, $p \in
      \sat{\getFormula{\failureTrace}{\sigma'}}$.
      It follows that $p \in \sat{\phi}$, which equals to $\phi \in
      \Lset_{\failureTrace}(p)$.

    \end{itemize}

  \item If $\genericSemanticsX \in \{ \readyTrace, \impossibleFutureTrace,
    \possibleFutureTrace, \impossibleSimulationTrace, \possibleSimulationTrace \}$,
    then the proof uses the same inductive argument as in the previous case, and
    thus we omit the details.

  \item If $\genericSemanticsX = \impossibleFuture$, then $\phi = \exist{\tau}
    \bigwedge_{\tau' \in \Gamma} \univer{\tau'}\false$ for $\langle \tau, \Gamma
    \rangle \in \impossibleFutureFinite(p) \subseteq \impossibleFuture(p)$.
    From the definition of $\impossibleFuture(p)$, it immediately follows that
    $p \in \sat{\phi}$, and we are done.

  \item If $\genericSemanticsX = \possibleFuture$, then $\phi = \exist{\tau}
    \existuniver{\Gamma}$ for $\langle \tau, \Gamma
    \rangle \in \possibleFutureFinite(p) \subseteq \possibleFuture(p)$.
    From the definition of $\possibleFuture(p)$, it immediately follows that $p
    \in \sat{\phi}$, and we are done.

  \item If $\genericSemanticsX = \impossibleSimulation$, then $\phi =
    \exist{\tau} \bigwedge_{\eqClassbisim{p'} \in \mathbb P} \neg \chiS(p')$ for
    $\langle \tau, \mathbb P \rangle \in \impossibleSimulationFinite(p) \subseteq
    \impossibleSimulation(p)$.
    By the definition of $\impossibleSimulation(p)$, $\langle \tau, \mathbb P
    \rangle \in \impossibleSimulation(p)$ implies that there is some $q$ such that
    $p \stackrel{\tau}{\rightarrow} q$ and $\Sequivclass{q} \cap \mathbb P =
    \emptyset$.
    From $\Sequivclass{q} \cap \mathbb P = \emptyset$, it follows that $p'
    \not\<_{\simulation} q$ for all $\eqClassbisim{p'} \in \mathbb P$, which means
    that $q \notin \sat{\chiS(p')}$ for all $\eqClassbisim{p'} \in \mathbb P$.
    Therefore, $p \in \sat{\phi}$, and we are done.

  \item If $\genericSemanticsX = \possibleSimulation$, then $\phi = \exist{\tau}
    \alpha(\mathbb P )$, with

    {\centering

      $\alpha (\mathbb P) = \bigwedge_{\eqClassbisim{p'} \in \mathbb P}
      \chiS(p') \wedge \bigvee_{\eqClassbisim{p'} \in \mathbb P} \simulatedby{p'}$

    }

    and

    {\centering

      $\simulatedby{p'} = \bigwedge_{a \in Act}\univer{a} \bigvee_{p'
        \stackrel{a}{\rightarrow} p''}\simulatedby{p''}$.

    }

    By the definition of $\possibleSimulation(p)$, $\langle \tau, \mathbb P
    \rangle \in \possibleSimulation(p)$ implies that there is some $q$ such that $p
    \stackrel{\tau}{\rightarrow} q$ and $\Sequivclass{q} = \mathbb P$.
    It follows that $p' \<_{\simulation} q$ for all $\eqClassbisim{p'} \in
    \mathbb P$, which means that $q \in \sat{\chiS(p')}$ for all $\eqClassbisim{p'}
    \in \mathbb P$.
    Moreover, by $\Sequivclass{q} = \mathbb P$, we have that $\eqClassbisim{q}
    \in \mathbb P$.
    Since $q \in \simulatedby{q}$ clearly holds, we conclude $p \in \sat{\phi}$,
    and we are done.
\qedhere
  \end{itemize}
\end{proof}

\begin{lemma}\label{lem:linear-B-is-characteristic}
  Formula $\bigwedge_{\psi \in \B_{\genericSemanticsX}(p)} \psi$ is
  characteristic for $p$ within $\Lset_{\genericSemanticsX}$, for all
  $\genericSemanticsX \in \lspectrumSet$ and $p \in P$.
\end{lemma}
\begin{proof}
  According to
  Proposition~\ref{prop:general_properties_of_formulae}(\ref{item:char1}),
  Theorem~\ref{theo:linear-semantics}, and Definition~\ref{def:linear-semantics},
  it suffices to show, for all $\genericSemanticsX \in \lspectrumSet$ and $p \in
  P$, that $\sat{\bigwedge_{\psi \in \B_{\genericSemanticsX}(p)} \psi} = \{ p' \in
  P \mid \genericSemanticsX(p)\subseteq\genericSemanticsX(p') \}$.

  In order to show the inclusion $\{ p' \in P \mid
  \genericSemanticsX(p)\subseteq\genericSemanticsX(p') \} \subseteq
  \sat{\bigwedge_{\psi \in \B_{\genericSemanticsX}(p)} \psi}$, let $p' \in P$ be
  such that $\genericSemanticsX(p)\subseteq\genericSemanticsX(p')$.
  By Proposition~\ref{prop:X-and-Xfinite} (\ref{item:Xp-subset-Xp} $\Rightarrow$
  \ref{item:Xfinitep-subset-Xfinitep}), we have
  $\genericSemanticsFiniteX(p)\subseteq\genericSemanticsFiniteX(p')$.
  We show that $p' \in \sat{\psi}$ for all $\psi \in
  \B_{\genericSemanticsX}(p)$.
  To this end, consider a generic element $\psi$ of
  $\B_{\genericSemanticsX}(p)$, which means that $\psi =
  \getFormula{\genericSemanticsX}{x}$ for some $x \in \genericSemanticsFiniteX(p)
  \subseteq \genericSemanticsX(p) \subseteq \genericSemanticsX(p')$.
  By Proposition~\ref{prop:p-sat-formula-x-for-x-in-Xp} (left-to-right
  direction), we have $p' \in \sat{\getFormula{\genericSemanticsX}{x}} =
  \sat{\psi}$.
  Since $\psi$ is a generic element of $\B_{\genericSemanticsX}(p)$, the claim
  follows.

  Let us turn now to the converse inclusion.
  Let $p' \in P$ be such that $p' \in \sat{\psi}$ for all $\psi \in
  \B_{\genericSemanticsX}(p)$, which means $p' \in
  \sat{\getFormula{\genericSemanticsX}{x}}$ for all $x \in
  \genericSemanticsFiniteX(p)$.
  By Proposition~\ref{prop:p-sat-formula-x-for-x-in-Xp} (right-to-left
  direction), it holds that $x \in \genericSemanticsX(p')$ for all $x \in
  \genericSemanticsFiniteX(p)$, that is, $\genericSemanticsFiniteX(p) \subseteq
  \genericSemanticsX(p')$.
  By Proposition~\ref{prop:X-and-Xfinite} (\ref{item:Xfinitep-subset-Xp}
  $\Rightarrow$ \ref{item:Xp-subset-Xp}), we have $\genericSemanticsX(p) \subseteq
  \genericSemanticsX(p')$, and we are done.
%
\end{proof}


\end{document}